\documentclass[reqno, 11pt]{amsart}

\usepackage{amsfonts,amssymb,amsbsy,amsmath,amsthm,enumerate}
\usepackage{txfonts}
\usepackage{eucal}
\usepackage{dsfont}
\usepackage{mathrsfs}
\usepackage{nicefrac}
\def\nicefrac#1#2{\frac{#1}{#2}}

\usepackage{stackrel}

\usepackage[small,nohug,
]{diagrams}
\diagramstyle[labelstyle=\scriptstyle]
\newarrow{Corresponds}<--->

\usepackage[dvips]{graphicx}
\usepackage{color}

\newtheorem{theorem}{Theorem}[section]
\newtheorem{proposition}[theorem]{Proposition}
\newtheorem{lemma}[theorem]{Lemma}
\newtheorem{corollary}[theorem]{Corollary}
\newtheorem{assumption}[theorem]{Assumption}
\newtheorem{definition}[theorem]{Definition}
\newtheorem{remark}[theorem]{Remark}

\numberwithin{equation}{section}

\numberwithin{figure}{section}
\numberwithin{table}{section}

\newcommand\beq{\begin{equation}}
\newcommand{\bea}{\begin{eqnarray}}
\newcommand{\eea}{\end{eqnarray}}
\newcommand{\beas}{\begin{eqnarray*}}
\newcommand{\eeas}{\end{eqnarray*}}
\newcommand{\beql}{\begin{equation} \label}
\newcommand{\eeq}{\end{equation}}

\newcommand{\R}{\mathbb R}
\newcommand{\N}{\mathbb N}
\newcommand{\C}{\mathbb C}                           

\newcommand{\Z}{\mathbb Z}
\newcommand{\T}{\mathbb T}

\newcommand{\empt}{{\text{\rm  \O}}}

\newcommand{\s}[1]{\CMcal{#1}}
\newcommand{\f}[1]{\mathcal{#1}}                  
\newcommand{\bb}[1]{\mathscr{#1}}
\newcommand{\rr}[1]{\mathfrak{#1}}
\newcommand{\n}[1]{\mathds {#1}}

\newcommand{\expo}[1]{{\rm e}^{#1}}                 

\newcommand{\ii}{\,{\rm i}\,}
\newcommand{\ncint}{\mathrel{{\ooalign{$\int$\cr\kern+.07em\raise.15ex\hbox{$\pmb{\scriptstyle-}$}\cr}}}}           \newcommand{\ncpartial}{\mathrel{{\ooalign{$\partial$\cr\kern+.29em\raise.79ex\hbox{$\pmb{\scriptstyle-}$}\cr}}}}

\newcommand{\virg}[1]{\lq\lq#1\rq\rq}                
\newcommand{\ie}{{\sl i.\,e.\,}}
\newcommand{\eg}{{\sl e.\,g.\,}}
\newcommand{\cf}{{\sl cf.\,}}

\begin{document}

\title[The FKMM-invariant in low dimension]{
The FKMM-invariant in low dimension}


\author[G. De~Nittis]{Giuseppe De Nittis}
\address[De~Nittis]{Facultad de Matem\'aticas \& Instituto de F\'{\i}sica,
Pontificia Universidad Cat\'olica,
Santiago, Chile}
\email{gidenittis@mat.puc.cl}

\author[K. Gomi]{Kiyonori Gomi}
\address[Gomi]{Department of Mathematical Sciences, Shinshu University,  Nagano, Japan}
\email{kgomi@math.shinshu-u.ac.jp}




\begin{abstract}
\vspace{-4mm}
In this paper we investigate the problem of the cohomological classification of \virg{Quaternionic} vector bundles in low-dimension ($d\leqslant 3$). We show that there exists a characteristic classes $\kappa$, called the FKMM-invariant, which takes value in the relative equivariant Borel cohomology and  completely classifies \virg{Quaternionic} vector bundles in low-dimension. The main subject of the paper concerns a discussion about  the surjectivity of  $\kappa$. 
\end{abstract}


\maketitle

\vspace{-5mm}
\tableofcontents

\section{Introduction}\label{sect:intro}

At a topological level, \emph{\virg{Real}} and \emph{\virg{Quaternionic}} vector bundles are complex vector bundles defined over spaces with \emph{involution} and endowed with a  homeomorphism of the total space which cover the involution and acts anti-linearly between conjugate fibers. Let us be a little bit more precise: 
 An {involution} $\tau$ on a topological space $X$ is a homeomorphism $\tau:X\to X$ of period 2, \ie  $\tau^2={\rm Id}_X$. We will refer to the pair 
 $(X,\tau)$ as an  \emph{involutive space}. The \emph{fixed-point} set of the {involutive space}  $(X,\tau)$
 is  by definition
 $$
X^{\tau}\,:=\, \{x\in X\ |\ \tau(x)=x\}\,.
$$ 
We are interested in (topological) complex vector bundles $\bb{E}\to X$ endowed with an extra  homeomorphism $\Theta:\bb{E}\to\bb{E}$ such that $\Theta:\bb{E}|_x\to \bb{E}|_{\tau(x)}$ acts anti-linearly between the fibers over conjugated points $x$ and $\tau(x)$. The pair $(\bb{E},\Theta)$ defines a \virg{Real} vector bundle (also called $\rr{R}$-bundle) over the involutive space $(X,\tau)$ if the square $\Theta^2:=\Theta\circ \Theta$ acts as the identity in each fiber of $\bb{E}$, or equivalently  if $\Theta^2=+1$ (with a small abuse of notation). In the opposite case when $\Theta^2=-1$ one says that $(\bb{E},\Theta)$ is a \virg{Quaternionic} vector bundle (or $\rr{Q}$-bundle) over  $(X,\tau)$. {\virg{Real} vector bundles were first defined by M.~F. Atiyah in \cite{atiyah-66} while the notion of \virg{Quaternionic} vector bundle has been introduced by  J.~L. Dupont in \cite{dupont-69} (under the name of symplectic vector bundle).
{\virg{Real} and {\virg{Quaternionic}} vector bundles provide \emph{new categories} of topological objects which are significantly different from the categories of complex or real vector bundles. For these reasons the problem of the classification of {\virg{Real} or {\virg{Quaternionic}} vector bundles over a given involutive space requires the use of appropriate tools which can differ, even significantly, from the tools usually used to classify  vector bundles in the complex or real category.
This paper mostly concerns with the problem of the classification of $\rr{Q}$-bundles over \emph{low-dimensional} $\Z_2$-CW-complexes.
Henceforth, we will assume that: 
\begin{assumption}[$\Z_2$-CW-complex]\label{ass:top}
$(X,\tau)$ is an involutive space such that $X$ is a (locally) compact and path-connected Hausdorff space which admits the structure of a $\Z_2$-CW-complex compatible with the involution $\tau$. The \emph{dimension} $d$ of $X$ is, by definition, the maximal dimension of its cells 
and we say that $X$ is \emph{low-dimensional} if $0\leqslant d\leqslant 3$.
\end{assumption}
\noindent
For the sake of completeness let us recall that an involutive space $(X,\tau)$ has the structure of a $\Z_2$-CW-complex if it admits a skeleton decomposition given by gluing cells of different dimensions which carry a $\Z_2$-action. 
For a precise definition of the notion of 
$\Z_2$-CW-complex, the reader can refer to  \cite{matumoto-71,allday-puppe-93} (see also \cite[Section 4.5]{denittis-gomi-14}).  

\medskip

It is well known that equivalence classes of rank $m$ complex  (resp. real) vector bundles over a topological space $X$ can be classified by means of homotopy classes of maps from $X$ into the Grassmannian
$G_m(\C^\infty)$ (resp. $G_m(\R^\infty)$). 
This can be summarized by
\begin{equation}\label{eq:intro_claaC&R}
{\rm Vec}_{\n{F}}^{m}\big(X\big)\;{\simeq}\;\big[X,G_m(\n{F}^\infty)\big]\;,\qquad\quad\ \n{F}=\R,\C\;
\end{equation}
where on the left-hand side one has the set of equivalence classes of $\n{F}$-vector bundles and on the right-hand side one has the set of homotopy classes of maps from $X$ into the related Grassmannian.
Unfortunately, the computation of homotopy classes of maps between  topological spaces is an extremely hard problem. For this reason people started to face the problem of the classification of vector bundles by using the \emph{characteristic classes} \cite{milnor-stasheff-74}. Due to the fact that a rank $m$ complex vector bundle is associated, up to isomorphisms, with  a map $\varphi:X\to G_m(\C^\infty)$ one can pull-back  via $\varphi$ the cohomology ring of the Grassmannian
into that of $X$, \ie $\varphi^*\:H^\bullet(G_m(\C^\infty),\Z)\to H^\bullet(X,\Z)$.
This procedure selects a set of cohomology classes $c_j(\bb{E})\in H^{2j}(X,\Z)$
called the \emph{Chern classes} of $\bb{E}$. A similar procedure for real vector bundles leads to the definition of the \emph{Stiefel-Whitney classes} $w_j(\bb{E})\in H^{j}(X,\Z_2)$. The relevant role of the characteristic classes in the classification of  complex and real vector bundles is a  well-established fact. For instance  F. P. Peterson proved in 1959 that the Chern classes completely classify complex vector bundles over CW-complexes in the stable rank regime
\cite{peterson-59}. In particular, for low-dimensional CW-complexes the result by F. P. Peterson says that the \emph{first} Chern class $c_1$ suffices to classify complex vector bundles independently of the rank of the fibers:
\begin{equation}\label{eq:intro_tqs3}
c_1\;:\;{\rm Vec}_{\C}^{m}(X)\;\stackrel{\simeq}{\longrightarrow}\;H^2\big(X,\Z\big)\;,\qquad\quad\forall\ m\in\N \quad\text{if} \quad {\rm dim}(X)\leqslant 3\;.
\end{equation}

\medskip

Also $\rr{R}$ and $\rr{Q}$-bundles over $(X,\tau)$ can be classified by means of 
homotopy classes of maps in the same spirit of  \eqref{eq:intro_claaC&R}.
For $\rr{R}$-bundles the homotopy classification theorem is due to 
A.~L. Edelson  \cite{edelson-71} (see also \cite[Theorem 4.13]{denittis-gomi-14}) and says that ${\rm Vec}_{\rr{R}}^{m}(X,\tau)$ is classified, up to $\Z_2$-homotopy equivalences, by \emph{$\Z_2$-equivariant} maps between $(X,\tau)$ and the Grassmannian $G_m(\C^\infty)$ endowed with the involution induced by the complex conjugation. 
In much the same way
for the \virg{Quaternionic} case one has that ${\rm Vec}_{\rr{Q}}^{2m}(X,\tau)$ is classified by {$\Z_2$-equivariant} maps from  $(X,\tau)$ into the Grassmannian $G_{2m}(\C^\infty)$ endowed with the \emph{quaternionic involution} (see \cite[Theorem 2.4]{denittis-gomi-14-gen} and related details). However, as in the real or complex case, the use of the homotopy classification theorem to compute ${\rm Vec}_{\rr{R}}^{m}(X,\tau)$ or ${\rm Vec}_{\rr{Q}}^{2m}(X,\tau)$  is not practicable due to insurmountable complications in the calculation of the equivariant homotopy classes. For these reasons one is naturally pushed to develop a theory of  characteristic classes  for $\rr{R}$ and $\rr{Q}$-bundles. In the case of the \virg{Real} category the associated theory of characteristic classes has been introduced by B. Kahn in 1959 \cite{kahn-59}. He defined the notion of \emph{\virg{Real} Chern classes} $c_j^{\rr{R}}(\bb{E},\Theta)\in H^{2j}(X,\Z(1))$ of the  $\rr{R}$-bundle $(\bb{E},\Theta)$ over $(X,\tau)$
 just by mimicking the standard procedure used for the construction of the usual Chern classes. The main difference with the complex case is the election of the cohomology theory that in the \virg{Real} case turns out to be the  \emph{equivariant Borel cohomology} of the involutive space $(X,\tau)$
with local coefficient system $\Z(1)$  (see Appendix \ref{subsec:borel_cohom} and references therein for more details). The \virg{Real} Chern classes provide a powerful tool for the classification of $\rr{R}$-bundles over $\Z_2$-CW-complexes.
First of all the \emph{Kahn's isomorphism} says that the \emph{\virg{Real} Picard group} ${\rm Pic}_{\rr{R}}(X,\tau):={\rm Vec}_{\rr{R}}^1(X,\tau)$ is classified exactly by $H^2_{\Z_2}\big(X,\Z(1)\big)$ 
 \cite[Proposition 1]{kahn-59}
(see also  \cite[Corollary A.5]{gomi-13}). Moreover, if   one combines the Kahn's isomorphism with the splitting  due to the stable condition  \cite[Theorem 4.25]{denittis-gomi-14}
one obtains that
\begin{equation}\label{eq:intro_tqs4}
c_1^{\rr{R}}\;:\;{\rm Vec}_{\rr{R}}^{m}(X,\tau)\;\stackrel{\simeq}{\longrightarrow}\;H^2_{\Z_2}\big(X,\Z(1)\big)\;,\qquad\quad\forall\ m\in\N \quad\text{if} \quad {\rm dim}(X)\leqslant 3\;,
\end{equation}
namely the first \virg{Real} Chern class $c_1^{\rr{R}}$ turns out to be a complete invariant for the classification of $\rr{R}$-bundles in low-dimension.
It is interesting to  note the close similarity between the equations \eqref{eq:intro_tqs3} and \eqref{eq:intro_tqs4}.

\medskip

\emph{And what about the theory of characteristic classes in the \virg{Quaternionic} category?} At the best of our knowledge, a completely satisfactory theory of characteristic classes for $\rr{Q}$-bundles is not yet available in the literature. Of course already exist attempts to  define \virg{Quaternionic} Chern classes, see \eg \cite{dos-santos-lima-filho-04,lawson-lima-filho-michelsohn-05}.  
However, these topological objects don't seem to be flexible enough to re-produce in an \virg{easy way}  classification results of the type \eqref{eq:intro_tqs3} or \eqref{eq:intro_tqs4}.
A different approach was introduced first by M. Furuta,  Y. Kametani, H. Matsue, and N. Minami in \cite{furuta-kametani-matsue-minami-00} and then improved in \cite{denittis-gomi-14-gen} and
generalized in \cite{denittis-gomi-16}. The core of the construction is to properly define a map 
\begin{equation}\label{eq:intro_tqs5}
\kappa\;:\;{\rm Vec}_{\rr{Q}}^{2m}\big(X,\tau\big)\;\stackrel{}{\longrightarrow}\;H^2_{\Z_2}\big(X|X^\tau,\Z(1)\big)\;,
\end{equation}
called the \emph{FKMM-invariant}, which associates to each  $\rr{Q}$-bundle $(\bb{E},\Theta)$ over $(X,\tau)$ a cohomology class $\kappa(\bb{E},\Theta)$ in the \emph{relative} equivariant cohomology group $H^2_{\Z_2}(X|X^\tau,\Z(1))$ (\cf Appendix \ref{subsec:borel_cohom} for more details). The definition of the map $\kappa$ is based on two important observations: (1) each cohomology class in $H^2_{\Z_2}(X|X^\tau,\Z(1))$ can be uniquely associated with a pair $(\bb{L},s)$ given by an $\rr{R}$-line bundle along with a trivializing section $s:X^\tau\to\bb{L}$; (2) each even-rank $\rr{Q}$-bundle canonically defines a pair $(\bb{L},s)$ through the \emph{determinant functor}. The details of the construction of $\kappa$ are summarized in Section \ref{subsec:fkmm-invariant}.

\medskip

In \cite{denittis-gomi-16} we proved that $\kappa$ is a characteristic class in the sense that it can be obtained as the pull-back of a universal class. Moreover, when the base space has dimension $d\leqslant 3$ the map $\kappa$ turns out to be \emph{injective}. These two properties suggest a parallel between the {FKMM-invariant} \eqref{eq:intro_tqs5} and the first Chern class as in \eqref{eq:intro_tqs3} or  the first \virg{Real} Chern class as in \eqref{eq:intro_tqs4}. However, the injectivity of $\kappa$, although is an important property, is not sufficient to reduce the classification of ${\rm Vec}_{\rr{Q}}^{2m}(X,\tau)$  to the mere computation of a  cohomology group at least until one can say something about the \emph{surjectivity} of $\kappa$. The analysis of the surjectivity of the {FKMM-invariant}   in low-dimension is the main achievement of this paper.

\medskip

Before stating our main results let us point out a quite obvious fact: The FKMM-invariant $\kappa$, like $c_1$ in \eqref{eq:intro_tqs3} or $c_1^{\rr{R}}$ \eqref{eq:intro_tqs4}, generally fails to be surjective in dimension $d\geqslant 4$ where other invariants (\eg the second Chern class) start to enter in the problem of the classification (see \eg \cite[Theorem 1.4 \& Theorem 5.2]{denittis-gomi-14-gen}).
For this reason a classification theory only based on $\kappa$ can be \emph{effective} only in low-dimension $d\leqslant3$.

\medskip

Let us start with the analysis of the case of an involutive space $(X,\tau)$ with a \emph{free involution}, namely $X^\tau=\text{\rm \O}$. This is the only situation which allows the existence of odd-rank $\rr{Q}$-bundles. The classification of odd-rank $\rr{Q}$-bundles in low-dimension is quite simple. First of all one notices that the set of $\rr{Q}$-line bundles ${\rm Pic}_{\rr{Q}}(X,\tau):={\rm Vec}_{\rr{Q}}^1(X,\tau)$ is not a group as in the \virg{Real} case but only a \emph{torsor} under the action of ${\rm Pic}_{\rr{R}}(X,\tau)$. This fact leads to the following dichotomy: or ${\rm Pic}_{\rr{Q}}(X,\tau)=\text{\rm \O}$ or 
${\rm Pic}_{\rr{Q}}(X,\tau)\simeq {\rm Pic}_{\rr{R}}(X,\tau)$ (\cf Theorem \ref{teo:torsor_pica}). When one combines this result with the
stable condition for odd-rank $\rr{Q}$-bundles in Proposition \ref{theo:stab_ran_Q_odd} one obtains that the topological classification of the odd-rank $\rr{Q}$-bundles in the regime of low-dimension is completely specified by ${\rm Pic}_{\rr{Q}}(X,\tau)$:
\begin{theorem}[Classification of odd-rank $\rr{Q}$-bundles in low-dimension]
\label{theo:class_low_odd_R}
Let $(X,\tau)$ be a low-dimensional involutive space in the sense of Assumption \ref{ass:top}. Assume that $X^\tau=\text{\emph{\O}}$. Then
$$
{\rm Vec}^{2m-1}_{\rr{Q}}\big(X, \tau\big)\;\simeq\; 
\left\{
\begin{aligned}
&\text{\emph{\O}}&\ \ \text{if}\ \ &{\rm Pic}_{\rr{Q}}(X,\tau)=\text{\emph{\O}}\\
&H^{2}_{\Z_2}\big(X,\Z(1)\big)&\ \ \text{if}\ \ &{\rm Pic}_{\rr{Q}}(X,\tau)\neq\text{\emph{\O}}\\
\end{aligned}
\right.
\qquad\quad \forall\ m\in\N
$$
and in the second case the bijection is induced (not canonically) by the first \virg{Real} Chern class. 
\end{theorem}
\noindent 
This result deserves two important observations: First of all in absence of fixed-points the FKMM-invariant can be always reduced to the first \virg{Real} Chern class mainly due to the natural isomorphism $H^2_{\Z_2}(X|\text{{\O}},\Z(1) )\simeq H^2_{\Z_2}\big(X,\Z(1)\big)$ (\cf Proposition \ref{corol:II}); 
Second of all the \virg{ambiguity} in the normalization of the FKMM-invariant just 
depends on the election of a particular (\emph{reference}) $\rr{Q}$-line bundle $\bb{L}_{\rm ref}$ which generates
the whole ${\rm Pic}_{\rr{Q}}(X,\tau)$ under the action of ${\rm Pic}_{\rr{R}}(X,\tau)$ (\cf \cite[Definition 3.3]{denittis-gomi-16}).

\medskip

The topological classification for the even-rank case
rests on \cite[Theorem 4.7]{denittis-gomi-16} which proves the injectivity of the FKMM-invariant (see Proposition \ref{prop:injectivity_in_low_dimension}) along with an analysis of the surjectivity of $\kappa$. In summary one has:
\begin{theorem}[Classification of even-rank $\rr{Q}$-bundles in low dimension]\label{theo:A_inject_fixed}
Let $(X,\tau)$ be a low-dimensional involutive space in the sense of Assumption \ref{ass:top}. Then:
\begin{itemize}
\item[(1)] If $X^\tau\neq\text{\rm \O}$ and $d=0,1,2$ the FKMM-invariant provides a bijection
$$
\kappa\;:\;{\rm Vec}_{\rr{Q}}^{2m}\big(X,\tau\big)\;\stackrel{\simeq}{\longrightarrow}\; H^{2}_{\Z_2}\big(X|X^\tau,\Z(1)\big)\;,\qquad\qquad \forall\ m\in\N\;
$$
and for $d=3$ the map 
$$
\kappa\;:\;{\rm Vec}_{\rr{Q}}^{2m}\big(X,\tau\big)\;{\hookrightarrow}\; H^{2}_{\Z_2}\big(X|X^\tau,\Z(1)\big)\;,\qquad\qquad \forall\ m\in\N\;.
$$
is in general only an injection;
\vspace{1.0mm}
\item[(2)] If $X^\tau=\text{\rm \O}$ the FKMM-invariant (identified as in Proposition \ref{corol:II}) provides a bijection
$$
\kappa\;:\;{\rm Vec}_{\rr{Q}}^{2m}\big(X,\tau\big)\;\stackrel{\simeq}{\longrightarrow}\; H^{2}_{\Z_2}\big(X,\Z(1)\big)\;,\qquad\qquad \forall\ m\in\N\;.
$$
 \end{itemize}
\end{theorem}
\noindent
The proof of (1) and (2) in the cases $d=0,1$ follows by  combining the injectivity of
$\kappa$ with the computations
\begin{equation}\label{eq:rel_equi_cohom_d=0,1}
H^{2}_{\Z_2}\big(X|X^\tau,\Z(1)\big)\;=\;0 \qquad 
\mbox{if $X$ has dimension $d=0,1$ and $X^\tau\neq \text{\O}$,}
\end{equation}
or
\begin{equation}\label{eq:equi_cohom_d=0,1}
H^{2}_{\Z_2}\big(X,\Z(1)\big)\;=\;0 \qquad 
\mbox{if $X$ has dimension $d=0,1$ and $X^\tau= \text{\O}$,}
\end{equation}
both proved in Proposition \ref{prop:zero_rel_cohom_low}. In both cases $d=0,1$ the isomorphism established by the FKMM-invariant turns out to be \emph{trivial} in view of the fact
$$
{\rm Vec}^{2m}_{\rr{Q}}\big(X, \tau\big)\;=\; 0\;,\qquad\quad\forall\ m\in\N \quad\text{if} \quad {\rm dim}(X)\leqslant 1\;.
$$ 
Then in dimension $d=0,1$ the only possible even-rank $\rr{Q}$-bundle is, up to isomorphisms, the trivial one (\cf  Proposition \ref{theo:stab_ran_Q_even}).  The proof in dimension $d=2$ for the cases (1) and (2)
is contained in  Proposition \ref{prop:surg_2_di-2step}. 
The proof of item (2) is completed by Proposition \ref{prop_free_surg_d=3} which proves the surjectivity of $\kappa$ in the free case for $d=3$.
Finally, the absence of surjectivity in $d=3$ in the case of a non-free involution  is showed by a concrete example discussed in Section \ref{app:quaternionic_lens} (\cf Corollary
\ref{prop:fkmm_lens}).

\medskip

Theorem \ref{theo:class_low_odd_R} and Theorem \ref{theo:A_inject_fixed} 
allow to compare the map
 \eqref{eq:intro_tqs5} with \eqref{eq:intro_tqs3} and \eqref{eq:intro_tqs4}. It turns out that the FKMM-invariant $\kappa$ shares all the virtues of the first Chern class or of the first \virg{Real} Chern class at least up to dimension $d=2$. This fact gives us the motivation to claim that: 

\medskip

\emph{\virg{The FKMM-invariant is the first (or fundamental) characteristic class for the category of \virg{Quaternionic} vector bundles; The one which classifies \virg{Quaternionic} vector bundles in low-dimension.} }

\medskip

The last sentence can be considered as the main (meta-)result of this paper. Of course, as already mentioned, starting from $d=4$ the invariant $\kappa$ cannot be able to produce a complete classification by itself and \emph{higher} characteristic classes are needful (see \eg \cite[Theorem 5.2]{denittis-gomi-14-gen}).

\medskip

Let us spend few words about the case $d=3$. Theorem \ref{theo:A_inject_fixed} (1) says that in this case $\kappa$ is generally only injective but not surjective. In any case the property  of being injective says that the invariant $\kappa$ is strong  enough to distinguish between non-equivalent $\rr{Q}$-bundles.
However, with the failure of the surjectivity
one loses    the possibility of classifying all the possible inequivalent realizations of $\rr{Q}$-bundles over $(X,\tau)$ just by means of the computation of $H^{2}_{\Z_2}(X|X^\tau,\Z(1))$ which is usually  an \emph{algorithmic} problem. Anyway, does not seem \virg{trivial} to find examples of 3-dimensional involutive spaces on which $\kappa$ fails to be surjective. For instance we already know that:
\begin{itemize}
\item Let $(X,\tau)$ be a $\Z_2$-CW-complex of dimension $d=3$ such that $X^\tau=\text{\O}$. Then $\kappa$ is bijective as a consequence of Theorem \ref{theo:A_inject_fixed} (2);
\vspace{1mm}

\item Let  $(X,\tau)$ be an \emph{FKMM-space} 
of dimension $d=3$ (\cf Definition \ref{defi:FKMM-space}). Then $\kappa$ is bijective as proved in Proposition \ref{propos_surjFKMM_spac}.
\vspace{1mm}

\item Let $\n{S}^{p,q}\subset\R^{p+q}$ the $p+q-1$ dimensional  sphere of radius 1 endowed with the involution
$$
(x_1,\ldots,x_p,x_{p+1},\ldots,x_{p+q})\;\longmapsto\; (x_1,\ldots,x_p,-x_{p+1},\ldots,-x_{p+q})\;.
$$
Then $\kappa$ is bijective in all the cases $\n{S}^{4-n,n}$ with $n=0,1,2,3,4$, see \cite[Section 4.4]{denittis-gomi-16}.

\vspace{1mm}

\item Consider the involutive tori
\begin{equation}\label{eq:not_inv_tor}
\n{T}^{a,b,c}\;:=\;\left(\n{S}^{2,0}\right)^{\times a}\;\times\;\left(\n{S}^{1,1}\right)^{\times b}\;\times\;\left(\n{S}^{0,2}\right)^{\times c}\;.
\end{equation}
Then $\kappa$ is bijective in all the cases $\n{T}^{a,b,c}$ with $a+b+c=3$, see \cite[Section 4.5]{denittis-gomi-16}.
\end{itemize}
Finally, the bijectivity on $3$-dimensional FKMM-spaces can be generalized to:
\begin{theorem}\label{theo-int-4}
Let $(X, \tau)$ be a $3$-dimensional compact manifold without boundary equipped with a smooth involution such that the fixed point set $X^\tau \neq \emptyset$ consists of a finite number of points. Then the FKMM-invariant is bijective.
\end{theorem}
\noindent 
The Proof of this result is postponed to Section \ref{app:surj_fkmm_free_finit}.

\medskip
To find an involutive space of dimension $d=3$ on which the FKMM-invariant fails to be surjective is quite laborious. In Section \ref{app:quaternionic_lens}
 we present an explicit example. For $q > 0$, we consider the $3$-dimensional \emph{lens space}
$L_{2q}:=\n{S}^3/\Z_{2q}$ endowed with the involution $\tau$ induced by the complex conjugation on $\n{S}^3\subset\C^2$. The fixed point set is then the disjoint union of two circles. On this space one can explicitly construct all the $\rr{Q}$-bundles by means of the equivariant version of the \emph{clutching construction} for vector bundles. This provides ${\rm Vec}_{\rr{Q}}^{2m}(L_{2q},\tau)\simeq\Z_{2q}$ (Proposition \ref{prop:B_Z_4}). On the other hand an explicit computation shows that $H^{2}_{\Z_2}(L_{2q}|L_{2q}^\tau,\Z(1)\big)\simeq\Z_{4q}$ (Proposition \ref{prop:B_Z_8}) and the immediate consequence is that $\kappa$ cannot be surjective (Corollary \ref{prop:fkmm_lens}).

\medskip

This paper is organized as follows: In {\bf Section \ref{sect:motiv_liter}} we provide some physical motivation related to the problem of the classification of (low-dimensional) $\rr{Q}$-bundles and we discuss the related literature.
In {\bf Section \ref{sect:Q-VB_topological}}, for the benefit of the reader, we 
recall all the most important notions related to the theory of $\rr{Q}$-bundles. In particular we describe
the construction
of the FKMM-invariant in the  \emph{generalized version} introduced in  \cite{denittis-gomi-16}. {\bf Section \ref{subsec:tech_results}} contains all the technical results necessary for the proof of the surjectivity of the FKMM-invariant in all the situations considered in  Theorem \ref{theo:class_low_odd_R} and Theorem \ref{theo:A_inject_fixed}. 
In  {\bf Section \ref{app:quaternionic_lens}} we describe the classification of  $\rr{Q}$-bundles over the lens space $L_{2q}$ providing a family of examples in dimension $d=3$ where the FKMM-invariant fails to be surjective.
Finally, {\bf Appendix \ref{subsec:borel_cohom}} provides a \virg{crash introduction} to the theory of the equivariant Borel cohomology.

\medskip
\noindent
{\bf Acknowledgements.} 
GD's research is supported
 by
the  grant \emph{Iniciaci\'{o}n en Investigaci\'{o}n 2015} - $\text{N}^{\text{o}}$ 11150143 funded  by FONDECYT.	 KG's research is supported by 
the JSPS KAKENHI Grant Number 15K04871.
\medskip

\section{Physical motivation and related literature}\label{sect:motiv_liter}

Before entering in the core of the mathematical results proved in this paper let us briefly mention the relevance of the problem of the classification of {\virg{Real} and \virg{Quaternionic} vector bundles for a certain class of  mathematical physical applications. 

\medskip

In its simplest incarnation, a \emph{Topological Quantum System} (TQS) is a continuous matrix-valued map 
\begin{equation}\label{eq:intro_tqs0}
X\;\ni\;x\;\longmapsto\; H(x)\;\in\; \rm{Mat}_N(\C)
\end{equation}
defined on a \virg{nice} topological space $X$.
 Although a precise definition of TQS requires some more ingredients (see \eg \cite{denittis-gomi-14,denittis-gomi-14-gen,denittis-gomi-15,denittis-gomi-15-bis}), 
one can certainly
state that
the most relevant feature of these systems is the nature of the spectrum which is made by $N$ continuous \virg{energy} bands (taking into account possible degeneracies).
It is exactly this peculiar band structure which may encode information that are of topological nature. More precisely, let assume that it is possible to select $m<N$ bands that do not cross the other $N-m$ bands. Then, it is possible to construct a continuous projection-valued map  $X\ni x\mapsto P(x)\in \rm{Mat}_N(\C)$ such that $P(x)$ is the rank $m$ spectral projection of $H(x)$ associated to the spectral subspace selected by the $m$ energy bands at the point $x$.
Due to the classical Serre-Swan construction \cite{serre-55,swan-62} one can associate  to $x\mapsto P(x)$ a unique (up to isomorphisms) rank $m$ complex vector bundle $\bb{E}\to X$  sometimes called the \emph{spectral bundle} (see \cite[Section 2]{denittis-gomi-14} or \cite[Section 4]{denittis-lein-13} for the details of the construction). The considerable consequence of the duality between  gapped TQSs  and  spectral bundles  is that one can classify the possible topological phases of a TQS by means of the  elements of the set ${\rm Vec}_{\C}^{m}(X)$ of isomorphism classes of rank  $m$ complex vector bundles over $X$. As a result one is allowed to translate the problem of the enumeration of the possible topological phases of a TQS  into the classical problem in topology of the classification of 
${\rm Vec}_{\C}^{m}(X)$. The important result due to F. P. Peterson
\cite{peterson-59}  establishes that  this classification can be achieved in a \emph{computable way} by using the Chern classes which take values in the cohomology groups $H^{2j}(X,\Z)$.
In particular,  in  \emph{low dimension} 
the classification is completely specified by the first Chern class $c_1$ according to 
\eqref{eq:intro_tqs3}.

\medskip

TQSs of type \eqref{eq:intro_tqs0}
 are ubiquitous in mathematical physics (see e.g. the rich monographs \cite{bohm-mostafazadeh-koizumi-niu-zwanziger-03,chruscinski-jamiolkowski-04}). 
 They can be used to model
systems subjected to  \emph{cyclic adiabatic processes} in  classical and quantum mechanics \cite{pancharatnam-56,berry-84} or in the description of the \emph{magnetic monopole} \cite{dirac-31, yang-96}
and  the \emph{Aharonov-Bohm effect} \cite{aharonov-bohm-59} 
or in the molecular dynamics in the context of the  \emph{Born-Oppenheimer approximation} \cite{baer-06},
 just to mention few important examples. 
Probably the most popular example of a TQS comes from the Condensed Matter Physics
and concerns the dynamics of (independent) electrons in a crystalline periodic background. In this case the Bloch-Floquet formalism \cite{ashcroft-mermin-76,kuchment-93} allows to decompose the Schr\"odinger operator in a parametric family of operators like in \eqref{eq:intro_tqs0} labelled by the points of a torus $X=\n{T}^d$ ($d=1,2,3$),
usually known as the \emph{Brillouin zone}. In this particular case the classification of the topological phases  is completely specified by $H^2(\n{T}^d,\Z)$ due to \eqref{eq:intro_tqs3}, and the different topological phases are interpreted as the distinct \emph{quantized} values of the Hall conductance by means of the celebrated \emph{Kubo-Chern formula} \cite{thouless-kohmoto-nightingale-nijs-82,bellissard-elst-schulz-baldes-94}. The last result provides the theoretical explanation
of the \emph{Quantum Hall Effect} which is the prototypical example of topological insulating phases. Nowadays the study 
of topologically protected phases of \emph{topological insulators} is a \virg{hot topic} in Condensed Matter. Due to the vastness of the bibliography about topological insulators we refer to  two recent reviews  \cite{hasan-kane-10} and \cite{ando-fu-15} for an almost complete  overview on the subject.

\medskip

The problem of the classification of the topological phases becomes more interesting, and challenging, when the TQS is constrained by the presence of certain \emph{(pseudo-)symmetries} like the {time-reversal symmetry} (TRS). A system like \eqref{eq:intro_tqs0} is said to be time-reversal symmetric if there is an  \emph{involution} $\tau:X\to X$ on the base space and an \emph{anti}-unitary map $\Theta$ such that
\begin{equation}\label{eq:intro_tqs3bis}
\left\{
\begin{aligned}
\Theta\;H(x)\;\Theta^*\;&=\; H\big(\tau(x)\big)\;,&\quad\qquad \forall\ x\in\ X\;\\
\Theta^2\;&=\;\epsilon\;\n{1}_{N}& \epsilon=\pm1\;
\end{aligned}
\right.
\end{equation}
where $\n{1}_{N}$ denotes the $N\times N$ identity matrix.

\medskip

The case $\epsilon=+1$ corresponds to an \emph{even} (sometimes called \emph{bosonic}) TRS. In this case, the spectral 
 bundle $\bb{E}$ turns out to be equipped with the additional structure of an $\rr{R}$-bundle as showed in \cite[Section 2]{denittis-gomi-14}. Therefore, in the presence of an even TRS the  classification problem of the topological phases 
of a TQS is reduced to the study of the set  
${\rm Vec}_{\rr{R}}^{m}(X,\tau)$ and, at least in low-dimension,  the isomorphism 
\eqref{eq:intro_tqs4}
induced by  the first \virg{Real} Chern class $c_1^{\rr{R}}$  completely answers the question.

\medskip

The case $\epsilon=-1$ describes an \emph{odd} (sometimes called \emph{fermionic}) TRS. Also in this situation the spectral 
vector bundle $\bb{E}$ acquires an additional structure which converts $\bb{E}$ in a $\rr{Q}$-bundle.
Then, the topological phases of a TQS with an odd TRS are labelled  by the set ${\rm Vec}_{\rr{Q}}^{m}(X,\tau)$ and it becomes relevant for the study of these systems to have proper tools able to classify equivalence classes of $\rr{Q}$-bundles.

\medskip

The study of   systems with an odd TRS is more interesting, and for several reasons also harder, than the case of an even TRS. Historically, the fame of these fermionic systems begins with the seminal papers \cite{kane-mele-05,fu-kane-mele-95}  by 
L. Fu,  C. L. Kane and E. J. Mele. The central result of these works is the interpretation of a physical phenomenon called \emph{Quantum Spin Hall Effect} as the evidence  of a non-trivial topology for TQS constrained by an odd TRS. Specifically, the papers \cite{kane-mele-05,fu-kane-mele-95} are concerned about the study of systems like \eqref{eq:intro_tqs0} (with $N=4$ and $m=2$) where the base space is an involutive torus of type $\n{T}^{0,d,0}$ with $d=2,3$ (compare with \eqref{eq:not_inv_tor} for the notation).
 The distinctive aspect of this special \emph{involutive Brillouin zone}   is the existence of a \emph{fixed point set} formed by $2^d$ isolated points. The latter plays a crucial role in the classification scheme proposed in \cite{kane-mele-05,fu-kane-mele-95} where the different topological phases are distinguished by the signs that a particular function $\rr{d}_{\bb{E}}$  (essentially the inverse of a normalized pfaffian constructed from a particular frame of $\bb{E}$) takes on the $2^d$ fixed-points.
These $\Z_2$-numbers are usually known as \emph{Fu-Kane-Mele indices}. 

\medskip

In the last years the problem of the topological classification  of systems with an odd TRS has been discussed with several different approaches.  
As a matter of fact, many (if not almost all) of these approaches focus on the particular cases $\n{T}^{0,2,0}$ and $\n{T}^{0,3,0}$  with the aim of reproducing in different way the $\Z_2$-invariants described by the {Fu-Kane-Mele indices}.
From one hand there are classification schemes based on \emph{K-theory}  and \emph{$KK$-theory} \cite{kitaev-09,freed-moore-13,thiang-16,kellendonk-15,kubota-17,bourne-carey-rennie-16,prodan-schulz-baldes-16}
or \emph{equivariant homotopy} techniques \cite{kennedy-guggenheim-15,kennedy-zirnbauer-16} which are extremely general. In the opposite side there are constructive procedures based on the interpretation of the topological phases as \emph{obstructions} for the construction of trivial time-reversal symmetric frames (see \cite{porta-graf-13} for the case ${\T}^{0,2,0}$ and \cite{fiorenza-monaco-panati-14,monaco-cornean-teufel-16} for the generalization to the case ${\T}^{0,3,0}$)
or as \emph{spectral-flows} \cite{carey-phillips-schulz-baldes-16,denittis-schulz-baldes-13} or as \emph{index pairings} \cite{grossmann-schulz-baldes-15}.
All these approaches, in our opinion, present some limitations. The $K$-theory 
is unable to distinguish the \virg{spurious phases} (possibly) present outside of the \emph{stable rank} regime. The homotopy calculations are non-algorithmic and usually extremely hard. The use of the $KK$-theory 
up to now is restricted only to  the non-commutative version of the involutive Brillouin tori ${\T}^{0,2,0}$ and 
${\T}^{0,3,0}$. A similar consideration holds for the  recipes for the \virg{handmade} construction of trivial frames which are strongly dependent of the specific form of the involutive spaces ${\T}^{0,2,0}$ and 
${\T}^{0,3,0}$, and thus are  difficult to generalize  to other spaces and higher dimensions. Finally, none of these approaches clearly identifies the invariant which labels the different phases as a \emph{topological class},  as it happens in the cases of systems with broken TRS (\cf eq. \eqref{eq:intro_tqs3}) or with even TRS (\cf eq. \eqref{eq:intro_tqs4}).

\medskip

As already discussed,
TQS  of type \eqref{eq:intro_tqs0}
are ubiquitous in mathematical physics and there are no reason to focus the interest  only on the special examples coming from the physics of periodic electronic systems (a.k.a topological insulators).  Usually, one can think to the space $X$  as a \emph{configuration} space for {parameters} which describe an \emph{adiabatic} action of external fields on a system governed by the instantaneous
 Hamiltonian $H(x)$.
The phenomenology of these systems can be  enriched by the presence of certain symmetries like a TRS as in \eqref{eq:intro_tqs3bis}. Models of \emph{adiabatic} topological systems of this type have been recently investigated in 
\cite{carpentier-delplace-fruchart-gawedzki-15,carpentier-delplace-fruchart-gawedzki-tauber-15,gat-robbins-15}.
In particular, in  \cite{gat-robbins-15} the authors 
consider the adiabatically perturbed dynamics of a \emph{classic rigid rotor} and a \emph{classical  particle on a ring}. In the first case the classical phase space turns out to be $\n{S}^{0,3}$, namely a two-dimensional sphere endowed with the antipodal \emph{free} involution induced by the
 TRS. In the second case the phase space is a two-dimensional torus of type $\T^{1,1,0}$ which has a fixed-point set  of co-dimension one. None of these two cases can be treated with the technology developed for topological insulator due to the fact that  the
{Fu-Kane-Mele index} is just ill-defined when the fixed-point set is empty or of dimension higher than zero.
 The main result of 
\cite{gat-robbins-15} consists in  the classification of $\rr{Q}$-bundles  
over the involutive spaces $\n{S}^{0,3}$ and $\T^{1,1,0}$ and it is based on the analysis of the obstruction for
the \virg{handmade} construction of a trivial frame. As  a pay-off they obtained a classification which is not based on $\Z_2$-invariants, a fact that seems obvious from a geometric point of view but which seems to be \virg{revolutionary} if compared to the literature on TR-symmetric topological insulators  which is entirely focused on $\Z_2$-type invariants.

\medskip

The 
technique 
used in \cite{gat-robbins-15}
is hard (and tricky) to extend to higher dimensions or to  involutive spaces
different from $\n{S}^{0,3}$ and $\T^{1,1,0}$. Conversely, the classification provided by the map \eqref{eq:intro_tqs5} turns out to be 
extremely effective and versatile. In fact the invariant $\kappa$ is
\emph{intrinsic},
\emph{universal} and
 \emph{algorithmically computable}! As an example the formula \eqref{eq:intro_tqs5} allowed us to classify $\rr{Q}$-bundles over   a big class of involutive spheres and tori up to dimension three extending, in this way, the results in  \cite{gat-robbins-15}, see \cite[Table 1.1, Table 1.2 \& Table 1.3]{denittis-gomi-16}.

\medskip

 In conclusion, in our opinion, the identification of the FKMM-invariant $\kappa$ as the first (or fundamental) characteristic class for the category of \virg{Quaternionic} vector bundles can be of big utility in many questions of classification of topological phases arising from concrete physical problems.

\section{\virg{Quaternionic} vector bundles from a topological perspective}\label{sect:Q-VB_topological}
This section is devoted to the readers which are not familiar with the theory of \virg{Quaternionic} vector bundles. We provide here the main definitions,  discuss
the possibility of $\rr{Q}$-line bundles,  describe
 the stable rank splitting of $\rr{Q}$-bundles and finally  introduce the FKMM-invariant.

\subsection{Basic facts about \virg{Quaternionic} vector bundles}
\label{sect:basic_def}
In this section we recall some basic facts about the topological category of \virg{Quaternionic} vector bundles and we refer to \cite{dupont-69,denittis-gomi-14-gen,denittis-gomi-16} for a more systematic presentation. 

\medskip

\begin{definition}[{\virg{Quaternionic}} vector bundles]\label{defi:Q_VB}
A {\virg{Quaternionic}} vector bundle, or $\rr{Q}$-bundle, over  $(X,\tau)$
is a complex vector bundle $\pi:\bb{E}\to X$ endowed with a (topological) homeomorphism $\Theta:\bb{E}\to \bb{E}$
 such that:
\begin{itemize}

\item[$(Q_1)$] The projection $\pi$ is \emph{equivariant} in the sense that $\pi\circ \Theta=\tau\circ \pi$;
\vspace{1mm}
\item[$(Q_2)$] $\Theta$ is \emph{anti-linear} on each fiber, \ie $\Theta(\lambda p)=\overline{\lambda}\ \Theta(p)$ for all $\lambda\in\C$ and $p\in\bb{E}$ where $\overline{\lambda}$ is the complex conjugate of $\lambda$;
\vspace{1mm}
\item[$(Q_3)$] $\Theta^2$ acts fiberwise as the multiplication by $-1$, namely $\Theta^2|_{\bb{E}_x}=-\n{1}_{\bb{E}_x}$.
\end{itemize}
\end{definition}

\noindent
Let us recall that it
is always possible to endow
$\bb{E}$ with
an essentially unique \emph{equivariant} Hermitian metric $\rr{m}$ with respect to which $\Theta$ is an \emph{anti-unitary} map between conjugate fibers
\cite[Proposition 2.5]{denittis-gomi-14-gen}. We recall that equivariant means that 
$$
\rr{m}\big(\Theta(p_1),\Theta(p_2)\big)\,=\,\rr{m}\big(p_2,p_1\big)\;,\qquad\quad\forall\ (p_1,p_2)\in\bb{E}\times_\pi\bb{E}
$$
where $\bb{E}\times_\pi\bb{E}
:=\{(p_1,p_2)\in\bb{E}\times\bb{E}\ |\ \pi(p_1)=\pi(p_2)\}$.
 
  \medskip
 
A vector bundle \emph{morphism} $f$ between two vector bundles  $\pi:\bb{E}\to X$ and $\pi':\bb{E}'\to X$ 
over the same base space
is a  continuous map $f:\bb{E}\to \bb{E}'$ which is \emph{fiber preserving} in the sense that  $\pi=\pi'\circ f$
 and that restricts to a \emph{linear} map on each fiber $\left.f\right|_x:\bb{E}_x\to \bb{E}'_x$. Complex vector bundles over  $X$ together with vector bundle morphisms define a category and the symbol ${\rm Vec}^m_\C(X)$
 is used to denote the set of equivalence classes of isomorphic vector bundles of rank $m$.
    Also 
 $\rr{Q}$-bundles define a category with respect to  \emph{$\rr{Q}$-morphisms}. A $\rr{Q}$-morphism $f$ between two  $\rr{Q}$-bundles
 $(\bb{E},\Theta)$ and $(\bb{E}',\Theta')$ over the same involutive space $(X,\tau)$ 
  is a vector bundle morphism  commuting with the involutions, \ie $f\circ\Theta\;=\;\Theta'\circ f$. The set of equivalence classes of isomorphic $\rr{Q}$-bundles of  rank $m$ over $(X,\tau)$ 
  will be denoted by ${\rm Vec}_{\rr{Q}}^m(X,\tau)$.

\begin{remark}[{\virg{Real}} vector bundles]\label{rk_real}{\upshape
By changing condition $(Q_3)$
 in Definition \ref{defi:Q_VB}   with 
\begin{itemize}
\item[$(R)$] \emph{$\Theta^2$ acts fiberwise as the multiplication by $1$, namely $\Theta^2|_{\bb{E}_x}=\n{1}_{\bb{E}_x}$}\;
\end{itemize}
one ends in the category of \emph{\virg{Real}} (or $\rr{R}$)  \emph{vector bundles}. 
Isomorphism classes of rank $m$ $\rr{R}$-bundles  over the involutive space $(X,\tau)$ 
are denoted by ${\rm Vec}_{\rr{R}}^m(X,\tau)$. 
For more details we refer to \cite{atiyah-66,denittis-gomi-14}.
}\hfill $\blacktriangleleft$
\end{remark}

\medskip
\noindent
In the case of a trivial involutive space $(X,{\rm Id}_X)$ one has isomorphisms
\begin{equation}\label{eq:iso_trivial_involution}
{\rm Vec}_{\rr{Q}}^{2m}\big(X,{\rm Id}_X\big)\,\simeq\,{\rm Vec}_{\n{H}}^{m}\big(X\big)\;,\qquad\quad {\rm Vec}_{\rr{R}}^{m}\big(X,{\rm Id}_X\big)\,\simeq\,{\rm Vec}_{\n{R}}^{m}\big(X,{\rm Id}_X\big)\;,\qquad \ \ m\in\N
\end{equation}
where ${\rm Vec}_{\n{F}}^{m}\big(X\big)$ is the set of equivalence classes of vector bundles over
$X$ with typical fiber $\n{F}^m$ and $\n{H}$ denotes the skew field of quaternions. The proof of these isomorphisms are easy, and interested reader can find them in \cite[Proposition 2.2]{denittis-gomi-14-gen} and in \cite[Proposition 4.5]{denittis-gomi-14} for instance.
These two results justify the names \virg{Quaternionic} and \virg{Real} for the related categories.

\medskip

 Let $x\in X^\tau$ and $\bb{E}_x\simeq\C^m$ be the related fiber. In this case the restriction $\Theta|_{\bb{E}_x}\equiv J$ defines an  \emph{anti}-linear map $J : \bb{E}_x \to \bb{E}_x$ such that $J^2 = -\n{1}_{\bb{E}_x}$. Said differently the fiber $\bb{E}_x$ over each fixed point $x\in X^\tau$ is endowed with a \emph{quaternionic} structure (\cf \cite[Remark 2.1]{denittis-gomi-14-gen}).
This fact has an important  consequence: 
\begin{proposition}[{\cite[Proposition 2.1]{denittis-gomi-14-gen}}]
If $X^\tau\neq \empt$ then every \virg{Quaternionic} vector bundle over $(X,\tau)$   has necessarily even rank.
\end{proposition}

\medskip

\noindent
The set
 ${\rm Vec}_{\rr{Q}}^{2m}(X,\tau)$ is non-empty since it contains at least the 
 \emph{trivial} element in the \virg{Quaternionic} category.
\begin{definition}[\virg{Quaternionic} product  bundle]\label{defi:quat_product_bund}
The rank $2m$ \virg{Quaternionic} product bundle over the involutive space $(X,\tau)$ is the complex vector bundle
$$
X\,\times\,\C^{2m}\,\longrightarrow\, X
$$
endowed with the \emph{product} $\rr{Q}$-structure
$$
\Theta_0 (x,{\rm v})\,=\,(\tau(x),Q\;\overline{{\rm v}})\;,\qquad\quad (x,{\rm v})\,\in\,X\,\times\,\C^{2m}
$$ 
where the matrix $Q$ is given by
\beql{eq:Q-mat}
Q
\;:=\;\left(
\begin{array}{rr}
0 & -1    \\
1 &  0 
\end{array}
\right)\;\otimes\;\n{1}_m\;=\;
\left(
\begin{array}{rr|rr|rr}
0 & -1 &        &        &   &    \\
1 &  0 &        &        &   &    \\
\hline
  &    & \ddots &        &   &    \\
  &    &        & \ddots &   &    \\
\hline
  &    &        &        & 0 & -1 \\
  &    &        &        & 1 &  0
\end{array}
\right)\;.
\eeq
\end{definition}

\medskip

\noindent
A \virg{Quaternionic} vector bundle is called \emph{$\rr{Q}$-trivial}  if and only if  it is isomorphic  to the \virg{Quaternionic} product  bundle in the category of $\rr{Q}$-bundles. Let us point out that when $X^\tau= \empt$ the sets ${\rm Vec}_{\rr{Q}}^{2m+1}(X,\tau)$ can be non-empty but in general there is no obvious candidate for the 
trivial \virg{Quaternionic} vector bundle in the odd rank case (see \eg \cite[Example 3.8]{denittis-gomi-16}).

\medskip

A \emph{section} of a complex vector bundle $\pi:\bb{E}\to X$ is a continuous map $s:X\to\bb{E}$ such that $\pi\circ s ={\rm Id}_X$. The set of sections $\Gamma(\bb{E})$ has the structure of a left $C(X)$-module with multiplication given by the pointwise product $(fs)(x):=f(x)s(x)$ for any $f\in C(X)$ and $s\in\Gamma(\bb{E})$ and for all $x\in X$. 
If $(\bb{E}, \Theta)$ is a $\rr{Q}$-bundle over $(X,\tau)$ then $\Gamma(\bb{E})$ is endowed with a natural anti-linear \emph{anti-involution} $\tau_\Theta:\Gamma(\bb{E})\to\Gamma(\bb{E})$ given by
$$
\tau_\Theta(s)\,:=\,\Theta\,\circ\, s\,\circ\, \tau\,.
$$
The compatibility with the $C(X)$-module structure is given by
$$
\tau_\Theta(fs)\,=\,\tau_*(f)\,\tau_\Theta(s)\,,\qquad\quad f\in C(X)\,,\ \ s\in\Gamma(\bb{E})
$$
where the anti-linear \emph{involution} $\tau_*:C(X)\to C(X)$ is defined by $\tau_*(f)(x):=\overline{f(\tau(x))}$. The triviality of a \virg{Quaternionic} vector bundle can be characterized in terms of global $\rr{Q}$-frames of sections \cite[Definition 2.1 \& Theorem 2.1]{denittis-gomi-14-gen}.

\medskip

$\rr{Q}$-bundles are  locally trivial  in the category of vector bundles over involutive spaces (that is $\rr{Q}$-locally trivial) \cite[Proposition 2.4]{denittis-gomi-14-gen} and fulfill the homotopy property with respect to equivariant homotopy deformation \cite[Theorem 2.3]{denittis-gomi-14-gen}. Let us just recall that given two involutive spaces
$(X_1,\tau_1)$ and $(X_2,\tau_2)$ one says that a continuous map $\phi:X_1\to X_2$ is \emph{equivariant} if and only if $\phi\circ \tau_1=\tau_2\circ\phi$. An equivariant homotopy between equivariant maps $\phi_0$ and $\phi_1$ is a continuous map $F:[0, 1]\times X_1\to X_2$ such that $\phi_t(\cdot):=F(t,\cdot)$  is equivariant for all $t\in[0,1]$. The set of the equivalence classes of equivariant maps between $(X_1,\tau_1)$ and $(X_2,\tau_2)$ with respect to the relation given by the equivariant homotopy is denoted by $[X_1, X_2]_{\Z_2}$. The equivariant homotopy property is the basis of the \emph{homotopy classification} for \virg{Quaternionic} vector bundles  \cite[Theorem 2.4]{denittis-gomi-14-gen}, namely
\begin{equation}\label{eq:homotopy_classification}
{\rm Vec}_{\rr{Q}}^{2m}\big(X,\tau\big)\,\simeq\,\big[X, \hat{G}_{2m}(\C^{\infty})\big]_{\Z_2}
\end{equation}
where $\hat{G}_{2m}(\C^{\infty}):=({G}_{2m}(\C^{\infty}),\rho)$ is the Grassmann manifold of $2m$-planes inside $\C^{\infty}$ endowed with a suitable quaternionic involution  $\rho$ (see \cite[Section 2.4]{denittis-gomi-14-gen} for more details).
The \virg{direct} classification of ${\rm Vec}_{\rr{Q}}^{2m}(X,\tau)$ through the calculation of $[X, \hat{G}_{2m}(\C^{\infty})]_{\Z_2}$ is an extraordinarily difficult task. For this reason one needs to develop other tools for classifying $\rr{Q}$-bundles. An effective way (at least in low dimension) is through the use of the  FKMM-invariant described in Section \ref{subsec:fkmm-invariant}.

\subsection{\virg{Real} and \virg{Quaternionic} line bundles}
\label{sect:real_quat_line_bund}
Let us introduce 
$$
{\rm Pic}_{\rr{R}}\big(X,\tau\big)\,\equiv\,{\rm Vec}_{\rr{R}}^1\big(X,\tau\big)\;,\qquad\quad {\rm Pic}_{\rr{Q}}\big(X,\tau\big)\,\equiv\,{\rm Vec}_{\rr{Q}}^1\big(X,\tau\big)
$$ 
which are the sets of isomorphism classes of \virg{Real} and \virg{Quaternionic} line bundles
on $(X,\tau)$, respectively. 
As a matter of fact $\rr{Q}$-bundles of odd rank can be defined only over involutive base spaces $(X,\tau)$  with a \emph{free} involution (meaning that $X^\tau=\text{\rm \O}$) \cite[Proposition 2.1]{denittis-gomi-14-gen}. On the other hand there is no restriction for the definition of $\rr{R}$-line bundles
and the set ${\rm Pic}_{\rr{R}}(X,\tau)$ gives rise to an abelian group under the tensor
product which  is known as the \emph{\virg{Real} Picard group}. This group is classified by the  \emph{Kahn's isomorphism} \cite{kahn-59} (see also  \cite[Corollary A.5]{gomi-13})
\begin{equation}\label{eq:kahn_iso}
c^{\rr{R}}_1\;:\;{\rm Pic}_{\rr{R}}(X,\tau)\;\stackrel{\simeq}{\longrightarrow}\;H^2_{\Z_2}\big(X,\Z(1)\big)
\end{equation}
where in the right-hand side there is the Borel equivariant cohomology group of $X$ with local coefficients $\Z(1)$ (we refer to Appendix \ref{subsec:borel_cohom} for more details) and the characteristic class $c^{\rr{R}}_1$ which realizes the isomorphism is called the first \emph{\virg{Real} Chern class}.

\medskip

On the other hand  ${\rm Pic}_{\rr{Q}}(X,\tau)$, when definable, does not possess a  group structure under the tensor
product.
 However, under the essential assumption that ${\rm Pic}_{\rr{Q}}(X,\tau)\neq\text{\rm \O}$,
 one can use the tensor product to define a left group-action of
${\rm Pic}_{\rr{R}}(X,\tau)$  on the set ${\rm Pic}_{\rr{Q}}(X,\tau)$:
$$
\begin{aligned}
 &{\rm Pic}_{\rr{R}}(X,\tau)\;\times\; {\rm Pic}_{\rr{Q}}(X,\tau)&\;\longrightarrow\;&\ \ \ \ \ \ \ {\rm Pic}_{\rr{Q}}(X,\tau)&\\
 &\ \ \ \ \ \  \ \  \ ([\bb{L}_{\rr{R}}],[\bb{L}_{\rr{Q}}])&\;\longmapsto\;&\ \ \ \ \  [\bb{L}_{\rr{R}}\;\otimes\;\bb{L}_{\rr{Q}}]\;.&
\end{aligned}
$$
It turns out that this action defines a \emph{torsor}. According to \cite[Theorem 3.1 \& Corollary 3.2]{denittis-gomi-16} one has that:
\begin{theorem}[\virg{Quaternionic} Picard torsor]\label{teo:torsor_pica}
Let $(X,\tau)$ be an involutive space with free involution, \ie  $X^\tau=\text{\rm \O}$.
Assume that ${\rm Pic}_{\rr{Q}}(X,\tau)\neq\text{\rm \O}$. Then, ${\rm Pic}_{\rr{Q}}(X,\tau)$ is a torsor under the group-action induced by ${\rm Pic}_{\rr{R}}(X,\tau)$
and one has a bijection of sets
$$
{\rm Pic}_{\rr{Q}}(X,\tau)\;\simeq\;{\rm Pic}_{\rr{R}}(X,\tau)\;\simeq\;H^2_{\Z_2}\big(X,\Z(1)\big)
$$
\end{theorem}

\medskip

In \cite[Section 3.2]{denittis-gomi-16} the interested reader can find various  examples of non-trivial   \virg{Quaternionic} Picard torsors defined over spheres with free antipodal involutions.

\subsection{Stable range}
\label{sect:stable_range}
The {stable rank condition} 
expresses the pretty general fact that the non trivial topology of a vector bundle
can be concentrated in a sub-vector bundle of \virg{minimal rank}. This minimal value  depends on the dimensionality of the base space and on the category of vector bundles under consideration. 
For complex (as well as real or quaternionic) vector bundles the  stable rank condition is a well-known result (see \eg \cite[Chapter 9, Theorem 1.2]{husemoller-94}), based on an obstruction-type argument \cite[Chapter 2, Theorem 7.1]{husemoller-94}.  The key argument can be generalized to vector bundles over spaces with involution 
by means of the notion of  $\Z_2$-CW-complex \cite{matumoto-71,allday-puppe-93} (see also \cite[Section 4.5]{denittis-gomi-14}). This allows to determinate   the 
stable rank condition in the case of the \virg{Real}  and the \virg{Quaternionic} categories. 
In the \virg{Real} case the {stable rank condition is described by the following result.
 \begin{proposition}[Stable condition for $\rr{R}$-bundles]\label{rk:stable_rank_Real}
Let $(X,\tau)$ be an involutive space such that $X$ has a finite $\Z_2$-CW-complex decomposition of dimension $d$. Assume that: (a)  $X^\tau=\text{\rm \O};$ or (b) $X^\tau$ is a $\Z_2$-CW-complex of dimension zero. Then,
\beql{eq:stab_rank_R_low_d}
{\rm Vec}^{m}_{\rr{R}}\big(X, \tau\big)\;\simeq\; {\rm Vec}^{\sigma}_{\rr{R}}(X, \tau)\qquad\quad \forall\ m \;\geqslant\; \frac{d+1}{2}\;
\eeq
where $\sigma:=[\frac{d}{2}]$ (here $[x]$ denotes the integer part of $x\in\R$). In particular, in low dimension one obtains:
\begin{align}
{\rm Vec}^{m}_{\rr{R}}\big(X, \tau\big)&\;=\; 0& \text{if}& \ \ d=0,1\ \ \ \ \ \ \forall\ m\in\N \label{eq:stab_rank_R_low_d=1}\\
{\rm Vec}^{m}_{\rr{R}}\big(X, \tau\big)&\;\simeq\; {\rm Pic}_{\rr{R}}\big(X,\tau\big)& \text{if}& \ \ 2\leqslant d\leqslant 3\ \ \ \forall\ m\in\N\;.\label{eq:stab_rank_R_low_d>1}
\end{align}
\end{proposition}
\medskip

\noindent
For the proof of this result the reader can refer to \cite[Theorem 4.25]{denittis-gomi-14} and \cite[Remark 4.3]{denittis-gomi-16}. Let us mention that the situation turns out to be quite different when $X^\tau$ is a $\Z_2$-CW-complex of dimension higher than zero
as pointed out in \cite[Remark 4.24]{denittis-gomi-14}.

\medskip

In the case of $\rr{Q}$-bundles it is easier to consider separately the even-rank case and the odd-rank case which necessarily requires  $X^\tau=\text{\rm \O}$.
\begin{proposition}[Stable condition for $\rr{Q}$-bundles: even-rank]
\label{theo:stab_ran_Q_even}
Let $(X,\tau)$ be an involutive space such that $X$ has a finite $\Z_2$-CW-complex decomposition of dimension $d$. Then,
\beql{eq:stab_rank_Q_low_d}
{\rm Vec}^{2m}_{\rr{Q}}\big(X, \tau\big)\;\simeq\; {\rm Vec}^{2\sigma}_{\rr{Q}}\big(X, \tau\big)\qquad\quad \forall\ m \;\geqslant\; \frac{d+3}{4}\;
\eeq
where $\sigma:=[\frac{d+2}{4}]$. In particular, in low dimension one obtains:
\begin{align}
{\rm Vec}^{2m}_{\rr{Q}}\big(X, \tau\big)&\;=\; 0& \text{if}& \ \ d=0,1\ \ \ \ \ \ \forall\ m\in\N \label{eq:stab_rank_Q_low_d=1}\\
{\rm Vec}^{2m}_{\rr{Q}}\big(X, \tau\big)&\;\simeq\; {\rm Vec}^{2}_{\rr{Q}}\big(X, \tau\big)& \text{if}& \ \ 2\leqslant d\leqslant 5\ \ \ \forall\ m\in\N\;.\label{eq:stab_rank_Q_low_d>1}
\end{align}
\end{proposition}

\medskip

\noindent
The proof of this result has been given first in \cite[Theorem 2.5]{denittis-gomi-14-gen} under certain conditions for the fixed-point set $X^\tau$ and then generalized in \cite[Theorem 4.2]{denittis-gomi-16}.

\medskip

The odd-rank case is slightly different  and strongly  depends on the existence
 of a $\rr{Q}$-line bundle. Consider an involutive space  $(X,\tau)$ such that $X^\tau=\text{\rm \O}$. Then one can show that
$$
{\rm Vec}^{2m+1}_{\rr{Q}}\big(X, \tau\big)\;\neq\;\text{\rm \O}\ \ \ \ \ \Leftrightarrow \ \ \ \ \  {\rm Pic}_{\rr{Q}}\big(X,\tau\big)\;\neq\;\text{\rm \O}\;.
$$ 
 and in the interesting case ${\rm Pic}_{\rr{Q}}(X,\tau)\neq\text{\rm \O}$ there are 
  bijections of sets
$$
{\rm Vec}^{m}_{\rr{Q}}\big(X, \tau\big)\;\simeq\;{\rm Vec}^{m}_{\rr{R}}\big(X, \tau\big)\;\qquad\quad \forall\ m\in\N\;.
$$
These last two facts are proved in  \cite[Lemma 4.4]{denittis-gomi-16} and are key arguments for the determination of the stable condition in the odd-rank case \cite[Theorem 4.5]{denittis-gomi-16}.
\begin{proposition}[Stable condition for $\rr{Q}$-bundles: odd-rank]
\label{theo:stab_ran_Q_odd}
Let $(X,\tau)$ be an involutive space such that $X$ has a finite $\Z_2$-CW-complex decomposition of dimension $d$,  $X^\tau=\text{\emph{\O}}$ and  ${\rm Pic}_{\rr{Q}}(X,\tau)\neq\text{\emph{\O}}$. Then 
$$
{\rm Vec}^{2m+1}_{\rr{Q}}\big(X, \tau\big)
\;\simeq\;  {\rm Vec}^{\sigma}_{\rr{Q}}\big(X, \tau\big)\qquad\quad \forall\ m \;\geqslant\; \frac{d-1}{4}
$$
where $\sigma:=[\frac{d}{2}]$.
In particular, in low dimension one obtains:
\begin{align}
{\rm Vec}^{2m-1}_{\rr{Q}}\big(X, \tau\big)&\;=\; 0& \text{if}& \ \ d=0,1\ \ \ \ \ \ \forall\ m\in\N \label{eq:stab_rank_Q_low_d>1!!xxx}\\
{\rm Vec}^{2m-1}_{\rr{Q}}\big(X, \tau\big)&\;=\; {\rm Pic}_{\rr{Q}}\big(X,\tau\big)& \text{if}& \ \ d=2,3\ \ \ \ \ \ \forall\ m\in\N \label{eq:stab_rank_Q_low_d=1_odd_case}\\
{\rm Vec}^{2m-1}_{\rr{Q}}\big(X, \tau\big)&\;\simeq\; {\rm Vec}^{2}_{\rr{Q}}\big(X, \tau\big)& \text{if}& \ \ d=4,5\ \ \ \ \ \ \forall\ m\in\N\;.\label{eq:stab_rank_Q_low_d>1!!}
\end{align}
\end{proposition}

\medskip

\noindent
We point out that the $0$ in the \eqref{eq:stab_rank_Q_low_d>1!!xxx} refers to the existence of a unique element which could also be different from the (trivial) product  
$\rr{Q}$-bundle. This is in contrast with the even-rank case where
  the condition ${\rm Vec}^{2m}_{\rr{Q}}(X, \tau)\neq\text{\rm \O}$ is 	always guaranteed by the existence of the trivial  element described in  Definition \ref{defi:quat_product_bund}.

\subsection{The determinant functor}
\label{subsec:det_construct}
Let $\bb{V}$ be a complex vector space of dimension $n$. The {determinant} of  $\bb{V}$ is by definition ${\rm det}(\bb{V}):=\bigwedge^n\bb{V}$ where the symbol  $\bigwedge^n$ denotes the top exterior power of $\bb{V}$ (\ie the skew-symmetrized $n$-th tensor power of $\bb{V}$). This is a complex vector space of dimension one.
If $\bb{W}$ is a second vector space of the same dimension $n$
and $T:\bb{V}\to \bb{W}$ is a linear map then there is a naturally associated map ${\rm det}(T):{\rm det}(\bb{V})\to{\rm det}(\bb{W})$ which in the special case $\bb{V}= \bb{W}$
 coincides with the multiplication by the determinant of the endomorphism $T$.
 This determinant construction
is a functor from the category of vector spaces  to itself
and by a standard argument \cite[Chapter 5, Section 6]{husemoller-94} induces a functor on the category of complex vector bundles over an arbitrary space $X$. More precisely, for each rank $n$ complex vector bundle $\bb{E}\to X$, the associated \emph{determinant line bundle} ${\rm det}(\bb{E})\to X$ is the rank 1 complex vector bundle with fibers
\beql{eq:fib_descr}
 {\rm det}(\bb{E})_x\;=\; {\rm det}(\bb{E}_x)\qquad\quad x\in X \;.
\eeq
To each local trivializing frame of sections $\{s_1,\ldots,s_n\}$ of $\bb{E}$ over an open set $\f{U}\subset X$ one can associate the  section $s_1\wedge\ldots\wedge s_n$ which provides a trivialization of ${\rm det}(\bb{E})$
over the same $\f{U}$. For each map $\varphi:X\to Y$ one has the isomorphism ${\rm det}(\varphi^*(\bb{E}))\simeq \varphi^*({\rm det}(\bb{E}))$ which is a special case of the compatibility between pullback and tensor product.
Finally, if $\bb{E}=\bb{E}_1\oplus\bb{E}_2$ in the sense of 
Whitney, then ${\rm det}(\bb{E})={\rm det}(\bb{E}_1)\otimes {\rm det}(\bb{E}_2)$.

\medskip

Let $(\bb{E},\Theta)$ be a rank $2m$ $\rr{Q}$-bundle over
$(X,\tau)$. The associated determinant line bundle ${\rm det}(\bb{E})$ inherits an involutive structure given by the map ${\rm det}(\Theta)$ which acts \emph{anti}-linearly between the fibers ${\rm det}(\bb{E})_x$ and ${\rm det}(\bb{E})_{\tau(x)}$ according to ${\rm det}(\Theta)(p_1\wedge\ldots\wedge p_n)= \Theta(p_1)\wedge\ldots\wedge \Theta(p_n)$. Clearly ${\rm det}(\Theta)^2$ is a fiber preserving map which coincides with the multiplication by $(-1)^{2m}=1$. Therefore (\cf Remark \ref{rk_real}) one concludes that:
\begin{lemma}\label{lemma:R_Q_det_bun}
Let $(\bb{E},\Theta)$ be a rank $2m$ $\rr{Q}$-bundle over
 $(X,\tau)$. The associated determinant line bundle  ${\rm det}(\bb{E})$ endowed with  the involutive structure ${\rm det}(\Theta)$ is
 a \virg{Real} line bundle over $(X,\tau)$.
\end{lemma}

\medskip

Assume that the rank $2m$ $\rr{Q}$-bundle $(\bb{E},\Theta)$  over
 $(X,\tau)$ has  an equivariant Hermitian metric $\rr{m}$. The latter  fixes  a unique Hermitian metric $\rr{m}_{\rm det}$ on ${\rm det}(\bb{E})$ which is equivariant with respect to the induced $\rr{R}$-structure
${\rm det}(\Theta)$. More explicitly, if $(p_i,q_i)\in \bb{E}\times_\pi\bb{E}$, $i=1,\ldots,2m$ then,
$$
\rr{m}_{\rm det}\big(p_1\wedge\ldots\wedge p_{2m},q_1\wedge\ldots\wedge q_{2m}\big)\;:=\;\prod_{i=1}^{2m}\rr{m}(p_i,q_i)\;.
$$
The $\rr{R}$-line bundle $({\rm det}(\bb{E}),{\rm det}(\Theta))$ endowed with   $\rr{m}_{\rm det}$  is $\rr{R}$-trivial if and only if there exists an isometric $\rr{R}$-isomorphism with $X\times\C$, or equivalently, if and only if
 there exists a global $\rr{R}$-section $s:X\to {\rm det}(\bb{E})$ of unit length (\cf \cite[Theorem 4.8]{denittis-gomi-14}). Let us recall that an $\rr{R}$-section meets the condition ${\rm det}(\Theta)\circ s\circ\tau=s$.
 Let
 $$
 \n{S}\big({\rm det}(\bb{E})\big)\;:=\;\left\{p\in{\rm det}(\bb{E})\ |\  \rr{m}_{\rm det}(p,p)=1\right\}
 $$
 be the  \emph{circle $\rr{R}$-bundle} underlying to  $({\rm det}(\bb{E}),{\rm det}(\Theta))$.
  Then the $\rr{R}$-triviality of ${\rm det}(\bb{E})$ can be rephrased as the existence of a global $\rr{R}$-section $\n{S}({\rm det}(\bb{E}))\to X$.
One has the following important result.
\begin{proposition}[{\cite[Lemma 3.3]{denittis-gomi-14-gen}}]
\label{lemma:R_Q_det_bun2}
Let $(\bb{E},\Theta)$ be a $\rr{Q}$-bundle over a space $X$ with trivial involution $\tau={\rm Id}_X$. Then, the associated  determinant line bundle ${\rm det}(\bb{E})$
endowed with the \virg{Real} structure ${\rm det}(\Theta)$
 is $\rr{R}$-trivial and admits a unique  {canonical} trivialization $h_{\rm can}:{\rm det}(\bb{E})\to X\times\C$ compatible with the \virg{Real} structure:
 $$
 \big(h_{\rm can}\circ {\rm det}(\Theta)\big)(p)\,=\,\overline{h_{\rm can}(p)}\,\qquad\quad \forall\ p\in{\rm det}(\Theta)\;.
 $$ 
 This trivialization fixes a unique
  \emph{canonical
 $\rr{R}$-section} $s_{\rm can}:X\to\n{S}({\rm det}(\bb{E}))$
 defined by
 $$
 s_{\rm can}(x)\;:=\;h_{\rm can}^{-1}(x,1)\,\qquad\quad \forall\ x\in X\;.
 $$
 \end{proposition}

\medskip

\noindent
If $X^\tau\neq \empt$   the restricted vector bundle $\bb{E}|_{X^\tau}\to{X^\tau}$ can be seen as a $\rr{Q}$-bundle over a space with trivial involution. Preposition \ref{lemma:R_Q_det_bun2} assures that the restricted line bundle ${\rm det}(\bb{E}|_{X^\tau})$
is $\rr{R}$-trivial with respect to the restricted \virg{Real} structure ${\rm det}(\Theta|_{X^\tau})$ and admits a distinguished $\rr{R}$-section
\begin{equation}\label{eq:triv_2}
s_{\bb{E}}\;:\;X^\tau\;\to\;\n{S}\big({\rm det}(\bb{E})|_{X^\tau}\big)\;
\end{equation}
which  will be called  the \emph{canonical section} over $X^\tau$ associated to $(\bb{E},\Theta)$.

\subsection{The FKMM-invariant and related properties}
\label{subsec:fkmm-invariant}
In this section we recall the construction and the main properties of the \emph{FKMM}-invariant. For more details on this argument we refer to \cite{denittis-gomi-14-gen,denittis-gomi-16} and references therein.

\medskip

Let $(X,\tau)$ be an involutive space and $Y\subseteq X$ a closed $\tau$-invariant subspace $\tau(Y)=Y$ ({it is not required that $Y\subseteq X^\tau$}). Consider pairs $(\bb{L}, s)$ consisting of: (a) a \virg{Real} line bundle $\bb{L}\to X$ with a given \virg{Real} structure $\Theta$ and a Hermitian metric $\rr{m}$; (b) a \virg{Real} section $s:Y\to \n{S}(\bb{L}|_{Y})$ of the circle bundle associated to the restriction $\bb{L}|_{Y}\to Y$ (or equivalently a trivialization $h:\bb{L}|_{Y}\to Y\times\C$). 
Two pairs $(\bb{L}_1, s_1)$ and $(\bb{L}_2, s_2)$ builded over the same involutive base space
$(X,\tau)$ and the same invariant subspace $Y$ are said \emph{isomorphic} if there is an $\rr{R}$-isomorphism of line bundles $f:\bb{L}_1\to \bb{L}_2$ (preserving the Hermitian structure) such that $f\circ s_1=s_2$.

\begin{definition}[Relative \virg{Real} Picard group]
For an involutive space $(X, \tau)$ and a closed $\tau$-invariant subspace $Y \subseteq X$, we define ${\rm Vec}_{\rr{R}}^1(X|Y,\tau)$ to be the abelian group of the isomorphism classes of pairs $(\bb{L}, s)$, with  group structure  given by the tensor product
$$
(\bb{L}_1, s_1)\; \otimes\; (\bb{L}_2, s_2)\; \simeq\; (\bb{L}_1\otimes \bb{L}_1, s_1  \otimes s_2).
$$
\end{definition}

\medskip

\noindent
The relative \virg{Real} Picard group can be described in terms of the relative equivariant cohomology of the pair $Y\subseteq X$. A short introduction about the equivariant cohomology is provided in Appendix \ref{subsec:borel_cohom}.
\begin{proposition}[{\cite[Proposition 2.7]{denittis-gomi-16}}]
\label{prop:nat_iso}
There is a natural isomorphism of abelian groups
$$
\tilde{\kappa}\;:\;{\rm Vec}_{\rr{R}}^1(X|Y,\tau)\;\stackrel{\simeq}{\longrightarrow}\;H^2_{\Z_2}\big(X|Y,\Z(1)\big)\;.
$$
\end{proposition}

\medskip
By combining Proposition \ref{lemma:R_Q_det_bun2}
and Proposition \ref{prop:nat_iso}
one  can define an \emph{intrinsic} invariant for the category of \virg{Quaternionic}  vector bundles.
\begin{definition}[Generalized FKMM-invariant]\label{def:gen_FKMM_inv}
Let $(\bb{E},\Theta)$ be an even-rank \virg{Quaternionic} vector bundle over the involutive space $(X,\tau)$
and consider the pair $({\rm det}(\bb{E}), s_{\bb{E}})$ where ${\rm det}(\bb{E})$ is the determinant line bundle
associated to $\bb{E}$ endowed with the \virg{Real} structure  ${\rm det}(\Theta)$ and $s_{\bb{E}}$ the canonical section  \eqref{eq:triv_2}. The  \emph{FKMM-invariant} of  $(\bb{E},\Theta)$ is the cohomology class $\kappa(\bb{E},\Theta)\in H^2_{\Z_2}\big(X|X^\tau,\Z(1)\big)$ defined by
$$
\kappa(\bb{E},\Theta)\;:=\;\tilde{\kappa}\big([({\rm det}(\bb{E}), s_{\bb{E}})]\big)
$$
where  $[({\rm det}(\bb{E}), s_{\bb{E}})]\in {\rm Vec}_{\rr{R}}^1(X|X^\tau,\tau)$ is the 
 isomorphisms class of the pair $({\rm det}(\bb{E}), s_{\bb{E}})$ and $\tilde{\kappa}$ is the group isomorphism described in Proposition \ref{prop:nat_iso}.
\end{definition}

\medskip

\noindent The FKMM-invariant
$$
{\kappa}\;:\;{\rm Vec}_{\rr{Q}}^{2m}(X,\tau)\;{\longrightarrow}\;H^2_{\Z_2}\big(X|X^\tau,\Z(1)\big)\;
$$
defined above meets the following properties:
\begin{itemize}
\item[($\alpha$)] Isomorphic $\rr{Q}$-bundles define the same FKMM-invariant;
\vspace{1mm}
\item[($\beta$)] The FKMM-invariant is \emph{natural} under the pullback induced by equivariant maps;
\vspace{1mm}
\item[($\gamma$)] If $(\bb{E},\Theta)$ is $\rr{Q}$-trivial then $\kappa(\bb{E},\Theta)=0$; 
\vspace{1mm}
\item[($\delta$)] The FKMM-invariant is \emph{additive} with respect to the Whitney sum and the abelian structure of $H^2_{\Z_2}(X|X^\tau,\Z(1))$. More precisely
$$
\kappa(\bb{E}_1\oplus \bb{E}_2,\Theta_1\oplus\Theta_2)\;=\;\kappa(\bb{E}_1,\Theta_1) + \kappa( \bb{E}_2,\Theta_2)
$$
for each pair of $\rr{Q}$-bundles $(\bb{E}_1,\Theta_1)$ and $(\bb{E}_2,\Theta_2)$ over the same involutive space $(X,\tau)$;
\vspace{1mm}
\item[($\varepsilon$)] The FKMM-invariant is the image under the pullback induced by the classifying map of a \emph{universal} FKMM-invariant.
\end{itemize}
Properties ($\alpha$)-($\delta$) follow immediately from the definition. Property ($\varepsilon$) requires some work. First of all one has to define the {universal} FKMM-invariant over the \emph{universal $\rr{Q}$-bundle} which provides the homotopy classification \eqref{eq:homotopy_classification}. Then one has to prove that the FKMM-invariant, as introduced in Definition \ref{def:gen_FKMM_inv} coincides with the pullback induced by  $\varphi\in [X, \hat{G}_{2m}(\C^{\infty})]_{\Z_2}$
of the universal invariant. In this work we will never use  property ($\varepsilon$)
and the  interested reader is referred to \cite[Section 2.6]{denittis-gomi-16}.

\medskip

In the free-involution case $X^\tau=\empt$ one has that the relative cohomology group $H^2_{\Z_2}(X|X^\tau,\Z(1))$ reduces to  $H^2_{\Z_2}(X,\Z(1))$. Moreover, as a consequence of the {Kahn's isomorphism} \eqref{eq:kahn_iso}
one has the following result \cite[Corollary 2.12.]{denittis-gomi-16}:
\begin{proposition}\label{corol:II}
Let $(\bb{E},\Theta)$ be an even rank \virg{Quaternionic} vector bundle over the involutive space $(X,\tau)$. If $X^\tau=\empt$, then the FKMM-invariant $\kappa(\bb{E},\Theta)$ agrees with the first {\virg{Real} Chern class}
$c^{\rr{R}}_1({\rm det}(\bb{E}))$ of the associated determinant line bundle.
\end{proposition}

\medskip

A second interesting case concerns involutive spaces with only a finite number of fixed points.
The next definition encloses a large
class of interesting   involutive spaces.
\begin{definition}[FKMM-space {\cite{denittis-gomi-14-gen}}]\label{defi:FKMM-space}
Let $(X,\tau)$ be an involutive space which meets Assumption \ref{ass:top}. We say that  $(X,\tau)$ is an \emph{FKMM-space} if:
\begin{itemize}
\item[(a)] The fixed point set $X^\tau$ is not empty and consists of a finite number of points;
\vspace{1mm}
\item[(b)] $H^2_{\Z_2}\big(X,\Z(1)\big)=0$.
\end{itemize}
\end{definition}
\noindent
In the case that $(X,\tau)$ is an FKMM-space the following isomorphism 
\begin{equation}\label{eq:iso_for_fkmm}
H^2_{\Z_2}\big(X|X^\tau,\Z(1)\big)\;\simeq\;\big[X^\tau,\tilde{\n{U}}(1)\big]_{\Z_2}/\big[X,\tilde{\n{U}}(1)\big]_{\Z_2}\;\simeq\;{\rm Map}\big(X^\tau,\{\pm 1\}\big)/\big[X,\tilde{\n{U}}(1)\big]_{\Z_2}
\end{equation}
holds true \cite[Lemma 3.1]{denittis-gomi-14-gen}. 
Here $\tilde{\n{U}}(1)$ denotes the unitary group ${\n{U}}(1)$ endowed with the involution induced by the complex conjugation and
$[X,\tilde{\n{U}}(1)]_{\Z_2}$ is the set of classes of $\Z_2$-homotopy equivalent equivariant maps between  $(X,\tau)$ and the space $\tilde{\n{U}}(1)$. The action of $[X,\n{U}(1)]_{\Z_2}$ on 
$[X^\tau,\n{U}(1)]_{\Z_2}$ is given by the pointwise multiplication followed by the restriction to $X^\tau$. The second isomorphism is justified by 
$$
\big[X^\tau,\tilde{\n{U}}(1)\big]_{\Z_2}\;=\;\big[X^\tau,\pm1\big]\;=\;{\rm Map}\big(X^\tau,\{\pm 1\}\big)\;\simeq\;\{\pm 1\}^{|X^\tau|}
$$
where ${\rm Map}\big(X^\tau,\{\pm 1\}\big)$ is the set of the maps from $X^\tau$ to $\{\pm 1\}$ and $|X^\tau|\in\N$ is the cardinality of the finite set $X^\tau$. 
By combining Definition \ref{def:gen_FKMM_inv} with the isomorphism \eqref{eq:iso_for_fkmm} one concludes that:
\begin{proposition}\label{prop:fkmm-inv_fkmm-space}
Let $(\bb{E},\Theta)$ be a  \virg{Quaternionic} vector bundle over the FKMM-space $(X,\tau)$.  Then, the FKMM-invariant $\kappa(\bb{E},\Theta)$ can be represented by
$$
[s_{\bb{E}}]\;\in\;{\rm Map}\big(X^\tau,\{\pm 1\}\big)/\big[X,\tilde{\n{U}}(1)\big]_{\Z_2}
$$
where $s_{\bb{E}}$ is the canonical section \eqref{eq:triv_2}.
\end{proposition}
\noindent
In other words over an FKMM-space the FKMM-invariant $\kappa(\bb{E},\Theta)$
 is given by the canonical section $s_{\bb{E}}$ modulo the action (multiplication and restriction) of an equivariant map $s:X\to\tilde{\n{U}}(1)$ which can be seen as a global trivialization of the trivial line bundle 
${\rm det}(\bb{E})$. The result in Proposition \ref{prop:fkmm-inv_fkmm-space} agrees with the \emph{old} definition of FKMM-invariant given in \cite[Definition 3.2]{denittis-gomi-14-gen}. Proposition \ref{prop:fkmm-inv_fkmm-space} says that the 
FKMM-invariant $\kappa(\bb{E},\Theta)$ of a $\rr{Q}$-bundle $(\bb{E},\Theta)$ over an FKMM-space can be described by any representative $\phi\in \text{Map}(X^\tau,\{\pm 1\})$ in the class of $[s_{\bb{E}}]$. This description allows to compare the FKMM-invariant $\kappa$ with the 
\emph{Fu-Kane-Mele indices} in the case of the involutive spheres $\n{S}^{1,d}$
and the involutive tori $\n{T}^{0,d,0}$, $d=1,2,3,4$, as discussed in \cite{denittis-gomi-14-gen}.

\medskip

Finally, we describe the injectivity of the FKMM invariant in low dimension, which is the source of the \virg{half} of the Theorem \ref{theo:A_inject_fixed}.

\begin{proposition}{\cite[Theorem 4.7]{denittis-gomi-16}}
\label{prop:injectivity_in_low_dimension}
Let $(X,\tau)$ be a low-dimensional involutive space in the sense of Assumption \ref{ass:top}. Then the FKMM-invariant 
$$
{\kappa}\;:\;{\rm Vec}_{\rr{Q}}^{2m}(X,\tau)\;{\longrightarrow}\;H^2_{\Z_2}\big(X|X^\tau,\Z(1)\big)\;
$$
is injective for all $m\in\N$.
\end{proposition}

\section{The surjectivity of the FKMM-invariant in low-dimension}
\label{subsec:tech_results}
This technical section is devoted to the proof of various results that, all together, provide the justification of Theorem \ref{theo:class_low_odd_R}.
Many of the arguments that will be presented here are based on the fact that
the involutive space $(X,\tau)$ has the structure of a $\Z_2$-CW-complex, namely it admits a skeleton decomposition given by gluing cells of different dimensions which carry  $\Z_2$-actions. For a detailed presentation of the  theory of 
$\n{G}$-equivariant CW-complex we refer to  \cite{matumoto-71,allday-puppe-93}.
Here we recall some basic notions and we introduce the relevant notation (\cf \cite[Section 4.5]{denittis-gomi-14}) that will be used repeatedly  in the remainder of the section. 

\medskip

Let $(X,\tau)$ be a $\Z_2$-CW complex $X$ of dimension $d$. This means that $X$ consists of free and fixed $\Z_2$-cells of dimension at most $d$ glued together through equivariant attaching maps. A \emph{fixed} $\Z_2$-cell of dimension $k$ is given by the space $\boldsymbol{e}^k:=\{\ast\}\times\n{D}^k\simeq \n{D}^k$ where $\n{D}^k:=\{t\in\R^k\;|\; \|t\|\leqslant 1 \}$ denotes the closed unit ball (or disk) in $\R^k$
and the $\Z_2$-action on $\boldsymbol{e}^k$ is trivial on the two factors $\{\ast\}$ and $\n{D}^k$. A \emph{free} $k$-dimensional $\Z_2$-cell is described by 
$\tilde{\boldsymbol{e}}^k:=\Z_2\times\n{D}^k$ where the group $\Z_2$ is identified with the space $\{\pm 1\}$ endowed with the free flipping involution. The overall $\Z_2$-action on $\tilde{\boldsymbol{e}}^k$ is obtained by combining the free action on $\Z_2$ with the trivial action on $\n{D}^k$, namely it is given by
$(\pm1,t)\leftrightarrow (\mp1,t)$ for all $t\in\n{D}^k$
 As usual, let $X_k\subset X$, $k=0,1,\ldots,d$, be the $k$-skeleton of $X$.
In particular $X_{k-1}$ is the union of all the  $\Z_2$-cells of dimension less than $k$ and $X_{k}$ is
obtained by gluing a certain number of $k$-dimensional 
cells ${\boldsymbol{e}}^k_\lambda$ and $\tilde{\boldsymbol{e}}^k_\lambda$ to $X_{k-1}$:
\begin{equation}\label{eq:skeleton_caz}
X_k\;:=\;X_{k-1}\;\cup_{\phi_k}\;\left[\left(\bigsqcup_{\lambda\in\Lambda^k_{\text{fixed}}}\boldsymbol{e}^k_\lambda\right)\;\sqcup\;\left(\bigsqcup_{\lambda\in\Lambda^k_{\text{free}}}\tilde{\boldsymbol{e}}^k_\lambda\right)\right]\;.
\end{equation}
In the last expression $\Lambda_{\text{fixed}}^k$ and  $\Lambda_{\text{free}}^k$ are index sets (possibly infinite  in the case that $X$ is not a finite complex) and $\phi_k$ is   the equivariant gluing map which attaches the $k$ dimensional cells to the $(k-1)$-skeleton. The involutive space $(X,\tau)$ has a free $\Z_2$-action if and only if 
$\Lambda_{\text{fixed}}^k=\empt$ for all $k\leqslant d$.
{Finally, let us point out that for the very definition of CW-complex one always has that $X_0\neq\empt$.}

\subsection{Surjectivity  in $d=0,1$}
\label{app:surj_fkmm_d=2}
The aim of  this section is to  provide the proof for equation 
\eqref{eq:rel_equi_cohom_d=0,1}.
\begin{proposition}\label{prop:zero_rel_cohom_low}
Let $(X,\tau)$ be an involutive space which verifies Assumption \ref{ass:top}. Let its dimension be $d=0$ or $d=1$. Then
$$
H^{2}_{\Z_2}\big(X|X^\tau,\Z(1)\big)\;=\;0\;.
$$
\end{proposition}
\proof \emph{(Free case)} Let us start with the case $X^\tau=\empt$. An inspection of the exact sequence \cite[Proposition 2.3]{gomi-13}
\begin{equation}\label{eq:exact_app_C1}
\begin{aligned}
\ldots\;&\longrightarrow\;H^k_{\Z_2}\big(X,\Z\big)\;\stackrel{}{\longrightarrow}\;H^k\big(X,\Z\big)\;\stackrel{}{\longrightarrow}\;H^k_{\Z_2}\big(X,\Z(1)\big)\;\stackrel{}{\longrightarrow}\;H^{k+1}_{\Z_2}\big(X,\Z\big)\;\longrightarrow\;\ldots\;\\
\end{aligned}
\end{equation}
along with the vanishing of $H^k(X,\Z)=0$ and $H^k_{\Z_2}(X,\Z)=H^k(X/\Z_2,\Z)=0$ for all $k\geqslant 2$, assures that $H^k_{\Z_2}(X,\Z(1))=0$ for all $k\geqslant 2$. In this case the claim follows by observing that $H^{2}_{\Z_2}(X|X^\tau,\Z(1))=H^{2}_{\Z_2}(X,\Z(1))$ due to the assumption $X^\tau=\empt$.\\

\emph{(Non-free case)} Let assume now $X^\tau\neq\empt$ and consider first the case $d=0$. The full space splits as $X=X^\tau\sqcup X_{\text{free}}$
where $X^\tau$ is the disjoint union of $\tau$-invariant points $\{\ast\}$ and $X_{\text{free}}$ is a collection of free involution spaces $\Z_2$. Since $X^\tau\cap X_{\text{free}}=\empt$  the Mayer-Vietoris exact sequence  implies 
$$
H^k_{\Z_2}\big(X,\Z(1)\big)\;\simeq\;H^k_{\Z_2}\big(X^\tau,\Z(1)\big)\;\oplus\;H^k_{\Z_2}\big(X_{\text{free}},\Z(1)\big)\;\simeq\;H^k_{\Z_2}\big(X^\tau,\Z(1)\big)\;,\qquad\quad k=1,2
$$
where last isomorphism is justified by the exact sequence \eqref{eq:exact_app_C1}
along with
the fact that
$X_{\text{free}}$ is a $0$-dimensional space with a free involution. Recalling that  $H^k_{\Z_2}(\{\ast\},\Z(1))$  is trivial when $k$ is even (see \cite[Section 5.1]{denittis-gomi-14}) it follows by the additivity axiom in equivariant cohomology that $H^2_{\Z_2}(X,\Z(1))=0$.
In this particular case the exact sequence
\eqref{eq:Long_seq} reads
$$
H^1_{\Z_2}\big(X^\tau,\Z(1)\big)\;\stackrel{r}{\longrightarrow}\;H^1_{\Z_2}\big(X^\tau,\Z(1)\big)\;\stackrel{\delta_1}{\longrightarrow}\;H^2_{\Z_2}\big(X|X^\tau,\Z(1)\big)\;\stackrel{\delta_2}{\longrightarrow}\;0\;,
$$
with $r$, the restriction from $X$ to $X^\tau$, being evidently an isomorphism. Hence  $H^2_{\Z_2}(X|X^\tau,\Z(1))=0$.
For the case $d=1$ let us start with the exact sequence
\eqref{eq:Long_seq}
$$
H^1_{\Z_2}\big(X,\Z(1)\big)\;\stackrel{r}{\longrightarrow}\;H^1_{\Z_2}\big(X^\tau,\Z(1)\big)\;\stackrel{\delta_1}{\longrightarrow}\;H^2_{\Z_2}\big(X|X^\tau,\Z(1)\big)\;\stackrel{\delta_2}{\longrightarrow}\;H^2_{\Z_2}\big(X,\Z(1)\big)\;\stackrel{r}{\longrightarrow}\;H^2_{\Z_2}\big(X^\tau,\Z(1)\big)\;.
$$
 In order to prove the claim it is enough to show that: (i) the map $r:H^1_{\Z_2}(X,\Z(1))\to H^1_{\Z_2}(X^\tau,\Z(1))$ is surjective; (ii) the map $r:H^2_{\Z_2}(X,\Z(1))\to H^2_{\Z_2}(X^\tau,\Z(1))$ is injective.
To do that let $U$ be a closed neighborhood of $X^\tau$ such that: (a) $U$ is $\Z_2$-equivariantly homotopy equivalent to $X^\tau$; (b) Let $V:=\overline{X\setminus U}$. Then $U\cap V$ is $\Z_2$-equivariantly homotopy equivalent to the disjoint union of several copies of the free involutive space $\Z_2$.  The reason why  such a  closed neighborhood $U$ exists relies on the following construction: The fixed point set $X^\tau$ consists
of all the fixed $\Z_2$-cells of dimension $0$ and $1$. The $0$-dimensional free $\Z_2$-cells in $X$ are disjoint to $X^\tau$, and $X$ is obtained by gluing the remaining $1$-dimensional 
free $\Z_2$-cells to $X^\tau$ or to the  $0$-dimensional free $\Z_2$-cells. The typical $1$-dimensional  free $\Z_2$-cell has the form $\tilde{\boldsymbol{e}}^1:=\Z_2\times[0,1]$ endowed with  the free $\Z_2$-action implemented by $(\pm1,t)\leftrightarrow (\mp1,t)$ for all $t\in[0,1]$. The $\Z_2$ boundary of $\tilde{\boldsymbol{e}}^1$ is $\partial\tilde{\boldsymbol{e}}^1:=\Z_2\times\{0,1\}$. If $\partial\tilde{\boldsymbol{e}}^1\cap X^\tau=\empt$, then $\tilde{\boldsymbol{e}}^1$ is attached to a $0$-dimensional free $\Z_2$-cell and forms part of a free involutive space which does not intersect $X^\tau$. In this case we can take $U$ such that $U\cap \tilde{\boldsymbol{e}}^1=\empt$. If  $\partial\tilde{\boldsymbol{e}}^1$ intersects $X^\tau$ in a single point we can assume,
without loss of generality, that the attachment
 is along the $\Z_2$-left boundary $\{\pm 1\}\times\{0\}\in\tilde{\boldsymbol{e}}^1$ and we can take $\{\pm 1\}\times[0,\nicefrac{1}{3}]$ as the intersection between  $\tilde{\boldsymbol{e}}^1$ and $U$. Finally, if $\partial\tilde{\boldsymbol{e}}^1\subset X^\tau$, namely if the $\Z_2$ boundary of $\tilde{\boldsymbol{e}}^1$ is attached to $X^\tau$, we can take $\{\pm 1\}\times([0,\nicefrac{1}{3}]\cup[\nicefrac{2}{3},1])$ as the intersection between  $\tilde{\boldsymbol{e}}^1$ and $U$. The proof of the property (i) follows by inspecting the Mayer-Vietoris exact sequence
$$
\ldots\stackrel{}{\longrightarrow}\;H^1_{\Z_2}\big(X,\Z(1)\big)\;\stackrel{(\imath_U^*,\imath_V^*)}{\longrightarrow}\;H^1_{\Z_2}\big(U,\Z(1)\big)\oplus H^1_{\Z_2}\big(V,\Z(1)\big)\;\stackrel{}{\longrightarrow}\;H^1_{\Z_2}\big(U\cap V,\Z(1)\big)\;\stackrel{}{\longrightarrow}\ldots\;
$$
where $\imath_U^*$ and $\imath_V^*$ are the maps induced by the inclusion  $\imath_U:U\to X$ and $\imath_V:V\to X$. 
Lemma \ref{appA_lemma} assures that 
$H^j_{\Z_2}(\Z_2,\Z(1))=0$ for all $j\geqslant 1$ and by additivity it holds that 
 $H^1_{\Z_2}(U\cap V,\Z(1))=0$. Moreover, by construction, one has 
 $H^j_{\Z_2}(U,\Z(1))\simeq H^j_{\Z_2}(X^\tau,\Z(1))$ and under this identification the map $\imath_U^*$  can be replaced by the map
 $r:H^1_{\Z_2}(X,\Z(1))\to H^1_{\Z_2}(X^\tau,\Z(1))$. Putting these facts together, one obtains from the Mayer-Vietoris exact sequence that the map
 $$
 H^1_{\Z_2}\big(X,\Z(1)\big)\;\stackrel{(r,\imath_V^*)}{\longrightarrow}\;H^1_{\Z_2}\big(X^\tau,\Z(1)\big)\oplus H^1_{\Z_2}\big(V,\Z(1)\big)
 $$
is surjective, and in particular (i) is proved. In order to prove (ii) we need a second part of the  Mayer-Vietoris exact sequence
$$
0\;\stackrel{}{\longrightarrow}H^2_{\Z_2}\big(X,\Z(1)\big)\;\stackrel{(\imath_U^*,\imath_V^*)}{\longrightarrow}\;H^2_{\Z_2}\big(U,\Z(1)\big)\oplus H^2_{\Z_2}\big(V,\Z(1)\big)\;\stackrel{}{\longrightarrow}\;0
$$
where we already plugged in $H^1_{\Z_2}(U\cap V,\Z(1))=H^2_{\Z_2}(U\cap V,\Z(1))=0$.
Since $V$ is a free involutive space of dimension at most 1 we already know  that $H^2_{\Z_2}(V,\Z(1))=0$. Moreover, in view of the isomorphism $H^j_{\Z_2}(U,\Z(1))\simeq H^j_{\Z_2}(X^\tau,\Z(1))$ we finally get
$$
0\;\stackrel{}{\longrightarrow}H^2_{\Z_2}\big(X,\Z(1)\big)\;\stackrel{(r,\imath_V^*)}{\longrightarrow}\;H^2_{\Z_2}\big(X^\tau,\Z(1)\big)\oplus 0\;\stackrel{}{\longrightarrow}\;0
$$
which means that $r:H^2_{\Z_2}(X,\Z(1))\to H^2_{\Z_2}(X^\tau,\Z(1))$ is an isomorphism,  therefore injective as required by (ii).
 \qed

\medskip

In view of Proposition \ref{prop:injectivity_in_low_dimension}, we immediately get the following corollary to Proposition \ref{prop:zero_rel_cohom_low}  which proves Theorem \ref{theo:A_inject_fixed} for $d = 0, 1$.

\begin{corollary} \label{cor:bijectivity_for_d_1_and_d_0}
Let $(X,\tau)$ by an involutive space which verifies Assumption \ref{ass:top}. Let its dimension be $d= 0, 1$. 
Then the FKMM-invariant
$$
\kappa\;:\; {\rm Vec}_{\rr{Q}}^{2m}\big(X,\tau\big)\;\stackrel{\simeq}{\longrightarrow}\; H^2_{\Z_2}\big(X| X^\tau,\Z(1)\big)   \;,\qquad\qquad \forall\ m\in\N \;
$$
induces a (trivial) bijection.
\end{corollary}

\subsection{The relative FKMM-invariant}
\label{subsec:relative_FKMM_inv}
In this section we want to construct a \virg{relative} version of the FKMM-invariant that will be useful to derive some crucial results in Section \ref{app:surj_fkmm_d=2_true}. Let us start with a definition.
\begin{definition}
Let $(X,\tau)$ be an involutive space and $Y\subseteq X$ a closed and $\tau$-invariant
subspace $\tau(Y)=Y$ ({it is not required that $Y\subset X^\tau$}). We denote by ${\rm Vec}_{\rr{Q}}^{2m}(X|Y,\tau\big)$ the set of isomorphism classes of pairs $(\bb{E},h)$ given by:
\begin{itemize}
\item[(a)] $(\bb{E},\Theta)$ is a \virg{Quaternionic} vector bundle of rank $2m$ over $(X,\tau)$ such that the restriction $\bb{E}|_Y$ is trivial in the category of  \virg{Quaternionic} vector bundles;
\vspace{1mm}
\item[(b)] $h:\bb{E}|_Y\to Y\times\C^{2m}$ is a \virg{Quaternionic} trivialization of $\bb{E}|_Y$.
\end{itemize}
Two pairs $(\bb{E},h)$ and $(\bb{E}',h')$ are called isomorphic if there exists 
a \virg{Quaternionic} isomorphism $f:\bb{E}\to\bb{E}'$ such that $h'\circ f|_Y=h$.
\end{definition}
\medskip

\noindent
The existence of a $\rr{Q}$-trivialization $h$ for $\bb{E}|_Y$ is equivalent to the existence of a $\rr{Q}$-frame, namely a set of nowhere vanishing section $(s_1,s_2, \ldots, s_{2m-1}, s_{2m})\subset\Gamma(\bb{E}|_Y)$ with the property  that $s_{2j}=\tau_\Theta(s_{2j-1})$
(\cf \cite[Theorem 2.1]{denittis-gomi-14-gen}). Since any $\rr{Q}$-bundle (on a paracompact space) admits an essentially unique invariant Hermitian metric \cite[Proposition 2.1]{denittis-gomi-14-gen} we can assume, without loss of generality, that the   $\rr{Q}$-frame $(s_1,s_2,\cdots,s_{2m-1},s_{2m})$ which represents the trivialization $h$ is orthonormal. In the case $Y=\empt$ one has the obvious identification 
${\rm Vec}_{\rr{Q}}^{2m}(X|\empt,\tau)\equiv{\rm Vec}_{\rr{Q}}^{2m}(X,\tau)$. There is also the natural map ${\rm Vec}_{\rr{Q}}^{2m}(X|Y,\tau)\to {\rm Vec}_{\rr{Q}}^{2m}(X,\tau\big)$ which forgets the trivialization over $Y$. 

\medskip

Given the involutive space $X$ and the closed $\tau$-invariant subspace $Y\subseteq X$ let $X/Y$ be the \virg{new} involutive space obtained by collapsing $Y$ to a single point $\{\ast\}$ and 
$$
\pi\;:\; X\;\longrightarrow\;X/Y
$$ 
the equivariant \emph{projection} which maps $\pi:Y\to\{\ast\}\in(X/Y)^\tau\subset X/Y$.
\begin{lemma}\label{lemma:collapsing_relativ}
The pullback under $\pi$ induces the natural bijection
$$
\pi^*\;:\;{\rm Vec}_{\rr{Q}}^{2m}(X/Y|\{\ast\},\tau)\;\stackrel{\simeq}{\longrightarrow}{\rm Vec}_{\rr{Q}}^{2m}(X|Y,\tau)\;.
$$
\end{lemma}
\proof[{Proof} (sketch of)]
In order to prove the claim it is enough to show that the pullback map $\pi^*$ admits an inverse map.
For a complex vector bundle $\bb{E}\to X$ with a trivialization $h:\bb{E}|_Y\to Y\times\C^{2m}$ there is a classical construction which associate a complex vector bundle
$\bb{E}/h$ over $X/Y$ and the isomorphism class of $\bb{E}/h$ only depends on the homotopy class of $h$ \cite[Lemma 1.4.7]{atiyah-67}. This construction can be naturally generalized to the \virg{Quaternionic} category in such a way that $\bb{E}/h$ inherits a $\rr{Q}$-structure $\Theta/h$ from $(\bb{E},\Theta)$. Moreover, under this construction the $\rr{Q}$-trivialization $h$ over $Y$ is mapped in a $\rr{Q}$-trivialization $\tilde{h}:\bb{E}/h|_{\{\ast\}}\to\C^{2m}$ over the fixed point $\{\ast\}$. The map $(\bb{E},h)\mapsto (\bb{E}/h,\tilde{h})$ provides the inverse of the pullback map $\pi^*$.
\qed

\medskip

Given an element $(\bb{E},h)\in {\rm Vec}_{\rr{Q}}^{2m}(X|Y,\tau)$ one can apply the determinant functor described in Section \ref{subsec:det_construct} in order to obtain a element $({\rm det}(\bb{E}),{\rm det}(\Theta))\in{\rm Pic}_{\rr{R}}(X,\tau)$ of the 
\virg{Real} Picard group. Moreover ${\rm det}(h):{\rm det}(\bb{E})|_Y\to Y\times\C$ provides an $\rr{R}$-trivialization over $Y$ which can be used to define the section $s_h:Y\to {\rm det}(\bb{E})|_Y$ through the  prescription $s_h(x):=h^{-1}(x,1)$.
It turns out that $s_h=s_1\wedge s_2 \wedge \cdots \wedge s_{2m-1} \wedge s_{2m}$ if $(s_1,s_2, \ldots, s_{2m-1}, s_{2m})$ is any orthonormal $\rr{Q}$-frame associated with the trivialization $h$. 
Let $s_{\bb{E}}:X^\tau\to{\rm det}(\bb{E})|_{X^\tau}$ be the canonical section of $\bb{E}$
described in Section \ref{subsec:det_construct}. By construction, on the intersection $Y\cap X^\tau$ one has that $s_h|_{Y\cap X^\tau}=s_{\bb{E}}|_{Y\cap X^\tau}$ by the unicity of the canonical section. This implies that the $\rr{R}$-section $s_h\cup s_{\bb{E}}$ is well defined on the union $Y\cup X^\tau$. After all these premises we are in a position to give the following:
\begin{definition}[Relative FKMM-invariant]
Consider the chains of maps
$$
{\rm Vec}_{\rr{Q}}^{2m}\big(X|Y,\tau\big)\;\stackrel{{\zeta}}{\longrightarrow}\;{\rm Vec}_{\rr{R}}^1(X|Y\cup X^\tau,\tau)\;\stackrel{\tilde{\kappa}}{\longrightarrow}\;H^2_{\Z_2}\big(X|Y\cup X^\tau,\Z(1)\big)\;
$$
where $\zeta$ is defined by $\zeta: (\bb{E},h)\mapsto ({\rm det}(\bb{E}),s_h\cup s_{\bb{E}})$ and the isomorphism $\tilde{\kappa}$ is described in Proposition
\ref{prop:nat_iso}. The \emph{relative FKMM-invariant} of an element $(\bb{E},h)$ in ${\rm Vec}_{\rr{Q}}^{2m}(X|Y,\tau)$
is the cohomology class $\kappa_{\rm rel}[(\bb{E},h)]$ in $H^2_{\Z_2}(X|Y\cup X^\tau,\Z(1))$ obtained as
$$
\kappa_{\rm rel}[(\bb{E},h)]\;:=\;(\tilde{\kappa}\circ\zeta)[(\bb{E},h)]\;.
$$
\end{definition}

\medskip

\noindent
The map $f$ which forgets the role of the space $Y$ gives rise to the commutative diagram
\beql{eq:forg_Y1}
\begin{diagram}
  {\rm Vec}_{\rr{Q}}^{2}\big(X|Y,\tau\big)        &   \rTo^f                &   {\rm Vec}_{\rr{Q}}^{2}\big(X,\tau\big)\\
       \dTo^{\kappa_{\rm rel}}     &   &   \dTo_{\kappa} \\
H^2_{\Z_2}\big(X|Y\cup X^\tau,\Z(1)\big)&       \rTo           & H^2_{\Z_2}\big(X| X^\tau,\Z(1)\big)     
\end{diagram}\;.
\eeq
If $Y\neq\empt$, by naturality and using Lemma \ref{lemma:collapsing_relativ}, one obtains the commutative diagram
\beql{eq:forg_Y2}
\begin{diagram}
  {\rm Vec}_{\rr{Q}}^2(X/Y|\{\ast\},\tau)       &   \rTo^{\pi^*}_{\simeq}                &{\rm Vec}_{\rr{Q}}^{2}\big(X|Y,\tau\big)        &   \rTo^{f}                &   {\rm Vec}_{\rr{Q}}^{2}\big(X,\tau\big)\\
  \dTo^{\kappa^\pi_{\rm rel}}    &&  \dTo^{\kappa_{\rm rel}}     &   &   \dTo_{\kappa} \\
H^2_{\Z_2}\big(X/Y|\{\ast\}\cup (X/Y)^\tau,\Z(1)\big)&       \rTo^{\pi^*}_{\simeq}             &H^2_{\Z_2}\big(X|Y\cup X^\tau,\Z(1)\big)&       \rTo            & H^2_{\Z_2}\big(X| X^\tau,\Z(1)\big)     
\end{diagram}\;.
\eeq
Also the homomorphism $\pi^*$ among the cohomology groups in the bottom row  turns out to be bijective as well. This fact can be proven, for example, along the same lines of Lemma \ref{lemma:collapsing_relativ} using the geometric representation for  the relative cohomology groups described in
Proposition
\ref{prop:nat_iso}.

\subsection{Surjectivity  in $d=2$}
\label{app:surj_fkmm_d=2_true}
In this section we will prove that the FKMM-invariant is surjective in dimension $d=2$. 
We will prove this fact first for a very special type of involutive space and then we will extend the proof to general spaces with a $\Z_2$-CW complex structure of dimension $d=2$.

\medskip

Let $X_*^N$ be a $\Z_2$-CW complex with the $0$-skeleton made by a single invariant point
$(X_*^N)_0=\{\ast\}$, the $1$-skeleton which agrees with the $0$-skeleton (absence of $1$-cells) and the $2$-skeleton (\ie the full space) obtained by gluing $N$ free $\Z_2$-cells $\tilde{\boldsymbol{e}}^2_\lambda$, $N\in\N\cup\{\infty\}$, to $\{\ast\}$.
The skeleton decomposition of $X_*^N$ is 
\begin{equation}\label{eq_decomp_special_X}
X_*^N\;:=\;\{\ast\}\;\cup_{\phi_2}\;\left(\bigsqcup_{\lambda=1}^N\tilde{\boldsymbol{e}}^2_\lambda\right)\;\simeq\;\bigvee_{\lambda=1}^N\tilde{\boldsymbol{e}}^2_\lambda
\;\simeq\;
\bigvee_{\lambda=1}^N\tilde{\n{V}}^2_\lambda
\end{equation}
where $\vee$ denotes the \emph{wedge sum} (or \emph{one-point union}) of a family of topological spaces. The last identification in \eqref{eq_decomp_special_X}  follows by observing that 
$\tilde{\n{V}}_\lambda^2:=\tilde{\boldsymbol{e}}^2_\lambda/\partial\tilde{\boldsymbol{e}}^2_\lambda\simeq\n{S}^2\vee\n{S}^2$ is the wedge sum of two spheres endowed with the $\Z_2$-action which flips the two spheres and fixes only the joining point $\{\ast\}$. Clearly $X_*^N$ has a fixed point set  made by the unique invariant point $(X_*^N)^\tau=(X_*^N)_0=\{\ast\}$.

\medskip

Let $\imath_\lambda:\tilde{\n{V}}_\lambda^2\to X_*^N$ be the inclusion map. 
By means of the additivity of the equivariant cohomology, one has the commutative diagram 
\begin{equation}\label{diag:surj_d=2}
\begin{diagram}
{\rm Vec}_{\rr{Q}}^{2m}\big(X_*^N,\tau\big)       &   
\rTo^{(\imath_\lambda^*)}_{\simeq}                &   
\prod_{\lambda=1}^N{\rm Vec}_{\rr{Q}}^{2m}\big(\tilde{\n{V}}_\lambda^2\big) \\
       &   &   \\
   \dTo_{\kappa}     &   & \dTo^{(\kappa_{\lambda})}_{\simeq  }     \\
        &   &   \\
H^2_{\Z_2}\big(X_*^N| \{\ast\},\Z(1)\big)     &       
\rTo^{(\imath_\lambda^*)}_{\simeq}             & 
\prod_{\lambda=1}^NH^2_{\Z_2}\big(\tilde{\n{V}}_\lambda^2|\{\ast\},\Z(1)\big)
\end{diagram}\;
\end{equation}
where we used the symbol  of the product $\prod$ instead the of the sum $\oplus$ for the equivariant cohomology since, in principle, $X_*^N$ is not a finite $CW$-complex meaning that $N$ can be infinite. The bijectivity of the horizontal arrow in the bottom of the diagram follows from the observation that 
\begin{equation}\label{eq:lemma_single_ast01}
H^2_{\Z_2}\big(\tilde{\n{V}}_\lambda^2|\{\ast\},\Z(1)\big)\;\equiv\; 
\tilde{H}^2_{\Z_2}\big(\tilde{\n{V}}_\lambda^2,\Z(1)\big)\;\simeq\; {H}^2_{\Z_2}\big(\tilde{\n{V}}_\lambda^2,\Z(1)\big)
\end{equation}
where $\tilde{H}^\bullet_{\Z_2}$ denotes the reduced equivariant cohomology (\cf Appendix \ref{subsec:borel_cohom}) and the 
second isomorphism  comes from the combination of the splitting \eqref{eq:app_red_cohom1} and the computation \eqref{eq:app_red_cohom2} which provides
$$
{H}^2_{\Z_2}\big(\tilde{\n{V}}_\lambda^2,\Z(1)\big)\;\simeq\;\tilde{H}^2_{\Z_2}\big(\tilde{\n{V}}_\lambda^2,\Z(1)\big)\;\oplus\;{H}^2_{\Z_2}\big(\{\ast\},\Z(1)\big)\;\simeq\;\tilde{H}^2_{\Z_2}\big(\tilde{\n{V}}_\lambda^2,\Z(1)\big)\;.
$$
 Since the reduced equivariant cohomology of a 
wedge sum of spaces equals the product of the reduced equivariant cohomology of each
space (just an application of the Mayer-Vietoris exact sequence) one obtains
$$
\prod_{\lambda=1}^NH^2_{\Z_2}\big(\tilde{\n{V}}_\lambda^2|\{\ast\},\Z(1)\big)\;\equiv\;\prod_{\lambda=1}^N\tilde{H}^2_{\Z_2}\big(\tilde{\n{V}}_\lambda^2,\Z(1)\big)\;\simeq\;
\tilde{H}^2_{\Z_2}\left(\bigvee_{\lambda=1}^N\tilde{\n{V}}_\lambda^2,\Z(1)\right)\;\equiv\;H^2_{\Z_2}\big(X_*^N|\{\ast\},\Z(1)\big)\;,
$$
namely the bijectivity of the horizontal arrow in the bottom of the diagram \eqref{diag:surj_d=2}. The bijectivity of the other two arrows is proved in the following lemma.
\begin{lemma}
With reference to the commutative diagram \eqref{diag:surj_d=2} it holds that:
\begin{itemize}
\item[(1)] The map 
$$
(\imath_\lambda^*)\;:\;  {\rm Vec}_{\rr{Q}}^{2}\big(X_*^N,\tau\big)\;\longrightarrow\;\prod_{\lambda=1}^N{\rm Vec}_{\rr{Q}}^{2}\big(\tilde{\n{V}}_\lambda^2\big)        
$$
is bijective;
\vspace{1.0mm}
\item[(2)] The FKMM-invariant
$$
\kappa\;:\;{\rm Vec}_{\rr{Q}}^{2}\big(\tilde{\n{V}}^2\big)
    \;\longrightarrow\;  H^2_{\Z_2}\big(\tilde{\n{V}}^2|\{\ast\},\Z(1)\big)
$$
is bijective;
\end{itemize}
\end{lemma}
\proof
(1) By using the skeleton decomposition of $X_*^N$ in equation \eqref{eq_decomp_special_X} (first equality)
we can express $X_*^N$ as the union $X_*^N =U_1\cup U_2$ where $U_1$ and $U_2$
are  two
closed invariant subspaces which satisfy the   $\Z_2$-equivariant
homotopy equivalences 
\begin{equation}\label{eq:U_1_U_2_d=2_genesis}
U_1\simeq \{\ast\}\;,\qquad U_2\;\simeq\;\bigsqcup_{\lambda=1}^N\tilde{\boldsymbol{e}}^2_\lambda\;\simeq\;\bigsqcup_{\lambda=1}^N\Z_2\;,\qquad U_1\cap U_2\;\simeq\;\bigsqcup_{\lambda=1}^N\partial\tilde{\boldsymbol{e}}^2_\lambda\;\simeq\;\bigsqcup_{\lambda=1}^N\big(\Z_2\times \n{S}^{1}\big)\;.
\end{equation}
More precisely we can take $U_1$ as a \virg{small} closed invariant neighborhood of $X^\tau=\{\ast\}$ in $X$ and $U_2$ as the complement of the interior of $U_1$. 
 As a consequence of Lemma \ref{appA_lemma}, the additivity of the equivariant cohomology and $H^2(\n{S}^1,\Z)=0$ it holds   that $H^2_{\Z_2}(U_2,\Z(1))=H^2_{\Z_2}(U_1\cap U_2,\Z(1))=0$. Then, the absence of non-trivial $\rr{Q}$-bundles in dimension 0 and
 the injectivity of the FKMM-invariant in low-dimension (Proposition \ref{prop:injectivity_in_low_dimension}) imply that
 $$
 {\rm Vec}_{\rr{Q}}^{2}\big(U_1\big)\;=\;  {\rm Vec}_{\rr{Q}}^{2}\big(U_2\big)\;=\;  {\rm Vec}_{\rr{Q}}^{2}\big(U_1\cap U_2\big)\;=\; 0\;.
 $$
The \virg{Quaternionic} version of the clutching construction assures that any rank 2
$\rr{Q}$-bundle over $X^N_*$ can be constructed by gluing together the trivial $\rr{Q}$-bundles over $U_1$ and $U_2$ by means of a clutching function 
$g: U_1\cap U_2\to \n{U}(2)$ with the equivariant property $g(\tau(x))=\mu(g(x))$ where
$\mu:\n{U}(2)\to\n{U}(2)$ is the involution given by 
\begin{equation}\label{eq:mu-invol}
\mu(U)\;:=\;-Q\;\overline{U}\;Q\;,\qquad\quad U\in \n{U}(2)
\end{equation}
where $Q$ is the $2\times 2$ matrix \eqref{eq:Q-mat}.
Since $\Z_2$-equivariantly homotopy equivalent clutching functions generate  
isomorphic $\rr{Q}$-bundles, one obtains
\begin{equation}\label{eq:lemma_bor_000}
{\rm Vec}_{\rr{Q}}^{2}\big(X_*^N,\tau\big)\;\simeq\;
\big[U_1\cap U_2,\hat{\n{U}}(2)\big]_{\Z_2}
\end{equation}
where the first isomorphism is due to Proposition \ref{theo:stab_ran_Q_even} and $\hat{\n{U}}(2)$ denotes the space ${\n{U}}(2)$ endowed with the involution $\mu$.
Let us remark that the second isomorphism in \eqref{eq:lemma_bor_000} is correct  in view of the vanishings $[U_1,\hat{\n{U}}(2)]_{\Z_2}=0$ and $[U_2,\hat{\n{U}}(2)]_{\Z_2}=0$ which ensure
that the left and right  coset actions on $[U_1\cap U_2,\hat{\n{U}}(2)]_{\Z_2}$ are trivial (compare \eqref{eq:lemma_bor_000} with the general type of isomorphism described in \cite[Lemma 4.18]{denittis-gomi-14}). The first vanishing is indeed justified by
$$
\big[U_1,\hat{\n{U}}(2)\big]_{\Z_2}\;\simeq\;\big[\{\ast\},\hat{\n{U}}(2)\big]_{\Z_2}\;\simeq\;\big[\{\ast\},S\n{U}(2)\big]\;\simeq\; 0
$$ 
being $S\n{U}(2)$ the fixed-point set of the involutive space $\hat{\n{U}}(2)$. The second vanishing follows from $[U_2,\hat{\n{U}}(2)]_{\Z_2}\simeq\prod_\lambda[\Z_2,\hat{\n{U}}(2)]_{\Z_2}$ along with
$$
\big[\Z_2,\hat{\n{U}}(2)\big]_{\Z_2}\;\simeq\; \big[\{\ast\},\n{U}(2)]\;\simeq 0\;.
$$
One also has 
$$
\big[\Z_2\times \n{S}^{1},\hat{\n{U}}(2)\big]_{\Z_2}\;\simeq\;
\big[ \n{S}^{1},{\n{U}}(2)\big]\;\stackrel{\text{det}}{\simeq}\;
\big[ \n{S}^{1},{\n{U}}(1)\big] \;\stackrel{\text{deg}}{\simeq}\;\Z
$$
where the second isomorphism is induced by the determinant and is equivalent to the well-known fact $\pi_1({\n{U}}(2))\simeq\pi_1({\n{U}}(1))$. The third isomorphism is given by the degree and coincides with $\pi_1({\n{U}}(1))\simeq\Z$. The last computation
along with \eqref{eq:lemma_bor_000} provides
\begin{equation}\label{eq:lemma_bor_01}
{\rm Vec}_{\rr{Q}}^{2}\big(X_*^N,\tau\big)\;\simeq\;\prod_{\lambda=1}^N\big[\Z_2\times \n{S}^{1},\hat{\n{U}}(2)\big]_{\Z_2}\;\simeq\;\prod_{\lambda=1}^N\big[ \n{S}^{1},{\n{U}}(1)\big]\;\simeq\;\Z^N\;.
\end{equation}
On the other hand let $U_{1,\lambda}$ and $U_{2,\lambda}$ be two closed invariant subspaces such that $\tilde{\n{V}}_\lambda^2=U_{1,\lambda}\cup U_{2,\lambda}$ and
$$U_{1,\lambda}\simeq \{\ast\}\;,\qquad
U_{2,\lambda}\simeq \tilde{\boldsymbol{e}}^2_\lambda\;,\qquad
U_{1,\lambda}\cap U_{2,\lambda}\simeq \partial\tilde{\boldsymbol{e}}^2_\lambda
$$
in analogy with \eqref{eq:U_1_U_2_d=2_genesis}.
The equivariant clutching construction applied to $\tilde{\n{V}}_\lambda^2$
  immediately provides
\begin{equation}\label{eq:lemma_bor_02}
{\rm Vec}_{\rr{Q}}^{2}\big(\tilde{\n{V}}_\lambda^2\big) \;\simeq\;\big[\Z_2\times \n{S}^{1},\hat{\n{U}}(2)\big]_{\Z_2}\;\simeq\;\big[ \n{S}^{1},{\n{U}}(1)\big] \;{\simeq}\;\Z
\end{equation}
for each $\lambda=1,\ldots,N$. By combining the expressions 
\eqref{eq:lemma_bor_01} and \eqref{eq:lemma_bor_02} one deduces the bijectivity of the 
(collection of) maps $(\imath_\lambda^*)$.\\

\medskip

(2) Equation \eqref{eq:lemma_single_ast01} 
says that the FKMM-invariant of a $\rr{Q}$-bundle over $\tilde{\n{V}}_\lambda^2$ can be computed by means of the \virg{Real} Chern class of the associate determinant line bundle. In fact one has 
\begin{equation}\label{eq:lemma_bor_03X}
{\rm det}\;:\;{\rm Vec}_{\rr{Q}}^{2m}\big(\tilde{\n{V}}_\lambda^2\big)\;\longrightarrow\; {\rm Pic}_{\rr{R}}\big(\tilde{\n{V}}_\lambda^2\big)\;\simeq\;{H}^2_{\Z_2}\big(\tilde{\n{V}}_\lambda^2,\Z(1)\big)\;.
\end{equation}
The {Kahn's isomorphism} \eqref{eq:kahn_iso}, along with Lemma \ref{appA_lemma} and equation \eqref{eq:app_red_cohom1}, assures that 
$$
{\rm Pic}_{\rr{R}}\big(U_{1,\lambda}\big)\;=\;  {\rm Pic}_{\rr{R}}\big(U_{2,\lambda}\big)\;=\;  {\rm Pic}_{\rr{R}}\big(U_{1,\lambda}\cap U_{2,\lambda}\big)\;=\; 0\;.
 $$ 
This allows us to construct $\rr{R}$-line bundles over $\tilde{\n{V}}_\lambda^2$ by means of the clutching construction providing a classification for ${\rm Pic}_{\rr{R}}(\tilde{\n{V}}_\lambda^2)$ of the type described by  \cite[Lemma 4.18]{denittis-gomi-14}. More precisely, let $\tilde{\n{U}}(1)$ be the unitary group endowed with the involution given by the complex conjugation. One has that 
$$
\big[U_{1,\lambda},\tilde{\n{U}}(1)\big]_{\Z_2}\;\simeq\;\big[\{\ast\},\tilde{\n{U}}(1)\big]_{\Z_2}\;\simeq\;\big[\{\ast\},\{\pm 1\}\big]\;\simeq\; \Z_2
$$ 
and 
$$
\big[U_{2,\lambda},\tilde{\n{U}}(1)\big]_{\Z_2}\;\simeq\;\big[\Z_2,\tilde{\n{U}}(1)\big]_{\Z_2}\;\simeq\;\big[\{\ast\},{\n{U}}(1)\big]\;\simeq\; 0\;.
$$ 
Then, by using the same construction of  \cite[Lemma 4.18]{denittis-gomi-14}, one has that 
\begin{equation}\label{eq:lemma_bor_06}
{\rm Pic}_{\rr{R}}\big(\tilde{\n{V}}_\lambda^2\big)\;\simeq\;\big[U_{1,\lambda}\cap U_{2,\lambda},\tilde{\n{U}}(1)\big]_{\Z_2}/\big[U_{1,\lambda},\tilde{\n{U}}(1)\big]_{\Z_2}\;\simeq\;\big[\Z_2\times\n{S}^1,\tilde{\n{U}}(1)\big]_{\Z_2}/\Z_2
\end{equation}
where the $\Z_2$-action on a class $[\varphi]\in [\Z_2\times\n{S}^1,\tilde{\n{U}}(1)]_{\Z_2}$ is given by the sign-flipping $[\varphi]\mapsto [-\varphi]$. Evidently,
$$
\big[\Z_2\times \n{S}^{1},\tilde{\n{U}}(2)\big]_{\Z_2}\;\simeq\;
\big[ \n{S}^{1},{\n{U}}(1)\big]\;\stackrel{\text{deg.}}{\simeq}\;\Z
$$
where the last isomorphism is provided by the degree map. Since the degree of the map $\varphi:\n{S}^{1}\to{\n{U}}(1)$ agrees with that of the map $-\varphi$ it follows that the $\Z_2$-action in \eqref{eq:lemma_bor_06} is trivial and so
\begin{equation}\label{eq:lemma_bor_07}
{\rm Pic}_{\rr{R}}\big(\tilde{\n{V}}_\lambda^2\big)\;\simeq\;\big[\Z_2\times\n{S}^1,\tilde{\n{U}}(1)\big]_{\Z_2}\;\simeq\;
\big[ \n{S}^{1},{\n{U}}(1)\big]\;\stackrel{\text{deg.}}{\simeq}\;\Z\;.
\end{equation}
A comparison between \eqref{eq:lemma_bor_02} and \eqref{eq:lemma_bor_07} shows that the 
determinant mapping \eqref{eq:lemma_bor_03X} can be rewritten as
\begin{equation}\label{eq:lemma_bor_03XX}
{\rm det}\;:\;\big[\Z_2\times \n{S}^{1},\hat{\n{U}}(2)\big]_{\Z_2}\;\stackrel{\simeq}{\longrightarrow}\;\big[\Z_2\times\n{S}^1,\tilde{\n{U}}(1)\big]_{\Z_2}.
\end{equation}
The isomorphism in \eqref{eq:lemma_bor_03XX} is a consequence of \cite[Lemma 4.19]{denittis-gomi-14}
which states that 
$$
\big[\Z_2\times \n{S}^{1},\hat{\n{U}}(2)\big]_{\Z_2}\;\simeq\;\big[\Z_2\times \n{S}^{1},\hat{\n{U}}(2)\cap S\n{U}(2)\big]_{\Z_2}\;\rtimes\;\big[\Z_2\times\n{S}^1,\tilde{\n{U}}(1)\big]_{\Z_2}
$$
where $\rtimes$ denotes the semidirect product. The vanishing of $[\Z_2\times \n{S}^{1},\hat{\n{U}}(2)\cap S\n{U}(2)]_{\Z_2}=0$, due to the fact that $\hat{\n{U}}(2)\cap S\n{U}(2)$ is the fixed point set of the involutive space $\hat{\n{U}}(2)$ and $\Z_2\times \n{S}^{1}$ is a free involution space, completes the argument for the isomorphism \eqref{eq:lemma_bor_03XX}. In conclusion the mapping \eqref{eq:lemma_bor_03X} turns out to be an isomorphism and this fact can be reformulated by saying that the FKMM-invariant
$$
\kappa_\lambda\;:\;{\rm Vec}_{\rr{Q}}^{2}\big(\tilde{\n{V}}_\lambda^2\big)\;\longrightarrow\; {H}^2_{\Z_2}\big(\tilde{\n{V}}_\lambda^2|\{\ast\},\Z(1)\big)\;\simeq\;{H}^2_{\Z_2}\big(\tilde{\n{V}}_\lambda^2,\Z(1)\big)
$$
provides a bijection.\qed

\medskip

The commutativity of the diagram \ref{diag:surj_d=2} immediately implies:
\begin{proposition}\label{prop:surg_2_di-1step}
Let $X_*^N$ be the involutive space described by \eqref{eq_decomp_special_X}. Then the FKMM-invariant
$$
\kappa\;:\; {\rm Vec}_{\rr{Q}}^{2}\big(X_*^N,\tau\big)\;\stackrel{\simeq}{\longrightarrow}\; H^2_{\Z_2}\big(X_*^N| \{\ast\},\Z(1)\big)   
$$
induces a bijection.
\end{proposition}

\medskip

The next steps are devoted to extending the validity of Proposition \ref{prop:surg_2_di-1step} to any
 $\Z_2$-CW complex 
 $(X,\tau)$ of dimension $d=2$. The quotient space $X/X_1$, obtained by collapsing the $1$-skeleton to a single invariant point $\{\ast\}$, still has the structure of a $\Z_2$-CW complex
given by the equivariant identification of the boundaries of  all the two-dimensional $\Z_2$-cells of $X$ with  $\{\ast\}$. More precisely one has that
\begin{equation}
\label{eq:surj_d=2_aca01}
X/X_1\;\simeq\;\left(\bigvee_{\lambda\in\Lambda^2_{\text{fixed}}}{\boldsymbol{e}}^2_\lambda\right)\vee\left(\bigvee_{\lambda\in\Lambda^2_{\text{free}}}\tilde{\boldsymbol{e}}^2_\lambda\right)\;\simeq\;\left(\bigvee_{\lambda\in\Lambda^2_{\text{fixed}}}\n{V}^2_\lambda\right)\vee\left(\bigvee_{\lambda\in\Lambda^2_{\text{free}}}\tilde{\n{V}}^2_\lambda\right)
\end{equation}
where the notation is the  same as used in \eqref{eq_decomp_special_X} and
  $\n{V}^2_\lambda:={\boldsymbol{e}}^2_\lambda/\partial{\boldsymbol{e}}^2_\lambda\simeq\n{S}^2$ is a two-dimensional sphere with trivial involution. The following result follows from Proposition \ref{prop:surg_2_di-1step}.
\begin{corollary}\label{corol:hrder_rememb}
Let $(X,\tau)$ be a $\Z_2$-CW complex of dimension $d=2$ and consider the quotient space $X/X_1$. Then the FKMM-invariant
$$
\kappa\;:\; {\rm Vec}_{\rr{Q}}^{2}\big(X/X_1,\tau\big)\;\stackrel{\simeq}{\longrightarrow}\; H^2_{\Z_2}\big(X/X_1| (X/X_1)^\tau,\Z(1)\big)   
$$
induces a bijection.
\end{corollary}
\proof
The main idea of the proof is to reduce the case to that considered in Proposition \ref{prop:surg_2_di-1step}. Let us define the subspace $Z_X\subset X/X_1$ defined as
$$
Z_X\;:=\;\bigvee_{\lambda\in\Lambda^2_{\text{free}}}\tilde{\boldsymbol{e}}^2_\lambda\;\simeq\; \bigvee_{\lambda\in\Lambda^2_{\text{free}}}\tilde{\n{V}}^2_\lambda\;.
$$
Clearly $Z_X$ is exactly of the type considered in Proposition \ref{prop:surg_2_di-1step}. In particular $Z_X^\tau=\{\ast\}$ and
$$
\kappa_Z\;:\; {\rm Vec}_{\rr{Q}}^{2}\big(Z_X,\tau\big)\;\stackrel{\simeq}{\longrightarrow}\; H^2_{\Z_2}\big(Z_X| \{\ast\},\Z(1)\big)  $$
is a bijection. Let us introduce two closed invariant subspaces $U_1$ and $U_2$ such that $U_1\cup U_2= X/X_1$
and
there are the following $\Z_2$-equivariant
homotopy equivalences 
\begin{equation}\label{eq:U_1_U_2_d=2_genesis_for_Z_X}
U_1\simeq Z_X\;,\qquad U_2\;\simeq\;\bigsqcup_{\lambda\in\Lambda^2_{\text{fixed}}}{\boldsymbol{e}}^2_\lambda
\;\simeq\;\bigsqcup_{\lambda\in\Lambda^2_{\text{fixed}}}\{\ast\}\;,\qquad U_1\cap U_2\;\simeq\;\bigsqcup_{\lambda\in\Lambda^2_{\text{fixed}}} \n{S}^{1}\;,
\end{equation}
where the involution on $\n{S}^{1}$ is trivial.
Now, by observing that 
$$
{\rm Vec}_{\rr{Q}}^{2}\big(U_2\big)\;\simeq\;\prod_{\lambda\in\Lambda^2_{\text{fixed}}}{\rm Vec}_{\rr{Q}}^{2}\big(\{\ast\}\big)\;=\;0
$$
and 
$$
\big[U_1\cap U_2, \hat{\n{U}}(2)\big]_{\Z_2}\;\simeq\;\prod_{\lambda\in\Lambda^2_{\text{fixed}}}\big[\n{S}^1, S{\n{U}}(2)\big]\;=\;0\,,
$$
one concludes from the clutching construction with respect to the decomposition $U_1\cup U_2= X/X_1$ that the inclusion $\imath: Z_X\to X/X_1$ induces a bijection
$$
\imath^*\;:\;{\rm Vec}_{\rr{Q}}^{2}\big(X/X_1,\tau\big)\stackrel{\simeq}{\longrightarrow}{\rm Vec}_{\rr{Q}}^{2}\big(Z_X,\tau\big)\;.
$$
The Mayer-Vietoris exact sequence for the decomposition $U_1\cup U_2= X/X_1$, along with the vanishings $H^k_{\Z_2}(U_2|U_2^\tau,\Z(1))=0=H^k_{\Z_2}(U_1\cap U_2|(U_1\cap U_2)^\tau,\Z(1))$ due to the equalities $U_2=U_2^\tau$ and $U_1\cap U_2=(U_1\cap U_2)^\tau
$, implies the isomorphism
$$
0\;\stackrel{}{\longrightarrow}H^2_{\Z_2}\big(X/X_1|(X/X_1)^\tau,\Z(1)\big)\;\stackrel{\imath^*}{\longrightarrow}\;H^2_{\Z_2}\big(Z_X|Z_X^\tau,\Z(1)\big)\;\stackrel{}{\longrightarrow}\;0\;.
$$
In conclusion one has by naturality the commutative diagram
\begin{equation*}
\begin{diagram}
  {\rm Vec}_{\rr{Q}}^{2}(X/X_1,\tau)       &   \rTo^{\imath^*}_{\simeq}                &{\rm Vec}_{\rr{Q}}^{2}(Z_X,\tau)        \\
  \dTo^{\kappa}    &&  \dTo^{\kappa_{Z}}_{\simeq}     \\
H^2_{\Z_2}\big(X/X_1|(X/X_1)^\tau,\Z(1)\big)&     \rTo^{\imath^*}_{\simeq}           &H^2_{\Z_2}\big(Z_X| (Z_X)^\tau,\Z(1)\big)
\end{diagram}\;
\end{equation*}
which shows that $\kappa$ is bijective. \qed 
  
\medskip 
  
We are now in the position to prove the main result of this section, which is Theorem \ref{theo:A_inject_fixed} for $d = 2$.

\begin{proposition}\label{prop:surg_2_di-2step}
Let $(X,\tau)$ by an involutive space which verifies Assumption \ref{ass:top}. Let its dimension be $d=2$. 
Then the FKMM-invariant
$$
\kappa\;:\; {\rm Vec}_{\rr{Q}}^{2m}\big(X,\tau\big)\;\stackrel{\simeq}{\longrightarrow}\; H^2_{\Z_2}\big(X| X^\tau,\Z(1)\big)   \;,\qquad\qquad \forall\ m\in\N \;
$$
induces a bijection.
\end{proposition}
\proof
Because of Proposition \ref{prop:injectivity_in_low_dimension}, we only prove the surjectivity of $\kappa$. In addition, it is enough to consider the surjectivity in the case of $m = 1$, which leads to the surjectivity in the case of $m > 1$ in view of the isomorphism \eqref{eq:stab_rank_Q_low_d>1}.
Let us start with the  commutative diagram \eqref{eq:forg_Y2} specialized for $Y=X_1\supset X_0\neq\empt$;
\begin{equation}\label{eq:forg_KY2}
\begin{diagram}
  {\rm Vec}_{\rr{Q}}^2(X/X_1|\{\ast\},\tau)       &   \rTo^{\pi^*}_{\simeq}                &{\rm Vec}_{\rr{Q}}^{2}\big(X|X_1,\tau\big)        &   \rTo^f_{}                &   {\rm Vec}_{\rr{Q}}^{2}\big(X,\tau\big)\\
  \dTo^{\kappa^\pi_{\rm rel}}    &&  \dTo^{\kappa_{\rm rel}}     &   &   \dTo_{\kappa} \\
H^2_{\Z_2}\big(X/X_1|\{\ast\}\cup (X/X_1)^\tau,\Z(1)\big)&       \rTo^{\pi^*}_{\simeq}             &H^2_{\Z_2}\big(X|X_1\cup X^\tau,\Z(1)\big)&       \rTo^{\imath^*}_{\text{surj.}}             & H^2_{\Z_2}\big(X| X^\tau,\Z(1)\big)     
\end{diagram}\;
\end{equation}
where the surjectivity of $\iota^*$ is shown as follows.
Let $X_1^\tau:=X_1\cap X^\tau$ and consider the \emph{triad} $(X; X_1,X^\tau)$. The long exact sequence in the equivariant cohomology associated with this triad reads \cite[Theorem 2.2]{kono-tamaki-02}
$$
\ldots \;H^1_{\Z_2}\big(X_1|X_1^\tau,\Z(1)\big)\;\stackrel{}{\longrightarrow}\;H^2_{\Z_2}\big(X|X_1\cup X^\tau,\Z(1)\big)\;\stackrel{\imath^*}{\longrightarrow}\;H^2_{\Z_2}\big(X| X^\tau,\Z(1)\big)\;\stackrel{}{\longrightarrow}\;H^2_{\Z_2}\big(X_1|X_1^\tau,\Z(1)\big)\;\ldots
$$
and the homomorphism $\imath^*:H^2_{\Z_2}(X|X_1\cup X^\tau,\Z(1))\to H^2_{\Z_2}(X| X^\tau,\Z(1))$
is induced by the inclusion $\imath:X^\tau\to X_1\cup X^\tau$. Since $H^2_{\Z_2}(X_1|X_1^\tau,\Z(1))=0$ in view of Proposition \ref{prop:zero_rel_cohom_low} it follows that  the homomorphism $\imath^*$ is surjective. As a result, $\kappa$ is surjective \emph{if and only if} $\kappa_{\rm rel}^\pi$ is surjective. 
By observing that $\{\ast\}\in (X/X_1)^\tau\subset X/X_1$ one concludes that 
$\{\ast\}\cup (X/X_1)^\tau= (X/X_1)^\tau$ and one has the commutative diagram
\begin{equation}\label{eq:forg_KY3}
\begin{diagram}
  {\rm Vec}_{\rr{Q}}^2(X/X_1|\{\ast\},\tau)       &   \rTo^{f}_{\text{surj.}}                &{\rm Vec}_{\rr{Q}}^2(X/X_1,\tau)        \\
  \dTo^{\kappa^\pi_{\rm rel}}    &&  \dTo^{\kappa}_\simeq     \\
H^2_{\Z_2}\big(X/X_1|\{\ast\}\cup (X/X_1)^\tau,\Z(1)\big)&       =             &H^2_{\Z_2}\big(X/X_1| (X/X_1)^\tau,\Z(1)\big)
\end{diagram}\;
\end{equation}
where the surjectivity of $f$ is due to the locall triviality of $\rr{Q}$-bundles, and the bijectivity of the right-hand side vertical arrow $\kappa$ is due to Corollary \ref{corol:hrder_rememb}. Summarizing, the  commutative diagram \eqref{eq:forg_KY3} assures that $\kappa^\pi_{\rm rel}$ is a surjection and so one gets the surjectivity of the map $\kappa$
 in \eqref{eq:forg_KY2}.\qed

\subsection{Surjectivity in $d=3$: The case of a free involutions}
\label{app:surj_fkmm_free}
In this section we will  prove that in case of a free involutive space $(X,\tau)$ the FKMM-invariant, identified with the  first \virg{Real} Chern class of the associate determinant line bundle, is surjective in dimension $d=3$. 

\medskip

The skeleton decomposition \eqref{eq:skeleton_caz}, adapted for the case of a free involution, suggests that we can express $X$ as the union $X =U_1\cup U_2$ where $U_1$ and $U_2$
are  two
closed invariant subspaces which satisfy the   $\Z_2$-equivariant
homotopy equivalences 
\begin{equation}\label{eq:U_1_U_2_d=3}
U_1\simeq X_{d-1}\;,\qquad U_2\;\simeq\;\bigsqcup_{\lambda\in\Lambda}\tilde{\boldsymbol{e}}^d_\lambda\;\simeq\;\bigsqcup_{\lambda\in\Lambda}\Z_2\;,\qquad U_1\cap U_2\;\simeq\;\bigsqcup_{\lambda\in\Lambda}\partial\tilde{\boldsymbol{e}}^d_\lambda\;\simeq\;\bigsqcup_{\lambda\in\Lambda}\big(\Z_2\times \n{S}^{d-1}\big)\;.
\end{equation}
More precisely we can take $U_1$ as a \virg{small} closed invariant neighborhood of $X_{d-1}$ in $X$ and $U_2$ as the complement of the interior of $U_1$. Let us prove the following technical result.
\begin{lemma}\label{Lemma:eq:U_1_U_2_d=3}
Let $U_1$ and $U_2$ as in \eqref{eq:U_1_U_2_d=3} with $d\geqslant 3$. Then the inclusions $U_1\cap U_2\stackrel{\jmath_1}{\to} U_1\stackrel{\imath_1}{\to} X$ induce the exact sequence
\begin{equation}\label{eq:exact_app_B1}
\begin{aligned}
0\;&\longrightarrow\;H^2_{\Z_2}\big(X,\Z(1)\big)\;\stackrel{\imath_1^*}{\longrightarrow}\;H^2\big(U_1,\Z(1)\big)\;\stackrel{\jmath_1^*}{\longrightarrow}\;H^2_{\Z_2}\big(U_1\cap U_2,\Z(1)\big)\;.
\end{aligned}
\end{equation}
Moreover, for $d\geqslant 4$ it holds that 
\begin{equation}\label{eq:exact_app_B1_bis}
H^2_{\Z_2}\big(U_1\cap U_2,\Z(1)\big)\;=\;0
\end{equation}
and consequently one has the isomorphism
\begin{equation}\label{eq:exact_app_B1o1}
\begin{aligned}
H^2_{\Z_2}\big(X,\Z(1)\big)\;\stackrel{\imath_1^*}{\simeq}\;H^2\big(U_1,\Z(1)\big)\;.
\end{aligned}
\end{equation}

\end{lemma}
\proof
Consider the Mayer-Vietoris exact sequence for the pair $\{U_1,U_2\}$
$$
H^1_{\Z_2}\big(U_1\cap U_2,\Z(1)\big)\;\stackrel{}{\longrightarrow}\;H^2_{\Z_2}\big(X,\Z(1)\big)\;\stackrel{(\imath_1^*,\imath_2^*)}{\longrightarrow}\;H^2_{\Z_2}\big(U_1,\Z(1)\big)\oplus H^2_{\Z_2}\big(U_2,\Z(1)\big)\;\stackrel{\jmath_1^*-\jmath_2^*}{\longrightarrow}\;H^2_{\Z_2}\big(U_1\cap U_2,\Z(1)\big)
$$
where $\imath_j^*$ and $\jmath_j^*$
are the maps induced by the inclusions $\imath_j:U_j\to X$
and   $\jmath_j:U_1\cap U_2\to U_j$, respectively. By using the additivity axiom in equivariant cohomology one obtains 
\begin{equation}\label{eq:exact_app_B2-x1}
H^2_{\Z_2}\big(U_2,\Z(1)\big)\;\simeq\;\bigoplus_{\lambda\in\Lambda}H^2_{\Z_2}\big(\Z_2,\Z(1)\big)\;\simeq\;0
\end{equation}
where the last equality is justified by Lemma \ref{appA_lemma}.
Similarly, again in view of  the additivity axiom,  one has that
\begin{equation}\label{eq:exact_app_B2}
H^k_{\Z_2}\big(U_1\cap U_2,\Z(1)\big)\;\simeq\;\bigoplus_{\lambda\in\Lambda}H^k_{\Z_2}\big(\Z_2\times\n{S}^{d-1},\Z(1)\big)\;\simeq\;\bigoplus_{\lambda\in\Lambda}H^k\big(\n{S}^{d-1},\Z\big)\;,\qquad\quad k=1,2\;.
\end{equation}
where the last isomorphism follows from Lemma \ref{appA_lemma}. Since $H^k(\n{S}^{d-1},\Z)=0$ whenever $0 < k < d-1$ one obtains that 
\begin{equation}\label{eq:exact_app_B5_bis}
\begin{aligned}
&H^1_{\Z_2}\big(U_1\cap U_2,\Z(1)\big)\;=\;0&\qquad& \forall\, d\geqslant 3\\
&H^2_{\Z_2}\big(U_1\cap U_2,\Z(1)\big)\;=\;0&\qquad& \forall\, d\geqslant 4\;.
\end{aligned}
\end{equation}
The result is proved by  plugging in \eqref{eq:exact_app_B2-x1} and \eqref{eq:exact_app_B5_bis} in the Mayer-Vietoris exact sequence.
\qed

\medskip

The exactness of the sequence \eqref{eq:exact_app_B1} enters crucially in the proof of the following result.
\begin{proposition}[Surjectivity for $d=3$]\label{prop_free_surg_d=3}
For any \emph{free} $\Z_2$-CW complex $(X,\tau)$ of dimension $d=3$ the FKMM-invariant $\kappa$, identified according to Proposition \ref{corol:II}, is \emph{surjective}.
\end{proposition}
\proof
The exact sequence \eqref{eq:exact_app_B1} and the FKMM-invariant can be combined to the following commutative diagram: 
\beql{eq:diag2}
\begin{diagram}
   {\rm Vec}_{\rr{Q}}^{2m}\big( X\big)        &   \rTo^{\imath_1^*}                 &   {\rm Vec}_{\rr{Q}}^{2m}\big(U_1\big)  & \rTo^{\jmath_1^*} & {\rm Vec}_{\rr{Q}}^{2m}\big(U_1\cap U_2\big)\\
       \dTo^{\kappa}     &   &   \dTo^{\kappa}_{\simeq}           &                                                          &            \dTo^{\kappa}_{\simeq}       \\
H^2_{\Z_2}\big(X,\Z(1)\big)&       \rTo^{\imath_1^*}_{\text{inj.}}              & H^2_{\Z_2}\big(U_1,\Z(1)\big)  &                                                     \rTo^{\jmath_1^*}& H^2_{\Z_2}\big(U_1\cap U_2,\Z(1)\big)     \\
\end{diagram}
\eeq
where the $\imath_1^*$ and $\jmath_1^*$ are the maps induced by the inclusions $\imath_1:U_1\to X$ and $\jmath_1:U_1\cap U_2\to U_1$, respectively. The injectivity of the horizontal arrow $\imath_1^*$ is due to Lemma \ref{Lemma:eq:U_1_U_2_d=3}, and the bijectivity of the last two vertical arrows is due to Proposition \ref{prop:surg_2_di-2step}, since both $U_1$ and $U_1\cap U_2$ are at most two-dimensional. Now, let $a\in H^2_{\Z_2}(X,\Z(1))$ be a given class, $\imath_1^*(a)$ its image in $H^2_{\Z_2}(U_1,\Z(1))$ and $\bb{E}_1\to U_1$ a \virg{Quaternionic} vector bundle in the equivalence class determined by $\kappa(\bb{E}_1)=\imath_1^*(a)$. Due to the commutativity of the diagram one has
$$
\kappa\big(\bb{E}_1|_{U_1\cap U_2}\big)\;=\;\kappa\big(\jmath_1^*(\bb{E}_1)\big)\;=\;
\jmath_1^*\big(\kappa(\bb{E}_1)\big)\;=\;\jmath_1^*\big(\imath_1^*(a)\big)\;=\;0
$$
where the last equality is a  consequence of the exactness of \eqref{eq:exact_app_B1} that forces $\jmath_1^*\circ\imath_1^*=0$. Since $\kappa\big(\bb{E}_1|_{U_1\cap U_2}\big)=0$ one concludes that $\bb{E}_1|_{U_1\cap U_2}\to U_1\cap U_2$  must be the trivial 
$\rr{Q}$-bundle. Let $\bb{E}_2= U_2\times\C^{2m}$ be the trivial 
$\rr{Q}$-bundle over $U_2$. By gluing $\bb{E}_1$ and $\bb{E}_2$ along ${U_1\cap U_2}$ one obtains a $\rr{Q}$-bundle $\bb{E}:=\bb{E}_1\cup\bb{E}_2$ over $X$. By construction it holds that
$$
\imath_1^*\big(\kappa(\bb{E})\big)\;=\;\kappa\big(\imath_1^*(\bb{E})\big)\;=\;\kappa\big(\bb{E}|_{U_1}\big)\;=\;\kappa\big(\bb{E}_1\big)\;=\;\imath_1^*\big(a\big)
$$
and the injectivity of $\imath_1^*$ implies that  $\kappa(\bb{E})=a$. This proves that for each class $a\in H^2_{\Z_2}\big(X,\Z(1)\big)$ there is an element $[\bb{E}]\in{\rm Vec}_{\rr{Q}}^{2m}\big( X\big)$ such that $\kappa(\bb{E})=a$, namely the FKMM-invariant is surjective.
\qed

\subsection{Surjectivity in $d=3$: The case of a finite fixed-point set}
\label{app:surj_fkmm_free_finit}
{In this section we will provide the proof of Theorem \ref{theo-int-4}, namely we will show
 the surjectivity of the FKMM-invariant in the case of an involutive space $(X,\tau)$ which satisfies the following assumption:
\begin{assumption}\label{ass:top-2}
Let $(X,\tau)$ be an involutive space and assume that
\begin{itemize}
\item[(a)]
$X$ is a compact manifold without boundary of dimension $d=3$;
\vspace{1mm}
\item[(b)] The involution $\tau$ is smooth.
\vspace{1mm}
\item[(c)] The fixed-point set $X^\tau$ consists of a finite collection of points. 
\end{itemize}
\end{assumption}
\noindent
Let us observe that a space  $X$ which fulfills Assumption \ref{ass:top-2}
 is a \emph{closed} manifold and the pair $(X,\tau)$ automatically admits the structure of a $\Z_2$-CW-complex (see \eg \cite[Theorem 3.6]{may-96}).
Moreover, under the extra assumption $H^2_{\Z_2}(X,\Z(1))=0$ an involutive space $(X,\tau)$ of the type described in Assumption \ref{ass:top-2} is also
an \emph{FKMM-space} (\cf. Definition \ref{defi:FKMM-space}) of dimension $d=3$. 
Let us point out that for this type of spaces $H^2_{\Z_2}(X|X^\tau,\Z(1))$ has a 
geometric representation provided by equation \eqref{eq:iso_for_fkmm}. This is a  crucial ingredient in the following results.}

\begin{proposition} \label{propos_surjFKMM_spac}
Let $(X, \tau)$ be an involutive space which meets Assumption \ref{ass:top-2}. Assume in addition that $(X, \tau)$ is an
FKMM-space, namely $H^2_{\Z_2}(X,\Z(1))=0$. Then the FKMM-invariant
$$
\kappa\;:\; {\rm Vec}_{\rr{Q}}^{2m}\big(X,\tau\big)\;\stackrel{\simeq}{\longrightarrow}\; H^2_{\Z_2}\big(X|X^\tau,\Z(1)\big)\;\simeq\; 
{\rm Map}\big(X^\tau,\{\pm 1\}\big)/\big[X,\tilde{\n{U}}(1)\big]_{\Z_2}  \;,\qquad\qquad \forall\ m\in\N
$$
induces a bijection.
\end{proposition}

\begin{proof}
Let $X^\tau = \{ x_1, \ldots, x_n \}$. In view of the geometric representation of $H^2_{\Z_2}(X|X^\tau,\Z(1))$ it is enough to show that for each
$i = 1, \ldots, n$ we can construct a  rank $2$ $\rr{Q}$-bundle $\bb{E}_i\to X$ such that its FKMM-invariant $\kappa(\bb{E}_i)$ 
 is represented by an equivariant map $\phi_i : X^\tau \to \{\pm1\}$ which verifies $\phi_i(x_i) = -1$ and $\phi_i(x_j) = 1$ for all $j \neq i$. 
 To construct such $\bb{E}_i$, we invoke the \emph{slice theorem}
 \cite[Chapter I, Section 3]{hsiang-75} which allows us 
   to choose a $\tau$-invariant 3-dimensional closed disk $\n{D}^3_i$ centered in $x_i$ which  can be equivariantly identified with the unit ball in $\R^3$ endowed with  the antipodal involution $x \mapsto -x$. Moreover, it is always possible to choose these discs in such a way that $\n{D}^3_i\cap\n{D}^3_j=\empt$ if $i\neq j$. Let $\bb{D}:=\bigcup_{i=1}^n\n{D}^3_i$,  $\partial \bb{D}:=\bigcup_{i=1}^n\partial \bb{D}_i\simeq\bigcup_{i=1}^n\n{S}^{0,3}$
 its boundary and $X':=\overline{X\setminus \bb{D}}$. We know from 
 \cite[Lemma 5.2]{denittis-gomi-14-gen} that $[\n{S}^{0,3}, \hat{\n{U}}(2)]_{\Z_2}\simeq \Z$ and  $\varphi:\n{S}^{0,3}\to \hat{\n{U}}(2)$ is an equivariant map in the class labelled by $k$ if and only if ${\rm deg}(\varphi)=k$ and ${\rm det}(\varphi)=(-1)^k$. Let $[\xi]\in [\n{S}^{0,3}, \hat{\n{U}}(2)]_{\Z_2}$ be the generator, namely let $\xi$ be an equivariant map such that  ${\rm deg}(\xi)=1$ and ${\rm det}(\xi)=-1$. Let  $\Psi_i:\partial \bb{D}\to \hat{\n{U}}(2)$ be the equivariant map defined as follows: 
$$
\Psi_i|_{\partial \bb{D}_j}\;:=\;
\left\{
\begin{aligned}
&\xi&\qquad&\text{if}\ \ j=i\\
&\n{1}_2&\qquad&\text{if}\ \ j\neq i
\end{aligned}
\right.\;.
$$
By means of the equivariant  clutching construction we can define the $\rr{Q}$-bundles
$$
\bb{E}_i\;:=\;\big(X'\times \C^2\big)\;\cup_{\Psi_i}\;\big(\bb{D}\times \C^2\big)\;,\qquad\quad i=1,\ldots,n\;.
$$
Since $\bb{E}_i|_{X^\tau}=\{x_1,\ldots,x_n\}\times\C^2$ is trivial it follows that the canonical section $s_{\bb{E}_i}$, defined according \eqref{eq:triv_2},
agrees with the constant map $s_{\bb{E}_i}:X^\tau\to\{+1\}$.
To prove that the FKMM-invariant of  $\bb{E}_i$ can be represented by the function $\phi_i$
we need to find an equivariant section  $t_i:X\to{\rm det}(\bb{E}_i)$
such that $t_i|_{X^\tau}=\phi_i$. In fact 
Proposition \ref{prop:fkmm-inv_fkmm-space} assures that the FKMM-invariant of $\bb{E}_i$ can be seen as the difference between the canonical section $s_{\bb{E}_i}$ and any other equivariant sections of ${\rm det}(\bb{E}_i)$. These equivariant sections  are in one-to-one correspondence with a pair of equivariant
maps $u_{X'}:X'\to \tilde{\n{U}}(1)$ and $u_{\bb{D}}:\bb{D}\to \tilde{\n{U}}(1)$
such that $u_{X'}(x):={\rm det}(\Psi_i)\;u_{\bb{D}}(x)$ for all $x\in X'\cap\bb{D}$. 
With the  specific choice  $u_{X'}\equiv 1$ and 
 $$
u_{\bb{D}}|_{ \bb{D}_j}\;:=\;
\left\{
\begin{aligned}
&-1&\qquad&\text{if}\ \ j=i\\
&+1&\qquad&\text{if}\ \ j\neq i
\end{aligned}
\right.\;
$$
one can check that  $t_i|_{X^\tau}=\phi_i$ and this concludes the proof.
\end{proof}

\medskip

The proof of Theorem \ref{theo-int-4} requires the next preliminary result.

\begin{lemma} \label{lem:surj_fkmm_finite_fixed_point}
Let $(X, \tau)$ be an involutive space which meets Assumption \ref{ass:top-2}. Then, for any $\rr{R}$-line bundle $\bb{L}$, there exists a  $\rr{Q}$-bundle $\bb{E}$ of rank 2 such that $\det (\bb{E}) \simeq \bb{L}$.
\end{lemma}

\begin{proof}
Let $\bb{D} \subset X$ and $X' = \overline{X \setminus \bb{D}}$ be as in the proof of Proposition \ref{propos_surjFKMM_spac}. Suppose that a \virg{Real} line bundle $\bb{L} \to X$ is given. On the one hand, we have $H^2_{\Z_2}(\bb{D}, \Z(1)) = 0$, because of the $\Z_2$-homotopy equivalence $\bb{D} \simeq X^\tau$. This implies that $\bb{L}|_{\bb{D}}$ is trivial. Hence, if $\bb{E}_{\bb{D}}\simeq \bb{D} \times \C^2$ is the product $\rr{Q}$-bundle one has that then $\det (\bb{E}_{\bb{D}}) \simeq \bb{L}|_{\bb{D}}$. On the other hand, $X'$ admits the structure of a $3$-dimensional $\Z_2$-CW complex with free involution. Therefore Proposition \ref{prop_free_surg_d=3} provides a $\rr{Q}$-bundle $\bb{E}_{X'} \to X'$ of rank $2$ such that $\det (\bb{E}_{X'}) \simeq \bb{L}|_{X'}$. Notice that 
$$
\det \big(\bb{E}_{X'}\big)|_{X' \cap \bb{D}} 
\;\simeq\;
\big(\bb{L}|_{X'}\big)|_{X' \cap \bb{D}} 
\;=\; 
\big(\bb{L}|_{\bb{D}}\big)|_{X' \cap \bb{D}}
\;\simeq\;
\det \big(\bb{E}_{\bb{D}}\big)|_{X' \cap \bb{D}}\;. 
$$
We then have an isomorphism $\bb{E}_{X'}|_{X' \cap \bb{D}} \simeq \bb{E}_{\bb{D}}|_{X' \cap \bb{D}}$ by Proposition \ref{prop:surg_2_di-2step}, since $X' \cap \bb{D}$ is a $2$-dimensional $\Z_2$-CW complex. Consequently, we can glue $\bb{E}_{X'}$ and $\bb{E}_{\bb{D}}$ together along $X' \cap \bb{D}$ to form a $\rr{Q}$-bundle $\bb{E} = \bb{E}_{X'} \cup \bb{E}_{\bb{D}}$ of rank $2$ such that 
$$
\det \big(\bb{E}\big)\;
\simeq\; \det \big(\bb{E}_{X'} \cup \bb{E}_{\bb{D}}\big)
\;\simeq\; \det \big(\bb{E}_{X'}\big)\; \cup\; \det \big(\bb{E}_{\bb{D}}\big)
\;\simeq\; \bb{L}|_{X'} \;\cup\; \bb{L}|_{\bb{D}}\;
\simeq\; \bb{L}.
$$
This completes the proof.
\end{proof}

\medskip

\begin{proof}[Proof of Theorem \ref{theo-int-4}]
Because of Proposition \ref{prop:injectivity_in_low_dimension}, it suffices to prove the surjectivity of 
$$
\kappa\;:\;{\rm Vec}_{\rr{Q}}^{2}\big(X,\tau\big)\;\stackrel{}{\longrightarrow}\;H^2_{\Z_2}\big(X|X^\tau,\Z(1)\big)\;.
$$
Let $(\bb{L}, s)$ be a pair representing a given element in $H^2_{\Z_2}\big(X|X^\tau,\Z(1)\big)$. Lemma \ref{lem:surj_fkmm_finite_fixed_point} ensures that there is a rank $2$ $\rr{Q}$-bundle $\bb{E}'$ such that $\det (\bb{E}') \simeq \bb{L}$. Let $\phi : X^\tau \to \Z_2$ be a map such that $s = s_{\bb{E}'} \cdot \phi$, where $s_{\bb{E'}}$ is the canonical section of $\bb{E}'$. 
Following the proof of
 Proposition \ref{propos_surjFKMM_spac} one concludes that it is possible to find another $\rr{Q}$-bundle $\bb{E}''$ of rank $2$ such that the $\rr{R}$-bundle $\det( \bb{E}'')$ is trivial and $s_{\bb{E}''}$ agrees with $\phi$ under the isomorphism $\det (\bb{E}'') \simeq X \times \C$ of $\rr{R}$-bundles. For the direct sum $\bb{E}' \oplus \bb{E}''$, one has
$$
\big(\det (\bb{E}' \oplus \bb{E}''), s_{\bb{E}' \oplus \bb{E}''}\big)\;
\simeq\; \big(\det (\bb{E}') \otimes \det (\bb{E}''), s_{\bb{E}'} \otimes s_{\bb{E}''}\big)
\;\simeq\; 
\big(\det (\bb{E}'), s_{\bb{E}'}\cdot \phi\big) 
\;\simeq\; (\bb{L}, s)\;.
$$
Proposition \ref{theo:stab_ran_Q_even} allows to find a rank $2$ $\rr{Q}$-bundle $\bb{E}$ such that the direct sum of $\bb{E}$ and the product $\rr{Q}$-bundle $X \times \C^2$ is isomorphic to $\bb{E}' \oplus \bb{E}''$. Since $\kappa(\bb{E}) = \kappa( \bb{E}' \oplus \bb{E}'')$ the proof is completed.
\end{proof}

\section{\virg{Quaternionic} vector bundles over the three-dimensional lens space}
\label{app:quaternionic_lens}

The aim of this section is to classify the \virg{Quaternionic} vector bundles over the three-dimensional \emph{lens space} endowed with its natural involution. 
The explicit computation of the {relative} equivariant cohomology of this involutive space which has a non-empty fixed-point set will provide an explicit example where the FKMM-invariant fails to be surjective.

\subsection{The three-dimensional lens space with its natural involution}
The three dimensional sphere can be parametrized as the unit sphere in $\C^2$,
\begin{equation}\label{eq:sphere_complex}
\n{S}^3\;\equiv\;\big\{(z_0,z_1)\in\C^2\ |\ |z_0|^2+|z_1|^2=1  \big\}\;\subset\;\C^2\;.
\end{equation}
This representation allows 
 $u\in\n{U}(1)$ to act on $\n{S}^3$ through the mapping $(z_0,z_1)\mapsto (uz_0,uz_1)$. 
This action of $\n{U}(1)$ on  $\n{S}^3$ is evidently \emph{free}.
 The inclusion of $\Z_p\subset \n{U}(1)$, given by the fact that $\Z_p$ can be identified with the set of the $p$-th roots of the unity, implies that one can  define a free action of every cyclic group $\Z_p$ on $\n{S}^3$. More precisely we can let
 $k\in\Z_p$ act on $\n{S}^3$ through the rotation 
 $$
 k\;:\;(z_0,z_1)\longmapsto \left(\expo{\ii2\pi\frac{k}{p}}z_0,\expo{\ii2\pi\frac{k}{p}}z_1\right)\;.
 $$ 
 The quotient space
 $$
 L_p\;:=\;\n{S}^3/\Z_p
 $$
is called the (three-dimensional) \emph{lens space} (see \cite[Example 18.5]{bott-tu-82} or \cite[Example 2.43]{hatcher-02} for more details) and sometime is denoted with the symbol $L(1;p)$. 
By combining the facts that $\n{S}^3$ is simply connected and the $\Z_p$-action on $\n{S}^3$
is free one concludes that $\n{S}^3$ is the {universal cover} of $L_p$. The last observation turns out to be relevant for the calculation of the homotopy groups:
$$
\pi_1\big(L_p\big)\;\simeq\;\Z_p\;,\qquad\quad \pi_j\big(L_p\big)\;\simeq\;\pi_j\big( \n{S}^3\big)\;,\ \ j\geqslant 2\;.
$$
The well-known fiber sequence
$$
\n{U}(1)\;\hookrightarrow\;\n{S}^3\;\to\; \n{S}^3/\n{U}(1)\;\simeq\;\C P^1
$$
says that $\n{S}^3$ can be seen as the total space of a principal $\n{U}(1)$-bundle over $\C P^1$
whose Chern number is 1. Similarly,  after observing that $\n{U}(1)/\Z_p\simeq \n{U}(1)$, one obtains the 
 fiber sequence
\begin{equation}\label{eq:lensPB}
\n{U}(1)\;\simeq\;\n{U}(1)/\Z_p\;\hookrightarrow\; L_p\;\to\; \n{S}^3/\n{U}(1)\;\simeq\;\C P^1
\end{equation}
which tells us that  $L_p$ is the total space of a principal $\n{U}(1)$-bundle over $\C P^1$ with typical fiber $\n{U}(1)/\Z_p$. The Chern number of this bundle can be computed to be $p$.
The above fiber sequence can be used for the computation of the cohomology of  $L_p$ which turns out to be \cite[Example 18.5]{bott-tu-82}
\begin{equation}\label{cohom_lens_spac}
H^k\big(L_p,\Z\big)\;\simeq\;
\left\{
\begin{aligned}
&\Z&&\ \ \ k=0,3\\
&\Z_p&&\ \ \ k=2\\
&0&&\ \ \ \text{otherwise}	\;.
\end{aligned}
\right.
\end{equation}

\medskip

The parametrization \eqref{eq:sphere_complex}
allows to equip 
$\n{S}^3\subset\C^2$
with the involution 
 induced by the complex conjugation 
$(z_0,z_1)\mapsto(\overline{z_0},\overline{z_1})$.
The computation
 $$
 \expo{\ii2\pi\frac{p-k}{p}}\;=\;\expo{-\ii2\pi\frac{k}{p}}\;=\;\overline{\expo{\ii2\pi\frac{k}{p}}}\;,\qquad\quad k\in\Z_p\;
 $$
shows that 
 $\n{Z}_p\subset\n{U}(1)$ 
 is preserved by the complex conjugation.
Therefore, the involution on $\n{S}^3$ descends to an involution $\tau$ on
 $L_p$.  
  The involutive space $(L_p,\tau)$ inherits the structure of a smooth (three-dimensional) manifold with a smooth involution, hence it admits a $\Z_2$-CW-complex structure \cite[Theorem 3.6]{may-96}. Let us point out that it is possible
to think of  $L_p\to \C P^1$ as a \virg{Real}
principal $\n{U}(1)$-bundle where the \virg{Real} structure on the total space is provided by $\tau$ and the
 involution $\tau'$ on the base space $\C P^1$ is still given by the complex conjugation $\tau':[z_0,z_1]\mapsto[\overline{z_0},\overline{z_1}]$.
 
 \medskip

 Let us focus now on the case $p=2q > 0$.
\begin{lemma}\label{lemma:B_fixed}
The fixed-point set of $L_{2q}$ under the involution $\tau$ has the form
$$
L_{2q}^\tau\;=\;S_0\;\sqcup\;S_1\;\simeq\;\n{S}^1\;\sqcup\;\n{S}^1
$$
where
$$
\begin{aligned}
S_0\;&:=\;\left.\left\{\big[\cos\theta,\sin\theta\big]\in L_{2q}\ \right|\ \theta\in\R\right\}\\
S_1\;&:=\;\left.\left\{\left[\expo{-\ii\frac{\pi}{2q}}\cos\theta,\expo{-\ii\frac{\pi}{2q}}\sin\theta\right]\in L_{2q}\ \right|\ \theta\in\R\right\}\\
\end{aligned}
$$
and $S_0\simeq S_1\simeq \n{S}^1$. 
\end{lemma} 
 \proof

A point $[z_0, w_0] \in X$ is a fixed point under the involution if and only if $\overline{z_0} =  \zeta^k z_0$ and $\overline{z_1} = \zeta^k z_1$ for some $k = 0, 1, \ldots, 2q - 1$, where $\zeta = \expo{\ii\frac{\pi}{q}}$. For convinience, let us introduce the following subset for each $k$:
$$
S_k\;: =\; \left.\left\{ [z_0, z_1] \in L_{2q}\ \right|\ 
(z_0, z_1) \in \n{S}^3,\;\;  
\overline{z_0} = \zeta^kz_0,\;\; \overline{z_0} = \zeta^kz_1 \right\},
$$
so that $L_{2q}^\tau$ is the union of $S_0, S_1, \ldots, S_{p-1}$. It is clear that:
$$
S_0\;=\;\left.\left\{\big[\cos\theta,\sin\theta\big]\in L_{2q}\ \right|\ \theta\in\R\right\}\;
\simeq\; \n{S}^1\;.
$$
Suppose here that $k$ is even, so that $k = 2\ell$. Then $S_0 = S_{2\ell}$. Actually, $[z_0, z_1] \in S_{2\ell}$ if and only if $[\zeta^\ell z_0, \zeta^\ell z_1] \in S_0$. Because $\zeta^\ell \in \Z_{2q}$, we find $S_0 = S_{2\ell}$ in $L_{2q}$. Similarly, in the case of $k = 1$, we obtain:
$$
S_1\; =\; 
\left.\left\{ [\zeta^{-\frac{1}{2}}\cos \theta, \zeta^{-\frac{1}{2}}\sin \theta] 
\in L_{2q}\ \right|\
\theta \in \R \right\}\; \simeq\;  \n{S}^1\;.
$$
 Suppose then that $k$ is odd, so that $k = 2\ell + 1$. In this case, $[z_0, z_1] \in S_{2\ell + 1}$ if and only if $[\zeta^\ell z_0, \zeta^\ell z_1] \in S_1$. Since $\zeta^\ell \in \Z_{2q}$, it holds that $S_1 = S_{2\ell + 1}$ in $L_{2q}$. Notice that $\zeta^{-\frac{1}{2}} \not\in \Z_{2q}$. Note also $-1 \not\in \{  \zeta^{\ell-\frac{1}{2}}\ |\ \ell \in \Z \}$ since $q > 0$. Consequently, $S_0 \cap S_1 = \emptyset$ and the lemma is proved. \qed

\subsection{The equivariant cohomology of $L_{2q}$}
The equivariant cohomology of the involutive space $(L_{2q},\tau)$ can be computed by means of the Gysin exact sequence \cite[Corollary 2.11]{gomi-13} applied to the \virg{Real} principal $\n{U}(1)$-bundle
$L_{2q}\to \C P^1$ described by the \eqref{eq:lensPB}.
 
\medskip
 
First of all one needs the computation of the equivariant cohomology of the space $\C P^1$ endowed with the involution given by the complex conjugation. 
This space, being a spherical conjugation complex \cite{hausmann-holm-puppe-05}, admits the following presentation of the $\Z_2$-equivariant cohomology ring
\begin{equation}\label{eq:ringCP1}
{H}^\bullet_{\Z_2}\big(\C P^1,\Z\big)\;\oplus\;{H}^\bullet_{\Z_2}\big(\C P^1,\Z(1)\big)\simeq\;\Z\big[t^{\frac{1}{2}},c\big]/\big(2t^{\frac{1}{2}},c^2\big)
\end{equation}
where $t^{\frac{1}{2}}\in {H}^1_{\Z_2}(\C P^1,\Z(1))\simeq{H}^1_{\Z_2}(\{\ast\},\Z(1))\simeq \Z_2$
 and $c\in {H}^2_{\Z_2}(\C P^1,\Z(1))\simeq {H}^2(\C P^1,\Z)\simeq\Z$ are basis elements. In particular $c$ can be understood both as the 
the first \virg{Real} Chern class of the \virg{Real} principal $\n{U}(1)$-bundle $\n{S}^2\to \C P^1$ as well as the first  Chern class of the same principal $\n{U}(1)$-bundle without additional structures. Indeed the map which forgets the \virg{Real} structure just acts as the identity $f:{H}^2_{\Z_2}(\C P^1,\Z(1))\to {H}^2(\C P^1,\Z)$. A proof of \eqref{eq:ringCP1} can be derived from \cite[Lemma 2.17]{gomi-13} where the more general case $\C P^\infty$ is considered. Let us  notice that the difference between 
 the cases $\C P^1$ and $\C P^\infty$ is given by the constraint $c^2=0$ in \eqref{eq:ringCP1} due to the low dimensionality of $\C P^1$. In low dimension the \eqref{eq:ringCP1} reads:

 %
 \begin{table}[h]\label{tab:B1}
 \centering
 \begin{tabular}{|c||c|c|c|c|c|}
 \hline
  & $k=0$   & $k=1$ & $k=2$ & $k=3$& $k=4$\\
\hline
 \hline
 \rule[-3mm]{0mm}{9mm}
 ${H}^k_{\Z_2}(\C P^1,\Z(1))$& $0$  & $\Z_2$ $\left[t^{\frac{1}{2}}\right]$ & $\Z$ $[c]$ & $\Z_2$ $\left[t^{\frac{3}{2}}\right]$& $\Z_2$ $[t c]$\\
\hline
 \rule[-3mm]{0mm}{9mm}
${H}^k_{\Z_2}(\C P^1,\Z)$ & $\Z$  & $0$ & $\Z_2$ $[t]$ 
&$\Z_2$ $\left[t^{\frac{1}{2}} c\right]$
& $\Z_2$ $\left[t^2\right]$\\
\hline
\end{tabular}\vspace{2mm}
 \caption{\footnotesize The equivariant cohomology of the involutive space $\C P^1$ with involution given by the complex conjugation up to degree $k=4$. The generators of the groups are listed in the square brackets.}
 \end{table}
%

  \medskip
  
  The data contained in Table 5.1 along with the fact that the first \virg{Real} Chern class of the \virg{Real} principal $\n{U}(1)$-bundle $L_{2q}\to \C P^1$ has value 
  $c^{\rr{R}}_1(L_{2q})=2qc$ can be used in  the Gysin exact sequence \cite[Corollary 2.11]{gomi-13} providing the following computation for the equivariant cohomology of
  $(L_{2q},\tau)$:
  %
 \begin{table}[h]\label{tab:B2}
 \centering
 \begin{tabular}{|c||c|c|c|c|}
 \hline
  & $k=0$   & $k=1$ & $k=2$ & $k=3$\\
\hline
 \hline
 \rule[-3mm]{0mm}{9mm}
 ${H}^k_{\Z_2}(L_{2q},\Z(1))$& $0$  & $\Z_2$ & $\Z_{2q}$  & $\Z_2\oplus \Z_2$ \\
\hline
 \rule[-3mm]{0mm}{9mm}
${H}^k_{\Z_2}(L_{2q},\Z)$ & $\Z$  & $0$ & $\Z_2\oplus \Z_2$ 
& $\Z\oplus \Z_2$
\\
\hline
 \rule[-3mm]{0mm}{9mm}
${H}^k(L_{2q},\Z)$ & $\Z$  & $0$ & $\Z_{2q}$ 
&$\Z$
\\
\hline
\end{tabular}\vspace{2mm}
 \caption{\footnotesize The (equivariant) cohomology of the involutive space $(L_{2q},\tau)$  up to degree $k=3$.}
 \end{table}
%

\noindent
From the exact sequence \cite[Proposition 2.3]{gomi-13} and the data contained in 
Table 5.2 one can conclude that the map which forgets the \virg{Real} structure induces a bijection
\begin{equation}\label{eq:appB_pic_Rpic}
f\;:\;{H}^2_{\Z_2}\big(L_{2q},\Z(1)\big)\;\stackrel{\simeq}{\longrightarrow}\;{H}^2\big(L_{2q},\Z\big)\;.
\end{equation}

\begin{remark}[The  Picard group of $L_{2q}$ and its \virg{Real} structure]\label{rk:app_B}{\upshape
Equation \eqref{eq:appB_pic_Rpic} implies that the Picard group of $L_{2q}$ and the 
\virg{Real} Picard group of the involutive space $(L_{2q},\tau)$ coincide:
$$
{\rm Pic}_{\rr{R}}\big(L_{2q},\tau\big)\;\stackrel{c^{\rr{R}}_1}{\simeq}\;H^2_{\Z_2}\big(L_{2q},\Z(1)\big)
\;\simeq\;\Z_{2q}\;\simeq\;
H^2_{\Z_2}\big(X,\Z(1)\big) \;\stackrel{c_1}{\simeq}\;{\rm Pic}\big(L_{2q}\big)\;.
$$
In particular this means that there are only $2q$ complex line bundles over  $L_{2q}$ (up to isomorphisms) and each one of these  can be endowed with a unique (up to isomorphisms)  $\rr{R}$-structure. The representatives of these line bundles can be constructed explicitly. For $k\in\Z$, we let $u\in\Z_{2q}$ act on $\n{S}^3\times\C$ by $((z_0,z_1),\lambda)\mapsto ((uz_0,uz_1),\overline{u}^k\lambda)$.
Since the action is free on the base space the quotient defines a complex line bundle  
$\bb{L}_k\to L_{2q}$ (\cf \cite[Proposition 1.6.1]{atiyah-67}). From the construction it results evident that $\bb{L}_k=\bb{L}_{k+2q}$ and $\bb{L}_0=L_{2q}\times \C$ is the trivial line bundle. Moreover, $\bb{L}_1$ provides a basis for ${\rm Pic}(L_{2q})\simeq\Z_{2q}$ in view of the fact that ${\bb{L}_1}^{\otimes k}\simeq \bb{L}_k$. The  $\rr{R}$-structure on $\bb{L}_k$ is evidently induced by the complex conjugation $[(z_0,z_1),\lambda]\mapsto[\tau{(z_0,z_1)},\overline{\lambda}]=[\overline{(z_0,z_1)},\overline{\lambda}]$.
}\hfill $\blacktriangleleft$
\end{remark}

   \medskip

The circle $\n{S}^1$ with trivial involution has the  cohomology groups  presented in Table 5.3.
 \begin{table}[h]\label{tab:B3}
 \centering
 \begin{tabular}{|c||c|c|c|c|}
 \hline
  & $k=0$   & $k=1$ & $k=2$ & $k=3$\\
\hline
 \hline
 \rule[-3mm]{0mm}{9mm}
 ${H}^k_{\Z_2}(\n{S}^1,\Z(1))$& $0$  & $\Z_2$ & $\Z_2$  & $\Z_2$ \\
\hline
 \rule[-3mm]{0mm}{9mm}
${H}^k_{\Z_2}(\n{S}^1,\Z)$ & $\Z$  & $\Z$ & $\Z_2$ 
& $ \Z_2$
\\
\hline
 \rule[-3mm]{0mm}{9mm}
${H}^k(\n{S}^1,\Z)$ & $\Z$  & $\Z$ & $0$ 
&$0$
\\
\hline
\end{tabular}\vspace{2mm}
 \caption{\footnotesize The (equivariant) cohomology of the circle $\n{S}^1$  with trivial involution up to degree $k=3$. The singular cohomology ${H}^\bullet(\n{S}^1,\Z)$ is well-know, the cohomology ${H}^\bullet_{\Z_2}(\n{S}^1,\Z)\simeq {H}^\bullet(\n{S}^1\times \R P^\infty,\Z)$ can be computed by means of the K\"{u}nneth formula  and  ${H}^\bullet_{\Z_2}(\n{S}^1,\Z(1))$ can be computed with the help of the exact sequences in \cite[Proposition 2.3]{gomi-13}.}
 \end{table}

\noindent 
In view of Lemma \ref{lemma:B_fixed}, the additivity axiom in equivariant cohomology and Table 5.3, one immediately gets Table 5.4.
%
 \begin{table}[h]\label{tab:B4}
 \centering
 \begin{tabular}{|c||c|c|c|c|}
 \hline
  & $k=0$   & $k=1$ & $k=2$ & $k=3$\\
\hline
 \hline
 \rule[-3mm]{0mm}{9mm}
 ${H}^k_{\Z_2}(L_{2q}^\tau,\Z(1))$& 
$0$  & $\Z_2 \oplus \Z_2$ & $\Z_2 \oplus \Z_2$  & $\Z_2\oplus\Z_2$ \\
\hline
\end{tabular}\vspace{2mm}
 \caption{\footnotesize The equivariant cohomology of the fixed-point set $L_{2q}^\tau\simeq\n{S}^1\sqcup\n{S}^1$  up to degree $k=3$.}
 \end{table}
%

\noindent
The bases of ${H}^1_{\Z_2}(L_{2q},\Z(1))\simeq\Z_2$ and ${H}^1_{\Z_2}(L_{2q}^\tau,\Z(1))={H}^1_{\Z_2}(S_0,\Z(1))\oplus{H}^1_{\Z_2}(S_1,\Z(1))\simeq \Z_2\oplus\Z_2$ are represented by the constants functions with value $-1$ on $L_{2q}$, $S_0$ and $S_1$. Hence the homomorphism
$$
\imath^*\;:\;{H}^1_{\Z_2}(L_{2q},\Z(1))\;\stackrel{\text{inj.}}{\longrightarrow} \; {H}^1_{\Z_2}(L_{2q}^\tau,\Z(1))
$$
induced by the inclusion $\imath: L_{2q}^\tau\to  L_{2q}$, is identified with the diagonal map $\Z_2\ni \epsilon\mapsto (\epsilon,\epsilon)\in \Z_2\oplus\Z_2$, and results injective. Consequently one has   
$$
{\rm Coker}^1\big(L_{2q}|L_{2q}^\tau,\Z(1)\big)\;:=\;H^1_{\Z_2}\big(L_{2q}^\tau,\Z(1)\big)\;/\;\imath^*\big(H^1_{\Z_2}(L_{2q},\Z(1))\big)\;\simeq\;\Z_2\;.
$$
The homomorphism
$$
\imath^*\;:\;{H}^2_{\Z_2}(L_{2q},\Z(1))\;\stackrel{0}{\longrightarrow} \; {H}^2_{\Z_2}(L_{2q}^\tau,\Z(1))
$$
turns out to be trivial. This fact can be proved as follow: According to Remark \ref{rk:app_B} the $\rr{R}$-line bundle ${\bb{L}_1}$ can be identified with the generator of ${H}^2_{\Z_2}(L_{2q},\Z(1))\simeq {\rm Pic}_{\rr{R}}(L_{2q},\tau)$. Then, if one proves that $\imath^*({\bb{L}_1})={\bb{L}_1}|_{S_0\sqcup S_1}$ is the trivial line bundle in ${\rm Pic}_{\rr{R}}(L_{2q}^\tau,\tau)\simeq {\rm Pic}_{\R}(L_{2q}^\tau)$
then one also obtains from ${H}^2_{\Z_2}(L_{2q}^\tau,\Z(1))\simeq {\rm Pic}_{\rr{R}}(L_{2q}^\tau,\tau)$ that $\imath^*$ is the trivial homomorphism in cohomology.
The triviality of $\imath^*({\bb{L}_1})$ is a  consequence of the next result.
\begin{lemma}\label{lemma:appB_triv_hom}
Let $\bb{L}_1\to L_{2q}$ be the $\rr{R}$-line bundle over the involutive space $(L_{2q},\tau)$ described in Remark \ref{rk:app_B} and $S_0\sqcup S_1= L_{2q}^\tau$ the fixed-point set according to Lemma \ref{lemma:B_fixed}. Let $\bb{L}_1|_{S_j}\to S_j$ be the restriction of $\bb{L}_1$ over $S_j$, $j=0,1$.
The $\rr{R}$-line bundles $\bb{L}_1|_{S_j}$ admit global $\rr{R}$-sections and consequently they  are trivial.
\end{lemma}
\proof
The total space of the line bundle $\bb{L}_1|_{S_0}\to S_0$ is described by
$$
\bb{L}_1|_{S_0}\;:=\;\big\{[(\cos\theta,\sin\theta),\lambda]\; |\; \theta\in\R\;, \lambda\in\C\big\}
$$
where the equivalence relation is given by $((\cos\theta,\sin\theta),\lambda)\sim ((u\cos\theta,u\sin\theta),u^{-1}\lambda)$, $u\in \Z_{2q}$. The map 
$s_0:S_0\to \bb{L}_1|_{S_0}$ defined by
$$
s_0\big([\cos\theta,\sin\theta]\big)\;:=\;[(\cos\theta,\sin\theta),1]
$$ 
is a nowhere vanishing section. It is also \virg{Real}, in fact
$$
\tau_\Theta(s_0)\big([\cos\theta,\sin\theta]\big)\;=\;\big[\tau(\cos\theta,\sin\theta), 1 \big]\;=\;\big[(\cos\theta,\sin\theta), 1 \big]
\;=\;s_0\big([\cos\theta,\sin\theta]\big)\;.
$$
In much of the same way, we can describe the total space of the line bundle $\bb{L}_1|_{S_1}\to S_1$ as
$$
\bb{L}_1|_{S_1}\;:=\;\left\{\left[\left(\expo{-\ii\frac{\pi}{2q}}\cos\theta,\expo{-\ii\frac{\pi}{2q}}\sin\theta\right),\lambda\right]\ \Big|\ \; \theta\in\R\;, \lambda\in\C\right\}\;.
$$
A nowhere vanishing section $s_1:S_1\to \bb{L}_1|_{S_1}$ is given by
$$
s_1\left(\left[\expo{-\ii\frac{\pi}{2q}}\cos\theta,\expo{-\ii\frac{\pi}{2q}}\sin\theta\right]\right)\;:=\;\left[\left(\expo{-\ii\frac{\pi}{2q}}\cos\theta,\expo{-\ii\frac{\pi}{2q}}\sin\theta\right),\expo{+\ii\frac{\pi}{2q}}\right]\;.
$$
Moreover the following computation
$$
\begin{aligned}
\tau_\Theta(s_1)
\left(\left[\expo{-\ii\frac{\pi}{2q}}\cos\theta,
\expo{-\ii\frac{\pi}{2q}}\sin\theta\right]\right)\;
&=\;
\left[\tau\left(\expo{-\ii\frac{\pi}{2q}}\cos\theta,
\expo{-\ii\frac{\pi}{2q}}\sin\theta\right), 
\expo{-\ii\frac{\pi}{2q}}\right]\\
&=\;
\left[\left(\expo{\ii\frac{\pi}{q}} \expo{-\ii\frac{\pi}{2q}}\cos\theta,
\expo{\ii\frac{\pi}{q}} \expo{-\ii\frac{\pi}{2q}}\sin\theta\right), 
\expo{-\ii\frac{\pi}{q}} \expo{+\ii\frac{\pi}{2q}}\right]\\
&=\;
\left[\left(\expo{-\ii\frac{\pi}{2q}}\cos\theta,\expo{-\ii\frac{\pi}{2q}}\sin\theta\right), \expo{+\ii\frac{\pi}{2q}}\right]\\
&=\;
s_1\left(\left[\expo{-\ii\frac{\pi}{2q}}\cos\theta,\expo{-\ii\frac{\pi}{2q}}\sin\theta\right]\right)
\end{aligned}
$$
shows that $s_1$ is also \virg{Real}.\qed
    
    \medskip
    
 Specializing the exact sequence \eqref{eq:Long_seq} to our situation one obtains
\begin{equation*}
\ldots\;H^1_{\Z_2}\big(L_{2q},{\Z}(1)\big)\;\stackrel{\imath^*}{\rightarrow}\;H^1_{\Z_2}\big(L_{2q}^\tau,{\Z}(1)\big)\;\stackrel{\delta_1}{\rightarrow}\;H^2_{\Z_2}\big(L_{2q}|L_{2q}^\tau,{\Z}(1)\big)\;\stackrel{\delta_2}{\rightarrow}\;H^2_{\Z_2}\big(L_{2q},{\Z}(1)\big)\;\stackrel{\imath^*}{\rightarrow}\;H^2_{\Z_2}\big(L_{2q}^\tau,{\Z}(1)\big)\;\ldots\;
\end{equation*}
and this reduces to     
\begin{equation*}
0\;\rightarrow\;H^1_{\Z_2}\big(L_{2q},{\Z}(1)\big)\;\stackrel{}{\rightarrow}\;H^1_{\Z_2}\big(L_{2q}^\tau,{\Z}(1)\big)\;\stackrel{\delta_1}{\rightarrow}\;H^2_{\Z_2}\big(L_{2q}|L_{2q}^\tau,{\Z}(1)\big)\;\stackrel{\delta_2}{\rightarrow}\;H^2_{\Z_2}\big(L_{2q},{\Z}(1)\big)\;\rightarrow\;0
\end{equation*}
in view of the properties of the maps $\imath^*$.  The last exact sequence is equivalent to the short  exact sequence 
\begin{equation*}
0\;\rightarrow\;{\rm Coker}^1\big(L_{2q}|L_{2q}^\tau,\Z(1)\big)\;\stackrel{\delta_1}{\rightarrow}\;H^2_{\Z_2}\big(L_{2q}|L_{2q}^\tau,{\Z}(1)\big)\;\stackrel{\delta_2}{\rightarrow}\;H^2_{\Z_2}\big(L_{2q},{\Z}(1)\big)\;\rightarrow\;0
\end{equation*}
and by inserting the numerical values one finally obtains    
\begin{equation}\label{eq:short_ex_Z8}
0\;\rightarrow\;\Z_2\;\stackrel{}{\rightarrow}\;H^2_{\Z_2}\big(L_{2q}|L_{2q}^\tau,{\Z}(1)\big)\;\stackrel{\delta_2}{\rightarrow}\;\Z_{2q}\;\rightarrow\;0\;.
\end{equation}
Since $\text{Ext}(\Z_{2q},\Z_2)\simeq\Z_2$ the  short  exact sequence \eqref{eq:short_ex_Z8} implies only two possibilities: 
\begin{itemize}
\item[(a)] $H^2_{\Z_2}(L_{2q}|L_{2q}^\tau,{\Z}(1))\simeq\Z_2\oplus\Z_{2q}$ \emph{(splitting case)};
\vspace{1.0mm}
\item[(b)] $H^2_{\Z_2}(L_{2q}|L_{2q}^\tau,{\Z}(1))\simeq\Z_{4q}$ \emph{(non-splitting case)}.
\end{itemize}
\begin{proposition}\label{prop:B_Z_8}
The short  exact sequence \eqref{eq:short_ex_Z8} is \emph{non-splitting}. Therefore it holds true that
$$
H^2_{\Z_2}\big(L_{2q}|L_{2q}^\tau,{\Z}(1)\big)\;\simeq\;\Z_{4q}\;.
$$
\end{proposition}
\proof
According to \cite[Proposition 2.7]{denittis-gomi-16} the group $H^2_{\Z_2}(L_{2q}|L_{2q}^\tau,{\Z}(1))$ classifies the pairs $(\bb{L},s)$ where $\bb{L}\to L_{2q}$ is a line bundle with $\rr{R}$-structure and $s:L_{2q}^\tau\to\bb{L}|_{\bb{L}}$ is a nowhere vanishing $\rr{R}$-section. In the proof of Lemma \ref{lemma:appB_triv_hom} we already constructed one of such  pairs $(\bb{L}_1,s_0\sqcup s_1)$.
The assignment $\bb{L}_k\simeq\bb{L}_1^{\otimes k}\to (\bb{L}_1^{\otimes k},s_0^{\otimes k}\sqcup s_1^{\otimes k})$ defines a section $\sigma:H^2_{\Z_2}(L_{2q},{\Z}(1))\to H^2_{\Z_2}(L_{2q}|L_{2q}^\tau,{\Z}(1))$. Then $\sigma(\bb{L}_1^{\otimes 2q})\in H^2_{\Z_2}(L_{2q}|L_{2q}^\tau,{\Z}(1))$ is the trivial element if and only if the  short  exact sequence \eqref{eq:short_ex_Z8} is splitting. Since $(\expo{+\ii\frac{\pi}{2q}})^{2q}=-1$  one immediately checks that
$$
\sigma(\bb{L}_1^{\otimes 2q})\;\simeq\;\left(\bb{L}_1^{\otimes 2q},s_0^{\otimes 2q}\sqcup s_1^{\otimes 2q}\right)\;\simeq\;\big(\bb{L}_0,(+1)\sqcup (-1)\big)
$$
where $\bb{L}_0\simeq L_{2q}\times\C$ is the trivial $\rr{R}$-line bundle and
$(+1):S_0\to\{+1\}$ and $(-1):S_1\to\{-1\}$ are constant  sections on $S_0$ and $S_1$. The pair $(\bb{L}_0,(+1)\sqcup (-1))$ generates the image of $\delta_1:H^1_{\Z_2}(L_{2q}^\tau,{\Z}(1))\to H^2_{\Z_2}(L_{2q}|L_{2q}^\tau,{\Z}(1))$. Hence 
$\sigma(\bb{L}_1^{\otimes 2q})$ is non-trivial and the  short  exact sequence \eqref{eq:short_ex_Z8} turns out to be non-splitting.
\qed


\subsection{\virg{Quaternionic} structures over $L_{2q}$} 

To carry out the classification of $\rr{Q}$-bundles over  
$L_{2q}$ we will use the fact that this space can be seen as the total space of an $\rr{R}$-sphere bundle over  
 $\C P^1$ whose Chern number is $2q$.  Let us start by observing that $\C P^1$ is topologically equivalent to a (Riemann) sphere
 $\n{S}^2$. Therefore $\C P^1$ can be  represented  as  the union of two two-dimensional disks $\n{D}^2$ (the two hemispheres) identified along 
the  equator. Moreover,  $\C P^1$ has the involution given by the complex conjugation.
After identifying the disk $\n{D}^2$ with the unit ball in $\C$ we can endow $\n{D}^2$ with the involution given by the complex conjugation. We will denote this involutive space with $\tilde{\n{D}}^2$. In this way we can represent the involutive space $\C P^1$ as $\C P^1\simeq \tilde{\n{D}}^2\cup \tilde{\n{D}}^2$. Let ${\n{S}}^{1,1}:=\partial\tilde{\n{D}}^2$ be the boundary of $\tilde{\n{D}}^2$ with involution given by the complex conjugation and consider the involutive space $X:=\tilde{\n{D}}^2\times{\n{S}}^{1,1}\simeq{\n{S}}^{1,1}$ where the last isomorphism is given by an equivariant homotopy. Observe that $\partial X= {\n{S}}^{1,1}\times{\n{S}}^{1,1}=:{\n{T}}^{0,2,0}$. The structure of the involutive space $(L_{2q},\tau)$ can be usefully described by means of the  equivariant clutching construction that provides
\begin{equation}\label{eq:split_L4}
L_{2q}\;\simeq\; X_1\;\cup_f\; X_2
\end{equation}
where $X_1=X_2$ are two copies of $X$ and the clutching function $f:\partial X_1\to \partial X_2$, defined by $f(z,\lambda):=(z,z^{2q}\lambda)$, takes care of the fact that 
the Chern class of $L_{2q}$ is $2q$.

\medskip

The representation \eqref{eq:split_L4} results quite useful to classify vector bundles over $L_{2q}$. Let us start with the classification of line bundles. According to the content of Remark \ref{rk:app_B} we already know that there are only $2q$ different classes of complex line bundles over $L_{2q}$ and each one of this classes admits a unique (up to isomorphisms) \virg{Real} structure. However, the result  ${\rm Pic}_{\rr{R}}(L_{2q},\tau)\simeq H^2_{\Z_2}(L_{2q},\Z(1))
\simeq\Z_{2q}$ can be re-obtained  by using the equivariant clutching construction for the representation \eqref{eq:split_L4}.
Since $H^2_{\Z_2}(X,\Z(1))\simeq H^2_{\Z_2}({\n{S}}^{1,1},\Z(1))=0$ (\cf \cite[eq. (5.17)]{denittis-gomi-14}) one concludes that $X_1$ and $X_2$ admit only trivial 
$\rr{R}$-line bundles. This observation allows to apply  the equivariant clutching construction in the form of \cite[Lemma 4.18]{denittis-gomi-14}:
\begin{equation}\label{eq:rep_pic_rel}
{\rm Pic}_{\rr{R}}\big(L_{2q},\tau\big)\;\simeq\;\big[X_1,\tilde{\n{U}}(1)\big]_{\Z_2}\; \backslash\;
\big[{\n{T}}^{0,2,0},\tilde{\n{U}}(1)\big]_{\Z_2}\;/\;
\big[X_2,\tilde{\n{U}}(1)\big]_{\Z_2}
\end{equation}
where $\tilde{\n{U}}(1)$ is the unitary group endowed by the involution given by the complex conjugation and the maps $[\psi_j]\in[X_j,\tilde{\n{U}}(1)]_{\Z_2}$, $j=1,2$
act on $[\varphi]\in [{\n{T}}^{0,2,0},\tilde{\n{U}}(1)]_{\Z_2}$ by
$$
\varphi(z,\lambda)\;\longmapsto\; \left((\psi_1|_{\partial X_1})^{-1}\cdot\varphi\cdot (f^*\psi_2|_{\partial X_2})\right)(z,\lambda)\;=\;\varphi(z,\lambda)\;\frac{\psi_2(z,z^{2q}\lambda)}{\psi_1(z,\lambda)}\;,\qquad (z,\lambda)\in{\n{T}}^{0,2,0}\;.
$$
By means of the isomorphisms (\cf \cite[eq. (5.9) and (5.17) and ]{denittis-gomi-14})
$$
\begin{aligned}
\big[X_j,\tilde{\n{U}}(1)\big]_{\Z_2}\;&\simeq\;H^1_{\Z_2}\big(X_j,\Z(1)\big)\;\simeq\;H^1_{\Z_2}\big({\n{S}}^{1,1},\Z(1)\big)\;\simeq\;\Z_2\;\oplus\;\Z\\
\big[{\n{T}}^{0,2,0},\tilde{\n{U}}(1)\big]_{\Z_2}\;&\simeq\;H^1_{\Z_2}\big({\n{T}}^{0,2,0},\Z(1)\big)\;\simeq\;\Z_2\;\oplus\;\Z^2\\
\end{aligned}
$$
one can directly compute ${\rm Pic}_{\rr{R}}(L_{2q},\tau)$ from the representation (\ref{eq:rep_pic_rel}).
Observe that bases for the summands $\Z$ are given by the projections $X_j\to {\n{S}}^{1,1}\simeq \tilde{\n{U}}(1)$ and bases for the summands $\Z_2$ 
are provided by the constant maps at $\{-1\}\in{\n{S}}^{1,1}\simeq \tilde{\n{U}}(1)$. With this information one can eventually verify that the map
$$
\big[X_1,\tilde{\n{U}}(1)\big]_{\Z_2}\; \backslash\;
\big[{\n{T}}^{0,2,0},\tilde{\n{U}}(1)\big]_{\Z_2}\;/\;
\big[X_2,\tilde{\n{U}}(1)\big]_{\Z_2}\;\longrightarrow\;\Z_{2q}
$$
induced by
$$
[\varphi]\;\longmapsto\;{\rm deg}\big(\varphi|_{\tilde{\n{S}}^1\times\{1\}}\big)\;\; \text{\rm mod.}\;\; 2q
$$
is in fact a bijection.

\medskip

Now we are in the position to classify $\rr{Q}$-bundles over $(L_{2q},\tau)$. We already know that ${\rm Vec}_{\rr{Q}}^{2}(X_j)\simeq
{\rm Vec}_{\rr{Q}}^{2}({\n{S}}^{1,1})=0$ \cite[Theorem 1.2 (ii)]{denittis-gomi-14-gen}. Hence any rank 2 $\rr{Q}$-bundles on $L_{2q}$ can be constructed by gluing the trivial product bundles on $X_j\simeq {\n{S}}^{1,1}$ along $\partial X_j\simeq {\n{T}}^{0,2,0}$. This \virg{Quaternionic} clutching
construction leads to a formula of the type \eqref{eq:rep_pic_rel}, namely
\begin{equation}\label{eq:rep_rk2_QUat}
{\rm Vec}_{\rr{Q}}^{2}\big(L_{2q},\tau\big)\;\simeq\;\big[X_1,\hat{\n{U}}(2)\big]_{\Z_2}\; \backslash\;
\big[{\n{T}}^{0,2,0},\hat{\n{U}}(2)\big]_{\Z_2}\;/\;
\big[X_2,\hat{\n{U}}(2)\big]_{\Z_2}
\end{equation}
where $\hat{\n{U}}(2)$ is the space ${\n{U}}(2)$ endowed with the involution \eqref{eq:mu-invol}
 and the maps $[\psi_j]\in[X_j,\hat{\n{U}}(2)]_{\Z_2}$, $j=1,2$
act on $[\varphi]\in [{\n{T}}^{0,2,0},\hat{\n{U}}(2)]_{\Z_2}$ by
$$
\varphi(z,\lambda)\;\longmapsto\; \left((\psi_1|_{\partial X_1})^{-1}\cdot\varphi\cdot (f^*\psi_2|_{\partial X_2})\right)(z,\lambda)\;=\;{\psi_1(z,\lambda)^{-1}}\cdot\varphi(z,\lambda)\cdot{\psi_2(z,z^{2q}\lambda)}\;,\qquad (z,\lambda)\in{\n{T}}^{0,2,0}\;.
$$
We need two technical results.

\begin{lemma}\label{lemma:tec_B_01}
The sequence of maps
$$
\big[X_j,\hat{\n{U}}(2)\big]_{\Z_2}\;\simeq\;\big[{\n{S}}^{1,1},\hat{\n{U}}(2)\big]_{\Z_2}\; \stackrel{\rm det}{\longrightarrow}\;\big[{\n{S}}^{1,1},\tilde{\n{U}}(1)\big]_{\Z_2}
\; \stackrel{\rm deg}{\longrightarrow}\;2\Z
$$
induces an isomorphism of groups.
\end{lemma}
\proof
The $\Z_2$-CW-complex structure of the 
involutive space ${\n{S}}^{1,1}$ has been described in 
\cite[Example 4.21]{denittis-gomi-14} and it is given by 
the 0-skeleton formed by two fixed 0-cells ${\boldsymbol{e}}^0_\pm:=\{\pm1,0\}\in\n{S}^1\subset\n{R}^2$ to which a free 1-cell $\tilde{\boldsymbol{e}}^1:=\{\ell_+,\ell_-\}$ is attached. Here $\ell_\pm:=\{(x_0,x_1)\in\n{S}^1\ |\ \pm x_1\geqslant0\}$ are the two (1-dimensional) hemispheres of $\n{S}^1$.
Let $\varphi:{\n{S}}^{1,1}\to \hat{\n{U}}(2)$ be an equivariant map. Since the fixed point set of $\hat{\n{U}}(2)$ coincides with $S{\n{U}}(2)$ which is path-connected, we can assume that $\varphi(\pm 1)=\n{1}_2\in S{\n{U}}(2)\subset {\n{U}}(2)$, in view of the equivariant homotopy extension property \cite{matumoto-71}.
Then the two restrictions $\varphi^\pm:=\varphi|_{\ell_\pm}$ define independently elements of $\pi_1({\n{U}}(2))$ and the equivariance condition implies
$$
\big[\varphi\big]_{\Z_2}\;\simeq\; 2\big[\varphi^+\big]\;=\; 2\big[\varphi^-\big]\;\in\;\pi_1\big({\n{U}}(2)\big)\;.
$$
Hence ${\rm deg}\circ{\rm det}(\varphi)\in 2\Z$ is always an even integer and one has the homomorphism
$$
{\rm deg}\circ{\rm det}\;:\;\big[{\n{S}}^{1,1},\hat{\n{U}}(2)\big]_{\Z_2}\;\longrightarrow\;2\Z\;.
$$
This  homomorphism is surjective. Let $k\in\Z$ and $\xi_k:{\n{S}}^1\to \n{U}(1)$ be 
defined by $\xi_k(z):=z^k$ for all $z\in {\n{S}}^1\subset\C$. It is well known that
 $[\xi_k]=k$ as element of $\pi_1(\n{U}(1))\simeq\Z$. Consider the map $\varphi_{2k}:{\n{S}}^1\to \n{U}(2)$ defined by
$$
\varphi_{2k}(z)\;:=\;
\left(\begin{array}{cc}
\xi_k(z) & 0 \\0 & \xi_k(z)
\end{array}\right)\;.
$$
This map is evidently equivariant and ${\rm deg}\circ{\rm det}(\varphi_{2k})= 2k$.
The  homomorphism is also injective. For that suppose  ${\rm deg}\circ{\rm det}(\varphi_{0})=0$ along with the condition $\varphi_0(\pm 1)=\n{1}_2$ which can be assumed without loss of generality. This immediately implies that $[\varphi^\pm_0]=0$
in $\pi_1(\n{U}(2))$, namely the two restrictions $\varphi_\pm$ are homotopy equivalent to the constant maps at $\n{1}_2$. By extending a homotopy realizing $[\varphi^+_0] = 0$ to a $\Z_2$-equivariant homotopy, one concludes that also $\varphi_0$ is equivariantly homotopy equivalent to the constant map at $\n{1}_2$.
\qed

\begin{lemma}\label{lemma:tec_B_02}
Let $\jmath_1:{\n{S}}^{1,1}\times\{1,0\}\hookrightarrow{\n{T}}^{0,2,0}$ and $\jmath_2:\{1,0\}\times{\n{S}}^{1,1}\hookrightarrow{\n{T}}^{0,2,0}$ be the natural inclusions of the first and the second component of ${\n{T}}^{0,2,0}$.
Then, the sequence of maps
$$
\big[{\n{T}}^{0,2,0},\hat{\n{U}}(2)\big]_{\Z_2}\; \stackrel{\rm det}{\longrightarrow}\;\big[{\n{T}}^{0,2,0},\tilde{\n{U}}(1)\big]_{\Z_2}\; \stackrel{(\jmath_1^*,\jmath_2^*)}{\longrightarrow}\;\big[{\n{S}}^{1,1},\tilde{\n{U}}(1)\big]_{\Z_2}\;\oplus\;\big[{\n{S}}^{1,1},\tilde{\n{U}}(1)\big]_{\Z_2}
\; \stackrel{\rm deg}{\longrightarrow}\;2\Z\;\oplus\;2\Z\;,
$$
 induces an isomorphism of groups.
\end{lemma}
\proof
The $\Z_2$-CW complex decomposition of the involutive torus ${\n{T}}^{0,2,0}={\n{S}}^{1,1}\times {\n{S}}^{1,1}$ can be derived combinatorially from  that of ${\n{S}}^{1,1}$ by exploiting the
cartesian product structure (see \cite[Example 4.22]{denittis-gomi-14}): 
\begin{itemize}
\item {\bf 0-celles.} There are four fixed 0-celles
$$
{\boldsymbol{e}}^0_{+,+}:={\boldsymbol{e}}^0_+\times{\boldsymbol{e}}^0_+\;,\qquad{\boldsymbol{e}}^0_{+,-}:={\boldsymbol{e}}^0_+\times{\boldsymbol{e}}^0_-\;,\qquad{\boldsymbol{e}}^0_{-,+}:={\boldsymbol{e}}^0_-\times{\boldsymbol{e}}^0_+\;,\qquad{\boldsymbol{e}}^0_{-,-}:={\boldsymbol{e}}^0_-\times{\boldsymbol{e}}^0_-\;;
$$
\vspace{1mm}
\item {\bf 1-celles.} There are four free 1-celles
$$
\tilde{\boldsymbol{e}}^1_{+,0}:={\boldsymbol{e}}^0_+\times\tilde{\boldsymbol{e}}^1\;,\qquad\tilde{\boldsymbol{e}}^1_{-,0}:={\boldsymbol{e}}^0_-\times\tilde{\boldsymbol{e}}^1\;,
\qquad\tilde{\boldsymbol{e}}^1_{0,+}:=\tilde{\boldsymbol{e}}^1\times{\boldsymbol{e}}^0_+\;,
\qquad\tilde{\boldsymbol{e}}^1_{0,-}:=\tilde{\boldsymbol{e}}^1\times{\boldsymbol{e}}^0_-\;;
$$
\vspace{1mm}
\item {\bf 2-celles.} There are two free 1-celles
$$
\tilde{\boldsymbol{e}}^2_{\rm even}:=\tilde{\boldsymbol{e}}^1\times\tilde{\boldsymbol{e}}^1\;,\qquad\qquad\tilde{\boldsymbol{e}}^2_{\rm odd}:=\tilde{\boldsymbol{e}}^1\times\tau\big(\tilde{\boldsymbol{e}}^1\big)\;,
$$
where $\tau(\tilde{\boldsymbol{e}}^1):=\{\ell_-,\ell_+\}$ is the image of the cell $\tilde{\boldsymbol{e}}^1$ under its involution.
\end{itemize}
Let $\varphi:{\n{T}}^{0,2,0}\to \hat{\n{U}}(2)$ be a $\Z_2$-equivariant map. Since 
the fixed-point set of $\hat{\n{U}}(2)$ is $S{\n{U}}(2)$ which is path-connected, we can assume that $\varphi$ sends all the  0-cells to $\n{1}_2$ in view of the equivariant homotopy extension property \cite{matumoto-71}. Therefore, one concludes that 
\begin{equation}\label{eq:torus_maps1}
\begin{aligned}
&\left[\jmath_1^*(\varphi)\right]_{\Z_2}\;:=\; \left[\varphi|_{\tilde{\n{S}}^1\times{\boldsymbol{e}}^0_+}\right]_{\Z_2}\;\simeq\; 2\big[\varphi^{\ell_{+},+}\big]\;=\; 2\big[\varphi^{\ell_{-},+}\big]\;\in\;\pi_1\big({\n{U}}(2)\big)\;,\\ 
&\left[\jmath_2^*(\varphi)\right]_{\Z_2}\;:=\;\left[\varphi|_{{\boldsymbol{e}}^0_+\times\tilde{\n{S}}^1}\right]_{\Z_2}\;\simeq\; 2\big[\varphi^{+,\ell_{+}}\big]\;=\; 2\big[\varphi^{+,\ell_{-}}\big]\;\in\;\pi_1\big({\n{U}}(2)\big)\;,
\end{aligned}
\end{equation}
where $\varphi^{\ell_{\pm},\pm}:=\varphi|_{\ell_\pm\times{\boldsymbol{e}}^0_\pm}$ and 
$\varphi^{\pm,\ell_{\pm}}:=\varphi|_{{\boldsymbol{e}}^0_\pm\times\ell_\pm}$ along with   ${\boldsymbol{e}}^0_\pm:=(\pm1,0)$. Since the map $\varphi$ is defined on the whole torus ${\n{T}}^{0,2,0}$, it also holds that
\begin{equation}\label{eq:torus_maps2}
\big[\varphi^{\ell_{\pm},+}\big]\;=\;\big[\varphi^{+,\ell_{\pm}}\big]\;,\qquad\quad \big[\varphi^{\ell_{\pm},-}\big]\;=\;\big[\varphi^{-,\ell_{\pm}}\big]
\end{equation}
by homotopy deformations of the restrictions along the 2-cell which is simply connected. The relations \eqref{eq:torus_maps1} say that the composition of maps
$$
{\rm deg}\circ(\jmath_1^*,\jmath_2^*)\circ{\rm det}\;:\;\big[{\n{T}}^{0,2,0},\hat{\n{U}}(2)\big]_{\Z_2}\;\longrightarrow\;2\Z\;\oplus\;2\Z
$$
is well defined. This map is injective. For that let $\varphi:{\n{T}}^{0,2,0}\to \hat{\n{U}}(2)$ be a $\Z_2$-equivariant map which sends the $0$-skeleton to $\n{1}_2$ and such that  $({\rm deg}\circ(\jmath_1^*,\jmath_2^*)\circ{\rm det}):\varphi\mapsto 0\oplus 0$. Then from \eqref{eq:torus_maps1} and \eqref{eq:torus_maps2} one concludes that $\varphi$ restricted to the $1$-skeleton is $\Z_2$-equivariantly homotopy equivalent to the constant maps at $\n{1}_2$. Therefore, by the homotopy extension property,  and without loss of generality, we can assume that $\varphi$ sends the $1$-skeleton to $\n{1}_2$ since from the beginning. This in turn implies that each of the two $2$-celles 
identifies a class in $\pi_2({\n{U}}(2))=0$. Using the homotopy which realizes the trivialization we can build a $\Z_2$-equivariant homotopy which deforms $\varphi$ to the constant map at $\n{1}_2$. The surjectivity can be proved with the help of the map
$\xi_k:{\n{S}}^1\to \n{U}(1)$ defined by $\xi_k(z):=z^k$ for $k\in\Z$ and $z\in {\n{S}}^1\subset\C$. In fact  the map $\varphi_{2k_1,2k_2}:{\n{T}}^{0,2,0}\to \hat{\n{U}}(2)$ defined by
$$
\varphi_{2k_1,2k_2}(z_1,z_2)\;:=\;
\left(\begin{array}{cc}
\xi_{k_1}(z_1)\xi_{k_2}(z_2) & 0 \\0 & \xi_{k_1}(z_1)\xi_{k_2}(z_2)
\end{array}\right)\;.
$$
is $\Z_2$-equivariant and $({\rm deg}\circ(\jmath_1^*,\jmath_2^*)\circ{\rm det}):\varphi_{2k_1,2k_2}\mapsto 2k_1\oplus 2k_2$.
\qed

\medskip

With the help of Lemma \ref{lemma:tec_B_01} and Lemma \ref{lemma:tec_B_02} we can finally complete the classification of  the $\rr{Q}$-bundles over $L_{2q}$.
\begin{proposition}\label{prop:B_Z_4}
There is a bijection 
$$
{\rm Vec}_{\rr{Q}}^{2}\big(L_{2q},\tau\big)\;\simeq\;{\rm Pic}_{\rr{R}}\big(L_{2q},\tau\big)
$$
realized by the action of ${\rm Pic}_{\rr{R}}(L_{2q},\tau)$  on ${\rm Vec}_{\rr{Q}}^{2}(L_{2q},\tau)$ through the tensor product. More precisely, let $(\bb{E}_0,\Theta_0)$ be the  trivial rank 2 $\rr{Q}$-bundle   over  $(L_{2q},\tau)$. Then the map
$$
{\rm Pic}_{\rr{R}}\big(L_{2q},\tau\big)\;\ni\;\bb{L}\;\longrightarrow\;\bb{L}\otimes\bb{E}_0\;\in\;{\rm Vec}_{\rr{Q}}^{2}\big(L_{2q},\tau\big)
$$
realizes such a bijection. In particular this implies that
$$
{\rm Vec}_{\rr{Q}}^{2m}\big(L_{2q},\tau\big)\;\simeq\;{\rm Vec}_{\rr{Q}}^{2}\big(L_{2q},\tau\big)\;\simeq\;\Z_4\;,\qquad\quad \forall\ m\in\N\;. 
$$
\end{proposition}
\proof
By looking at the explicit representatives of the homotopy classes $[X_j,\hat{\n{U}}(2)]_{\Z_2}\simeq 2\Z$
and  $[{\n{T}}^{0,2,0},\hat{\n{U}}(2)]_{\Z_2}\simeq 2\Z\oplus 2\Z$ constructed in 
Lemma \ref{lemma:tec_B_01} and Lemma \ref{lemma:tec_B_02} one immediately  concludes that the map
$$
{\rm Vec}_{\rr{Q}}^{2}\big(L_{2q},\tau\big)\;\simeq\;\big[X_1,\hat{\n{U}}(2)\big]_{\Z_2}\; \backslash\;
\big[{\n{T}}^{0,2,0},\hat{\n{U}}(2)\big]_{\Z_2}\;/\;
\big[X_2,\hat{\n{U}}(2)\big]_{\Z_2}\;\longrightarrow\; 2\Z/4q\Z\;=\; \Z_{2q}
$$ 
defined by
$$
[\varphi]\;\longmapsto\;\big({\rm deg}\circ \jmath_1^*\circ{\rm det}\big)(\varphi)\ \ \text{mod.}\ \ 4q
$$
is bijective. The $\Z_2$-equivariant map $\varphi_{2k,0}:{\n{T}}^{0,2,0}\to \hat{\n{U}}(2)$ defined by
$$
\varphi_{2k,0}(z_1,z_2)\;:=\;
\left(\begin{array}{cc}
\xi_{k}(z_1) & 0 \\0 & \xi_{k}(z_1)
\end{array}\right)\;
$$
provides a representative $\bb{E}_k$ for the class in ${\rm Vec}_{\rr{Q}}^{2}(L_{2q},\tau)$ labeled by $[2k]\in 2\Z/4q\Z$. Similarly the $\Z_2$-equivariant map $\phi_{k}:{\n{T}}^{0,2,0}\to \tilde{\n{U}}(1)$ given by $\phi_k(z)=z^k$ provides a representative $\bb{L}_k$ for the class in ${\rm Pic}_{\rr{R}}(L_{2q},\tau)$ represented by $[k]\in2\Z/4q\Z=\Z_{2q}$. The relation between equivariant maps $\varphi_{2k,0}(z_1,z_2)=\phi_{k}(z_1)\; \varphi_{0,0}(z_1,z_2)$ corresponds to the relation
$\bb{E}_k=\bb{L}_k\otimes \bb{E}_0$. Consequently the tensor product with $\bb{E}_0$ induces the bijection ${\rm Pic}_{\rr{R}}(L_{2q},\tau)\to {\rm Vec}_{\rr{Q}}^{2}(L_{2q},\tau)$ and this is compatible with the operation of tensoring $\rr{R}$-bundles. Finally the equivalence ${\rm Vec}_{\rr{Q}}^{2m}(L_{2q},\tau)\simeq {\rm Vec}_{\rr{Q}}^{2}(L_{2q},\tau)$ is a consequence of Proposition \ref{theo:stab_ran_Q_even}.
\qed

\medskip

The following result  just follows by comparing Proposition \ref{prop:B_Z_8}
with Proposition \ref{prop:B_Z_4} (1)
\begin{corollary}\label{prop:fkmm_lens}
The map
$$
\kappa\;:\;{\rm Vec}_{\rr{Q}}^{2m}\big(L_{2q},\tau\big)\;\longrightarrow\;H^2_{\Z_2}\big(L_{2q}|L_{2q}^\tau,\Z(1)\big)\;,\qquad\quad m\in\N\;.
$$
cannot be surjective.
\end{corollary}

\appendix

\section{A short reminder of the equivariant Borel cohomology}
\label{subsec:borel_cohom}

The proper cohomology theory for the analysis of vector bundle theories in the category of spaces with involution is the {equivariant cohomolgy} introduced by  A.~Borel in \cite{borel-60}. This cohomology has been used for the topological classification of \virg{Real} vector bundles \cite{denittis-gomi-14} and plays also a role in the classification of \virg{Quaternionic} vector bundles \cite{denittis-gomi-14-gen,denittis-gomi-16}. A short   self-consistent summary of this cohomology theory can be found in \cite[Section 5.1]{denittis-gomi-14} and we refer to \cite[Chapter 3]{hsiang-75} and \cite[Chapter 1]{allday-puppe-93}
for a more detailed introduction to the subject.

\medskip

Since we need this tool we briefly recall the main steps of the
\emph{Borel construction}. 
The \emph{homotopy quotient} of an involutive space   $(X,\tau)$ is the orbit space
\begin{equation}\label{eq:homot_quot}
{X}_{\sim\tau}\;:=\;X\times\ {\n{S}}^{0,\infty} /( \tau\times \theta_\infty)\;.
\end{equation}
Here $\theta_\infty$ is the {antipodal map} on the infinite sphere $\n{S}^\infty$ 
(\cf \cite[Example 4.1]{denittis-gomi-14}) and ${\n{S}}^{0,\infty}$ is used as short notation for the pair $(\n{S}^\infty,\theta_\infty)$.
The product space $X\times{\n{S}}^\infty$ (forgetting for a moment the $\Z_2$-action) has the \emph{same} homotopy type of $X$ 
since $\n{S}^\infty$ is contractible. Moreover, since $\theta_\infty$ is a free involution,  also the composed involution $\tau\times\theta_\infty$ is free, independently of $\tau$.
Let $\s{R}$ be any commutative ring (\eg, $\R,\Z,\Z_2,\ldots$). The \emph{equivariant cohomology} ring 
of $(X,\tau)$
with coefficients
in $\s{R}$ is defined as
$$
H^\bullet_{\Z_2}(X,\s{R})\;:=\; H^\bullet({X}_{\sim\tau},\s{R})\;.
$$
More precisely, each equivariant cohomology group $H^j_{\Z_2}(X,\s{R})$ is given by the
 singular cohomology group  $H^j({X}_{\sim\tau},\s{R})$ of the  homotopy quotient ${X}_{\sim\tau}$ with coefficients in $\s{R}$ and the ring structure is given, as usual, by the {cup product}.
As the coefficients of
the usual singular cohomology are generalized to \emph{local coefficients} (see \eg \cite[Section 3.H]{hatcher-02} or
\cite[Section 5]{davis-kirk-01}), the coefficients of the Borel  equivariant cohomology are also
generalized to local coefficients. Given an involutive space $(X,\tau)$ let us consider the homotopy group $\pi_1({X}_{\sim\tau})$
and the associated  \emph{group ring} $\Z[\pi_1({X}_{\sim\tau})]$. Each module $\s{Z}$ over the group $\Z[\pi_1({X}_{\sim\tau})]$ is, by definition,
a \emph{local system} on $X_{\sim\tau}$.  Using this local system one defines, as usual, the equivariant cohomology with local coefficients in $\s{Z}$:
$$
H^\bullet_{\Z_2}(X,\s{Z})\;:=\; H^\bullet({X}_{\sim\tau},\s{Z})\;.
$$
We are particularly interested in modules $\s{Z}$ whose underlying groups are identifiable with $\Z$. 
For each involutive space  $(X,\tau)$, there always exists a particular family of local systems $\Z(m)$
labelled by $m\in\Z$. Here
 $\Z(m)\simeq X\times\Z$ denotes the $\Z_2$-equivariant local system on $(X,\tau)$  made equivariant  by the $\Z_2$-action $(x,l)\mapsto(\tau(x),(-1)^ml)$.
Because the module structure depends only on the parity of $m$, we consider only the $\Z_2$-modules ${\Z}(0)$ and ${\Z}(1)$. Since ${\Z}(0)$ corresponds to the case of the trivial action of $\pi_1(X_{\sim\tau})$ on $\Z$ one has $H^k_{\Z_2}(X,\Z(0))\simeq H^k_{\Z_2}(X,\Z)$ \cite[Section 5.2]{davis-kirk-01}.

\medskip

We recall the two important group isomorphisms 
\begin{equation}\label{eq:iso:eq_cohom}
H^1_{\Z_2}\big(X,\Z(1)\big)\;\simeq\;\big[X,\n{U}(1)\big]_{\Z_2}\;,\qquad\qquad H^2_{\Z_2}\big(X,\Z(1)\big)\;\simeq\;{\rm Vec}_{\rr{R}}^1\big(X,\tau\big)\equiv {\rm Pic}_{\rr{R}}\big(X,\tau\big)\;
\end{equation}
 involving the 
first two equivariant cohomology groups. 
The first isomorphism \cite[Proposition A.2]{gomi-13} says that the first equivariant cohomology group is isomorphic to the set of $\Z_2$-equivariant homotopy classes of $\Z_2$-equivariant maps $\varphi:X\to\n{U}(1)$ where the involution on $\n{U}(1)$ is induced by the complex conjugation, \ie $\varphi(\tau(x))=\overline{\varphi(x)}$. The second isomorphism is due to B.~Kahn \cite{kahn-59} and 
expresses the equivalence between the Picard group of \virg{Real} line bundles (in the sense of \cite{atiyah-66,denittis-gomi-14}) over  $(X,\tau)$ and the second equivariant cohomology group of this space.
There are two important results that have been used in this work.
\begin{lemma}\label{appA_lemma}
Let $\Z_2$ be identified with the space $\{\pm1\}$ 
endowed with the \emph{free} flipping involution and $X$ a space endowed with any involution $\tau$. The following facts hold true for all $k\geqslant 0$:
$$
H^k_{\Z_2}(\Z_2\times X,\Z(1))\simeq H^k(X,\Z).
$$
In particular, $H^k_{\Z_2}(\Z_2, \Z(1)) = 0$ for $k\geqslant 1$.
\end{lemma}
 \proof
Let $g : X \times \Z_2 \to X \times \Z_2$ be a map given by $g(1, x) = (1, \tau(x))$ and $g(-1, x) = (-1, x)$. Then $g$ is a $\Z_2$-equivariant homeomorphism from the space $\Z_2 \times X$ with the involution $(\pm 1, x) \mapsto (\mp 1, \tau(x))$ to the space $\Z_2 \times X$ with the involution $(\pm 1, x) \mapsto (\mp 1, x)$. Hence we can assume that the involution $\tau$ on $X$ is trivial from the begining. 

One has that
\begin{equation}\label{eq:exact_appAAA1}
H^k\big(\Z_2\times X,\Z\big)\;\simeq\;H^k\big(\{-1\}\times X,\Z\big)\;\oplus\;H^k\big(\{+1\}\times X,\Z\big)\;\simeq\;H^k\big(X,\Z\big)\;\oplus\;H^k\big(X,\Z\big)
\end{equation}
for the ordinary cohomology and
\begin{equation}\label{eq:exact_appAAA2}
H^k_{\Z_2}\big(\Z_2\times X,\Z\big)\;\simeq\;H^k\big((\Z_2\times X)/\Z_2,\Z\big)\;\simeq\;H^k\big(X,\Z\big)\;
\end{equation}
for the equivariant cohomology with fixed coefficients.  
The exact sequence \cite[Proposition 2.3]{gomi-13} in this special case reads
\begin{equation}\label{eq:exact_appAAA3}
\begin{aligned}
&\longrightarrow\;H^k\big(X,\Z\big)\;\stackrel{f}{\longrightarrow}\;H^k\big(X,\Z\big)\;\oplus\;H^k\big(X,\Z\big)\;\stackrel{}{\longrightarrow}\;H^k_{\Z_2}\big(\Z_2\times X,\Z(1)\big)\;\stackrel{}{\longrightarrow}\;H^{k+1}\big(X,\Z\big)\;\stackrel{f}{\longrightarrow}\;\\
\end{aligned}
\end{equation}
where $f$, which originally is the homomorphism which forgets the $\Z_2$-action, acts as the diagonal map under the identifications \eqref{eq:exact_appAAA1} and \eqref{eq:exact_appAAA2}. Then $f$ turns out to be injective and the exact sequence \eqref{eq:exact_appAAA3} splits as
$$
\begin{aligned}
&0\longrightarrow\;H^k\big(X,\Z\big)\;\stackrel{f}{\longrightarrow}\;H^k\big(X,\Z\big)\;\oplus\;H^k\big(X,\Z\big)\;\stackrel{}{\longrightarrow}\;H^k_{\Z_2}\big(\Z_2\times X,\Z(1)\big)\;\stackrel{}{\longrightarrow}\;0\\
\end{aligned}
$$
proving the claim.\qed

\medskip

The fixed point subset $X^\tau\subset X$ is closed and $\tau$-invariant and the inclusion $\imath:X^\tau\hookrightarrow X$ extends to an inclusion $\imath:X^\tau_{\sim\tau}\hookrightarrow X_{\sim\tau}$ of the respective homotopy quotients. The \emph{relative} equivariant cohomology can be defined as usual by the identification
$$
H^\bullet_{\Z_2}\big(X|X^\tau,\s{Z}\big)\;:=\; H^\bullet\big({X}_{\sim\tau}|X^\tau_{\sim\tau},\s{Z}\big)\;.
$$
Consequently, one has the related long exact sequence in cohomology
\begin{equation}\label{eq:Long_seq}
\ldots\;\longrightarrow\;H^k_{\Z_2}\big(X|X^\tau,\s{Z}\big)\;\stackrel{\delta_2}{\longrightarrow}\;H^k_{\Z_2}\big(X,\s{Z}\big)\;\stackrel{r}{\longrightarrow}\;H^k_{\Z_2}\big(X^\tau,\s{Z}\big)\;\stackrel{\delta_1}{\longrightarrow}\;H^{k+1}_{\Z_2}\big(X|X^\tau,\s{Z}\big)\;\longrightarrow\;\ldots\;
\end{equation}
where the map $r:=\imath^*$ restricts  cochains on $X$ to  cochains on $X^\tau$. The $k$-th \emph{cokernel} of $r$ is by definition
$$
{\rm Coker}^k\big(X|X^\tau,\s{Z}\big)\;:=\;H^k_{\Z_2}\big(X^\tau,\s{Z}\big)\;/\;r\big(H^k_{\Z_2}(X,\s{Z})\big)\;.
$$
Let us point out that with the same construction 
one can define {relative} cohomology theories $H^\bullet_{\Z_2}(X|Y,\s{Z})$
for each $\tau$-invariant subset $Y\subset X^\tau$. If $Y=\empt$ then $H^k_{\Z_2}(X|\empt,\s{Z})\simeq H^k_{\Z_2}(X,\s{Z})$ by definition, hence it is reasonable to put $H^k_{\Z_2}(\empt,\s{Z})=0$ for consistency with the above long exact sequence.
The case $Y:=\{\ast\}$ of a single invariant point is important since it defines the \emph{reduced} cohomology theory  
$$
\tilde{H}^k_{\Z_2}\big(X,\s{Z}\big)\;:=\;H^k_{\Z_2}\big(X|\{\ast\},\s{Z}\big)\;.
$$
In this case, the obvious surjectivity of the map $r$ at each step of the exact sequence \eqref{eq:Long_seq} justifies the isomorphism
\begin{equation}\label{eq:app_red_cohom1}
H^k_{\Z_2}\big(X,\s{Z}\big)\;\simeq\;\tilde{H}^k_{\Z_2}\big(X,\s{Z}\big)\;\oplus {H}^k_{\Z_2}\big(\{\ast\},\s{Z}\big)
\end{equation}
When the system of coefficients is $\s{Z}=\Z(1)$ or $\s{Z}=\Z$, equation
\eqref{eq:app_red_cohom1} simplifies by using the explicit computations \cite[Section 5.1]{denittis-gomi-14}
\begin{equation}\label{eq:app_red_cohom2}
{H}^k_{\Z_2}\big(\{\ast\},\Z(1)\big)\;\simeq\;
\left\{
\begin{aligned}
&0&\qquad&\text{if}\ k=0\ \text{or}\ k \ \text{even}\\
&\Z_2&\qquad&\text{if}\ k > 0\ \text{odd}\\
\end{aligned}
\right.
\end{equation}
and
\begin{equation}\label{eq:app_red_cohom3}
{H}^k_{\Z_2}\big(\{\ast\},\Z\big)\;\simeq\;
\left\{
\begin{aligned}
&\Z&\qquad&\text{if}\ k=0\\
&\Z_2&\qquad&\text{if}\  k > 0\ \text{even}\\
&0&\qquad&\text{if}\ k\ \text{odd}\\
\end{aligned}
\right.\;.
\end{equation}
%


\medskip
\medskip

\end{document}